\documentclass[authorcolumns,numberwithinsect]{no-lipics-v2022}
\nolinenumbers

\usepackage{url}
\usepackage{graphicx,accents}
\usepackage[]{mdframed}
\usepackage{xcolor}
\usepackage{colortbl}
\definecolor{myblue}{rgb}{0.1,0.1,0.5}
\definecolor{mygreen}{rgb}{0.00,0.26,0.15}
\definecolor{iris}{rgb}{0.35, 0.31, 0.81}
\definecolor{shadecolor}{gray}{0.75}

\usepackage{xspace}
\usepackage{framed}
\usepackage{xifthen}
\usepackage{subcaption}
\usepackage{tikz}
\usetikzlibrary{calc}
\usetikzlibrary{arrows.meta}
\usepackage{amsmath}
\usepackage{amsthm}
\usepackage{thmtools,thm-restate}
\usepackage{amssymb}
\usepackage{nicefrac}
\usepackage{multirow}
\usepackage{lineno}
\usepackage{paralist}
\usepackage{todonotes}

\tikzstyle{smallvertex}=[circle, draw=black, fill=white, inner sep=2pt]

\makeatletter
\newcommand{\letter}[1]{\@alph{#1}}
\makeatother

\usepackage{comment}

\usepackage{bm}

\newcommand{\counttemphomsprob}{\#\textsc{TemporalHom}}

\newcommand{\counthoms}[2]{\#\mathsf{Hom}({#1} \to {#2})}

\newlength{\RoundedBoxWidth}
\newsavebox{\GrayRoundedBox}
\newenvironment{GrayBox}[1]%
   {\setlength{\RoundedBoxWidth}{.93\textwidth}
    \def\boxheading{#1}
    \begin{lrbox}{\GrayRoundedBox}
       \begin{minipage}{\RoundedBoxWidth}}%
   {   \end{minipage}
    \end{lrbox}
    \begin{center}
    \begin{tikzpicture}%
       \node(Text)[draw=black!20,fill=white,rounded corners,%
             inner sep=2ex,text width=\RoundedBoxWidth]%
             {\usebox{\GrayRoundedBox}};
        \coordinate(x) at (current bounding box.north west);
        \node [draw=white,rectangle,inner sep=3pt,anchor=north west,fill=white] 
        at ($(x)+(6pt,.75em)$) {\boxheading};
    \end{tikzpicture}
    \end{center}}     
\newenvironment{defproblemx}[2][]{\noindent\ignorespaces%
                                \FrameSep=6pt%
                                \parindent=0pt%
                \vspace*{-1.5em}
                \ifthenelse{\isempty{#1}}{%
                  \begin{GrayBox}{\textsc{#2}}%
                }{%
                  \begin{GrayBox}{\textsc{#2}  parameterized by~{#1}}%
                }
                \begin{tabular*}{\textwidth}{@{\hspace{.1em}} >{\itshape} p{1.8cm} p{0.8\textwidth} @{}}%
            }{
                \end{tabular*}%
                \end{GrayBox}%
                \ignorespacesafterend
            }

\newcommand{\card}[1]{\ensuremath{{\vert {#1} \vert }}}
\newcommand{\ca}[1]{\ensuremath{\mathcal{#1} }}
\newcommand{\set}[1]{\ensuremath{\left\{ {#1} \right\}}}
\newcommand{\fn}[3]{\ensuremath{{{#1} : {#2} \rightarrow {#3}}}}

\newcommand{\cO}{\mathcal{O}}

\newcommand{\ext}{\ensuremath{\mathsf{ext}}}

\newcommand{\lab}{\ensuremath{\mathtt{lab}}}

\mathchardef\mh="2D

\newcommand{\lt}{\ensuremath{\Lambda}}

\newcommand{\kw}{\ensuremath{\varrho}}

\newcommand{\lgp}{order-augmented-dual}
\newcommand{\lplus}{\ensuremath{\text{oad}}}
\newcommand{\ord}{\ensuremath{\preccurlyeq}}
 
\newcommand{\touch}{\ensuremath{\mathsf{touch}}}

\newcommand{\simsize}{\ensuremath{b}}

\newcommand{\simcw}{\ensuremath{4 \simsize^4 + 12 \simsize^3 + 14 \simsize^2 + 6 \simsize + 2}}

\newcommand{\simcwone}{\ensuremath{4 \simsize^4 + 12 \simsize^3 + 14 \simsize^2 + 6 \simsize + 1}}

\newcommand{\eqc}{\ensuremath{\#\mathsf{eqc}}}
\newcommand{\vmap}{\ensuremath{\alpha}}

\newcommand{\nnn}{\ensuremath{\mathsf{null}}}

\newcommand{\maps}{\ensuremath{\mathsf{maps}}}
\newcommand{\tms}{\ensuremath{\mathsf{times}}}

\newcommand{\bt}{\ensuremath{{\mathbf{t}}}}
\newcommand{\comhom}{\ensuremath{\mathsf{comhom}}}

\newcommand{\cwDP}{\ensuremath{\mathsf{DP}}}
\newcommand{\cwDPruntime}{\ensuremath{|P|^{\cO(1)}\cdot |\Gamma|^{\cO(\kw)}}}
\newcommand{\intro}{\ensuremath{\mathtt{intro}}}
\newcommand{\union}{\ensuremath{\mathtt{union}}}
\newcommand{\relabel}{\ensuremath{\mathtt{relabel}}}
\newcommand{\ejoin}{\ensuremath{\mathtt{e \mh join}}}
\newcommand{\ajoin}{\ensuremath{\mathtt{a \mh join}}}

\newtheorem{observation}{Observation}

\newcommand{\homs}[2]{\mathsf{Hom}(#1 \to #2)}
\newcommand{\stricthoms}[2]{\mathsf{Hom}_{<}(#1 \to #2)}

\newcommand{\W}{\mathrm{W}}

\newcommand{\R}{\mathbb{R}}

\newcommand{\fptred}{\leq^{\mathsf{FPT}}_{\mathsf{T}}}

\newcommand{\nzero}[1]{[#1)}

\newcommand{\colhom}[2]{\mathsf{ColHom}(#1 \to #2)}

\usepackage{tikz, tikz-cd}
\usetikzlibrary{matrix}
\usetikzlibrary{shapes}
\usetikzlibrary{arrows,decorations.markings, tikzmark}
\usetikzlibrary{decorations.pathreplacing,calligraphy,backgrounds}
\usetikzlibrary{arrows.meta,arrows}
\usetikzlibrary{conference}
\usetikzlibrary{cd}

\newcommand{\orderiso}{\cong_{\mathsf{o}}}

\usepackage{tikz}
\usetikzlibrary{calc}

\newcommand{\DiamondCell}[4]{%
\rule{0pt}{1.25cm}
\begin{tikzpicture}[
  scale=0.55,
  baseline=(current bounding box.center),
  vtx/.style={circle, fill=black, inner sep=1.1pt},
  lab/.style={font=\tiny, inner sep=0.5pt},
  cond/.style={font=\scriptsize, inner sep=0.5pt}
]
  \node[vtx] (L) at (-2.5,0) {};
  \node[vtx] (M) at (0,0) {};
  \node[vtx] (R) at (2.5,0) {};
  \node[vtx] (T) at (0,1.75) {};
  \node[vtx] (B) at (0,-1.75) {};

  \draw (L) -- (T) -- (R);
  \draw (L) -- (B) -- (R);
  \draw (L) -- (M) -- (R);

  \ifthenelse{\equal{#4}{same}}
{
  \node[lab, above] at (-1.0,0)
    {$e^\ell_{#1,#2}(u)$};
  \node[lab, above] at ( 1.0,0)
    {$e^r_{#1,#2}(u)$};
}
{
  \node[lab, above] at (-1.0,0)
    {$e^\ell_{#1,#2}(u)$};
  \node[lab, above] at ( 1.0,0)
    {$e^r_{#1,#2}(u')$};
}

  \node[cond, anchor=south] at (0,-2.5) {$#3$};
  \foreach \t in {0.35,0.5,0.65}
    \fill ($(M)!\t!(T)$) circle (0.6pt);
    \foreach \t in {0.35,0.5,0.65}
    \fill ($(B)!\t!(M)$) circle (0.6pt);
\end{tikzpicture}%

\rule[-1.2cm]{0pt}{1.2cm}
}

\newcommand{\PTwoCell}[4]{%
\rule{0pt}{1.25cm}
\begin{tikzpicture}[
  scale=0.55,
  baseline=(current bounding box.center),
  vtx/.style={circle, fill=black, inner sep=1.1pt},
  lab/.style={font=\small, inner sep=0.5pt},
  cond/.style={font=\scriptsize, inner sep=0.5pt}
]
  \node[vtx] (L) at (-2.5,0) {};
  \node[vtx] (M) at (0,0) {};
  \node[vtx] (R) at (2.5,0) {};

  \draw (L) -- (M) -- (R);

  \ifthenelse{\equal{#4}{same}}
{
  \node[lab, above=2pt] at (-1.0,0)
    {$e^\ell_{#1,#2}$};
  \node[lab, above=2pt] at ( 1.0,0)
    {$e^r_{#1,#2}$};
}
{
  \node[lab, above=2pt] at (-1.0,0)
    {$e^\ell_{#1,#2}$};
  \node[lab, above=2pt] at ( 1.0,0)
    {$e^r_{#1,#2}$};
}

\node[lab, below=3pt] at (-2.5,0) {$p^1_{#1,#2}$};
\node[lab, below=3pt] at (0,0) {$p^2_{#1,#2}$};
\node[lab, below=3pt] at (2.5,0) {$p^3_{#1,#2}$};

  \node[cond, anchor=south] at (0,-1) {$#3$};
\end{tikzpicture}%

\rule[-1.2cm]{0pt}{1.2cm}
}

\title{The Parameterised Complexity of Temporal Motif Counting, and a Lov{\'{a}}sz-Style Isomorphism Theorem}

\author{Jayakrishnan Madathil}{Indian Institute of Technology Palakkad}{jkmadathil@iitpkd.ac.in}{}{}

\author{Kitty Meeks}{School of Computing Science, University of Glasgow}{Kitty.Meeks@glasgow.ac.uk}{}{}

\author{Marc Roth}{School of Electronic Engineering and Computer Science, Queen Mary University of London}{m.roth@qmul.ac.uk}{}{}

\authorrunning{J. Madathil, K. Meeks, and M. Roth}

\titlerunning{Motif Counting in Temporal Graphs}

\begin{document}

\maketitle

\begin{abstract}
    We study the structural expressivity and the parameterised complexity of counting homomorphisms from small temporal patterns to large temporal graphs. Here, a temporal pattern $P$ consists of a graph together with a partial order on its edges, and a homomorphism from $P$ to a temporal graph must not only preserve edges, but also satisfy the temporal constraints imposed by the partial order of the edge set of the pattern.

    The main results of this work are three-fold:
    \begin{itemize}
        \item[(1)] We prove a temporal Lov{\'{a}}sz-style theorem, stating that two temporal graphs are isomorphic (under a natural definition of temporal isomorphisms) if and only if they have the same number of homomorphisms from all temporal patterns.
        \item[(2)] We introduce a cliquewidth-based measure on temporal patterns, called the temporally order-augmented dual width, the \emph{toadwidth} for short, and show that counting temporal homomorphisms is fixed-parameter tractable for temporal patterns of bounded toadwidth.
        \item[(3)] We provide a parameterised complexity dichotomy with an explicit tractability criterion for counting homomorphisms from totally ordered temporal patterns, classified along their underlying graph structure. 
    \end{itemize}
    The methods and tools invoked for proving (1)-(3) vary significantly: The proof of the Lov{\'{a}}sz-style Theorem is obtained by combining Lov{\'{a}}sz' original argument with an inclusion-exclusion construction to deal with temporal equalities and inequalities. The FPT algorithm in (2) is obtained by an involved dynamic programming algorithm along toadwidth. Finally, the upper bound for the dichotomy in (3) relies on the toadwidth-based algorithm in (2), and the lower bound follows from a reduction from the clique problem by embedding the temporal pattern into a grid, the connections of which are either formed by edges or by temporal constraints.
\end{abstract}

\newpage

\section{Introduction}
Temporal graphs model networks the connections of which are only available at certain points, or during certain intervals, of time. They find applications in the analysis of protein-protein interaction networks~\cite{Hanetal04}, social networks~\cite{ZhaoTHOJL10}, and in machine learning~\cite{Longaetal23,WalegaR25}, only to name a few examples; we refer the reader to the survey of Holme and Saramäki~\cite{HolmeS12} for a comprehensive exposition.

Following the success of ``Network Motifs'' in static (non-temporal) graphs~\cite{Miloetal02} recent years have seen a flurry of applied results on \emph{temporal motif counting problems}~\cite{Tempcount1,Tempcount2,Tempcount3,Tempcount4}: in a nutshell, given a pattern $P$ and a temporal graph $\Gamma$, the task is to count the number of occurrences of $P$ in $\Gamma$. Similarly as in the static case, it has been observed that frequencies of temporal patterns correlate with global features of the temporal network. Despite their relevance in practical applications witnessed by the previous works, we find a notable gap when it comes to our theoretical understanding on the inherent complexity of temporal motif counting problems: while there are results for selected patterns, such as temporal walks~\cite{EnrightMM25}, and temporal stars~\cite{Tempcount4}, the state of the art is far away from a comprehensive understanding of the complexity of counting temporal patterns in general. This stands in sharp contrast to the static, non-temporal, case, for which we know deep dichotomy results w.r.t.\ parameterised and fine-grained complexity theory that determine almost precisely the best possible running times for arbitrary motif counting problems under standard lower bound assumptions~\cite{DalmauJ04,Marx10,CurticapeanM14,CurticapeanDM17,FockeR24,DoringMW24}.

In this work, we address this gap and present the first comprehensive complexity analysis of temporal motif counting problems.

\paragraph*{The Model: Temporal Graphs and Parameterised Complexity}
The literature exhibits a variety of different, but equivalent, ways to define temporal graphs: for example, the snapshot model~\cite{Longaetal23,GaoR22} introduces a temporal graph as a sequence of static graphs on the same set of vertices, where the $i$-th element of the sequence represents the edges present at time $i$. The time-varying model~\cite{CasteigtsFQS12} equips a graph $G$ with a function $\rho: E(G) \times \mathcal{T} \to \{0,1\}$, where $\mathcal{T}$ is the set of available times, and $\rho(e,i)=1$ if and only if $e$ is available at time $i$. Note that the latter model accounts for both discrete time ($\mathcal{T}\subseteq \mathbb{N}$) and continuous time ($\mathcal{T}= \mathbb{R}_{\geq 0}$). In the present work we use the discrete time-varying model, defined as follows:
\begin{definition}[Temporal Graphs]\label{def:tempgraphs_intro}
    A temporal graph is a pair $\Gamma=(G,\tau)$ of a graph $G$ without selfloops, but with parallel edges allowed, and a mapping $\tau:E(G)\to \mathbb{N}$ satisfying that $\tau(e)\neq \tau(e')$ for each pair of distinct parallel edges $e,e' \in E(G)$. 
\end{definition}
Note that the previous definition is equivalent to the standard definition of Kempe, Kleinberg, and Kumar~\cite{KempeKK02} by enforcing $G$ to be a graph without parallel edges, but allowing $\tau$ to map edges to finite subsets of $\mathbb{N}$. The latter is also equivalent to the aforementioned time-varying model by setting $\rho(e,i)=1$ if and only if $i \in \tau(e)$. As those encodings can trivially be transformed into each other, our complexity-theoretic results are invariant on the choice; we emphasise that we  use the specific model in Definition~\ref{def:tempgraphs_intro} solely due to the fact that it will be most convenient, notation-wise, for our proofs.

To study the complexity of temporal motif counting we rely on the framework of parameterised complexity theory: Given a pattern $P$ and a temporal graph $(G,\tau)$ for which we are interested in the number of ocurences of $P$, we assume that $P$ is significantly smaller than $G$. This assumption is clearly reflected by existing practical applications of temporal motif counting, and captures, precisely, the intuition behind temporal network motifs being small and simple patterns whose frequencies in large networks correlate with global features. We also note that the same assumption is standard for the case of static (non-temporal) motif counting~\cite{Miloetal02,CurticapeanDM17}.

In the language of parameterised complexity, we say that a temporal motif counting problem is \emph{fixed-parameter tractable} (FPT) if there is a computable function $f$ such that that problem can be solved in time $f(|P|)\cdot |(G,\tau)|^{O(1)}$, where $|P|$ is the size of the pattern, and $|(G,\tau)|$ is the size of the temporal graph. 

\paragraph*{Temporal Patterns and Homomorphisms}
So far, we have not provided details on what is meant, formally, by a temporal pattern or an occurrence of a pattern in a temporal graph. Before introducing the mathematical details, let us consider a well-studied example: temporal walks and paths~\cite{KempeKK02,BentertHNN20,CasteigtsHMZ21,EnrightMM25}. A temporal walk of size $k$ in a temporal graph $(G,\tau)$ is a walk in $G$ consisting of $k$ edges $e_1,\dots,e_k$ such that $\tau(e_1)\leq \tau(e_2)\leq\dots\leq \tau(e_k)$. Similarly, a temporal path of size $k$ is a self-avoiding temporal walk of size $k$. 
A temporal walk can be equivalently expressed as a homomorphism\footnote{A homomorphism from a graph $H$ to a graph $G$ is an adjacency preserving mapping from $V(H)$ to $V(G)$; see Section~\ref{sec:prelimns} for the formal definition.} from $P_k$ to $G$ such that, for each $i \in [k-1]$ the time step of the image of the $i$'th edge is at most the time step of the image of the $(i+1)$'st edge. Moreover, a temporal path can be defined similarly via subgraph embeddings, that is, vertex-injective homomorphisms. 

In other words, the previous examples highlight that we can model temporal patterns as a combination of static patterns, such as homomorphisms and subgraph embeddings, and temporal constraints. Given that graph homomorphisms counts have shown to be, literally, the basis of static network motif counting~\cite{CurticapeanDM17}\footnote{See also~\cite[Section 4]{Roth26} for a survey on the homomorphism basis.}, we will model temporal motif counting via \emph{temporal homomorphisms} in the present paper:

\begin{definition}[Temporal Patterns and Homomorphisms]\label{def:temp_patternsandhoms_intro}
    A \emph{temporal pattern} is a triple $P=(H,R,\vartheta)$ where $H$ is a (static) graph without parallel edges, $R$ is a finite poset, and $\vartheta$ is a surjection from $E(H)$ to the groundset of $R$ such that for every distinct pair $e,e'$ of parallel edges of $H$, we have $\vartheta(e)\neq \vartheta(e')$. 

    Given a temporal graph $\Gamma=(G,\tau)$, a \emph{homomorphism} from $P$ to $\Gamma$ is a pair $\varphi=(\nu,\xi)$ such that
    \begin{enumerate}
        \item $\nu:V(H)\to V(G)$,
        \item $\xi: E(H) \to E(G)$ and for all $u,v \in V(H)$ and edge $e\in E(H)$ with endpoints $u$ and $v$, we have that $\nu(u)$ and $\nu(v)$ are the endpoints of $\xi(e)$,
        \item for all $e,e' \in E(H)$ we have that \[\vartheta(e)\leq_R\vartheta(e') \Rightarrow \tau(\xi(e))\leq \tau(\xi(e')) ~~\left(\text{implying that }\vartheta(e)=\vartheta(e') \Rightarrow \tau(\xi(e))= \tau(\xi(e'))\right) \,.\]
    \end{enumerate}
    We write $\homs{P}{\Gamma}$ for the set of all homomorphisms from $P$ to $\Gamma$
\end{definition}
We note that, in the previous definition, (1) and (2) enforce $(\nu,\xi)$ to be a graph homomorphism from $H$ to $G$, and (3) enforces that the homomorphism satisfies the temporal constraints imposed by the poset $R$. Specifically, we highlight that the condition $\vartheta(e)=\vartheta(e') \Rightarrow \tau(\xi(e))= \tau(\xi(e'))$ allows us to enforce that distinct edges of the pattern are mapped to edges of the temporal graph that are present at \emph{the same time} --- this illustrates the reason on why we do not just enforce $R$ to be a poset on $E(H)$, but rather to allow $R$ to be any finite poset the ground set of which is the image of the surjection $\vartheta$. Examples of temporal patterns are depicted in Figures~\ref{fig:temp_patterns_example} and~\ref{fig:mixed_graph_intro}.

\begin{remark}[Temporal Patterns vs.\ Temporal Graphs]
    We wish to highlight explicitly the fact that temporal graphs have edges with specific time steps while temporal patterns have relative constraints along pairs of edges via the surjection into a poset. This is due to the fact that we wish to express that edges in the image of a homomorphism satisfy temporal constraints such as ``no later than'' or ``at the same time'', rather than enforcing that the image of an edge needs to happen at a specific time; think for example of temporal paths and temporal walks which require that the times of the edges of the paths and walks are monotonically increasing.
\end{remark}

\subsection{Our Results}
This work establishes three main contributions: 
\begin{itemize}
    \item[(1)] We show that our choice for modeling temporal motif counting via temporal homomorphisms is not just some generalisation of counting temporal walks, but captures the principal expressivity desired and expected from a foundation for temporal motif counting: we establish a Lov{\'{a}}sz-style theorem stating that two temporal graphs are isomorphic (under a standard notion of temporal isomorphisms) if, and only if, they have the same number of temporal homomorphisms from all temporal patterns. In other words, the homomorphism counts from temporal patterns determine precisely the isomorphism type of a temporal graph.
    \item[(2)] We introduce a novel cliquewidth-based invariant, called the \textbf{t}emporally \textbf{o}rder-\textbf{a}ugmented \textbf{d}ual width (``\emph{toadwidth}'') on temporal patterns that does take into account not only the graph underlying the temporal pattern, but also its temporal constraints. We then provide an FPT algorithm for counting homomorphisms from temporal patterns of bounded toadwidth.
    \item[(3)] We show that our FPT result is optimal for totally ordered temporal patterns if we restrict input instances only by the underlying static graph of the temporal patterns. More specifically, we establish an explicit criterion on classes of (static) graphs $\mathcal{H}$ which, if satisfied, makes the problem of counting temporal homomorphisms from totally ordered patterns with underlying graphs in $\mathcal{H}$ fixed-parameter tractable, and which, if not satisfied, makes the problem hard for the parameterised complexity class~$\W[1]$.
\end{itemize}

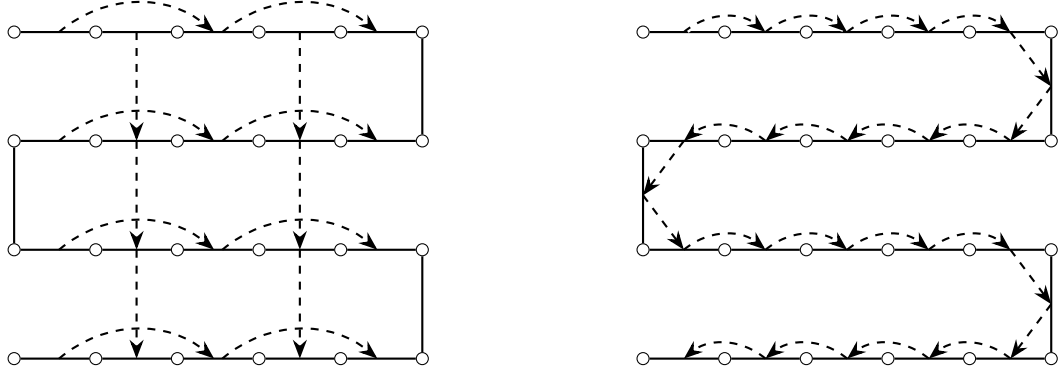
\begin{figure}[t]
    \centering
\begin{minipage}{0.48\textwidth}
\centering
  \begin{tikzpicture}[
  scale=0.9,
  vertex/.style={circle, draw, fill=white, inner sep=1.6pt},
  edge/.style={line width=0.8pt},
  darrow/.style={dashed, -{Stealth[length=2.5mm]}, line width=0.8pt}
]

\def\dx{1.2}
\def\dy{1.6}

\foreach \i in {0,...,5}
  \node[vertex] (v1-\i) at (\i*\dx,0) {};

\foreach \i in {0,...,5}
  \node[vertex] (v2-\i) at (\i*\dx,-\dy) {};

\foreach \i in {0,...,5}
  \node[vertex] (v3-\i) at (\i*\dx,-2*\dy) {};

\foreach \i in {0,...,5}
  \node[vertex] (v4-\i) at (\i*\dx,-3*\dy) {};

\foreach \i in {0,...,4}{
  \draw[edge] (v1-\i) -- (v1-\the\numexpr\i+1\relax);
  \draw[edge] (v2-\i) -- (v2-\the\numexpr\i+1\relax);
  \draw[edge] (v3-\i) -- (v3-\the\numexpr\i+1\relax);
  \draw[edge] (v4-\i) -- (v4-\the\numexpr\i+1\relax);
}

\draw[edge] (v1-5) -- (v2-5);
\draw[edge] (v2-0) -- (v3-0);
\draw[edge] (v3-5) -- (v4-5);

\foreach \r in {1,2,3,4}{
  \draw[darrow]
    ($(v\r-0)!0.55!(v\r-1)$)
    to[out=40,in=140]
    ($(v\r-2)!0.45!(v\r-3)$);

  \draw[darrow]
    ($(v\r-2)!0.55!(v\r-3)$)
    to[out=40,in=140]
    ($(v\r-4)!0.45!(v\r-5)$);
}

\foreach \r/\s in {1/2,2/3,3/4}{
  \draw[darrow]
    ($(v\r-1)!0.5!(v\r-2)$)
    -- ($(v\s-1)!0.5!(v\s-2)$);

  \draw[darrow]
    ($(v\r-3)!0.5!(v\r-4)$)
    -- ($(v\s-3)!0.5!(v\s-4)$);
}

\end{tikzpicture}
\end{minipage}
\hfill
\begin{minipage}{0.48\textwidth}
\centering
\begin{tikzpicture}[
  scale=0.9,
  vertex/.style={circle, draw, fill=white, inner sep=1.6pt},
  edge/.style={line width=0.8pt},
  darrow/.style={dashed, -{Stealth[length=2.5mm]}, line width=0.8pt}
]

\def\dx{1.2}
\def\dy{1.6}

\foreach \r in {1,2,3,4}{
  \foreach \i in {0,...,5}{
    \node[vertex] (v\r-\i) at (\i*\dx,{-(\r-1)*\dy}) {};
  }
}

\foreach \i in {0,...,4}{
  \draw[edge] (v1-\i) -- (v1-\the\numexpr\i+1\relax);
  \draw[edge] (v2-\i) -- (v2-\the\numexpr\i+1\relax);
  \draw[edge] (v3-\i) -- (v3-\the\numexpr\i+1\relax);
  \draw[edge] (v4-\i) -- (v4-\the\numexpr\i+1\relax);
}

\draw[edge] (v1-5) -- (v2-5);
\draw[edge] (v2-0) -- (v3-0);
\draw[edge] (v3-5) -- (v4-5);

\draw[darrow, draw=white]
  ($(v1-0)!0.55!(v1-1)$)
    to[out=40,in=140]
    ($(v1-2)!0.45!(v1-3)$);;

\draw[darrow] ($(v1-0)!0.5!(v1-1)$) to[out=40,in=140] ($(v1-1)!0.5!(v1-2)$);
\draw[darrow] ($(v1-1)!0.5!(v1-2)$) to[out=40,in=140] ($(v1-2)!0.5!(v1-3)$);
\draw[darrow] ($(v1-2)!0.5!(v1-3)$) to[out=40,in=140] ($(v1-3)!0.5!(v1-4)$);
\draw[darrow] ($(v1-3)!0.5!(v1-4)$) to[out=40,in=140] ($(v1-4)!0.5!(v1-5)$);
\draw[darrow] ($(v1-4)!0.5!(v1-5)$) -- ($(v1-5)!0.5!(v2-5)$);

\draw[darrow] ($(v1-5)!0.5!(v2-5)$) -- ($(v2-5)!0.5!(v2-4)$);
\draw[darrow] ($(v2-5)!0.5!(v2-4)$) to[out=140,in=40] ($(v2-4)!0.5!(v2-3)$);
\draw[darrow] ($(v2-4)!0.5!(v2-3)$) to[out=140,in=40] ($(v2-3)!0.5!(v2-2)$);
\draw[darrow] ($(v2-3)!0.5!(v2-2)$) to[out=140,in=40] ($(v2-2)!0.5!(v2-1)$);
\draw[darrow] ($(v2-2)!0.5!(v2-1)$) to[out=140,in=40] ($(v2-1)!0.5!(v2-0)$);
\draw[darrow] ($(v2-1)!0.5!(v2-0)$) -- ($(v2-0)!0.5!(v3-0)$);

\draw[darrow] ($(v2-0)!0.5!(v3-0)$) -- ($(v3-0)!0.5!(v3-1)$);
\draw[darrow] ($(v3-0)!0.5!(v3-1)$) to[out=40,in=140] ($(v3-1)!0.5!(v3-2)$);
\draw[darrow] ($(v3-1)!0.5!(v3-2)$) to[out=40,in=140] ($(v3-2)!0.5!(v3-3)$);
\draw[darrow] ($(v3-2)!0.5!(v3-3)$) to[out=40,in=140] ($(v3-3)!0.5!(v3-4)$);
\draw[darrow] ($(v3-3)!0.5!(v3-4)$) to[out=40,in=140] ($(v3-4)!0.5!(v3-5)$);
\draw[darrow] ($(v3-4)!0.5!(v3-5)$) -- ($(v3-5)!0.5!(v4-5)$);

\draw[darrow] ($(v3-5)!0.5!(v4-5)$) -- ($(v4-5)!0.5!(v4-4)$);
\draw[darrow] ($(v4-5)!0.5!(v4-4)$) to[out=140,in=40] ($(v4-4)!0.5!(v4-3)$);
\draw[darrow] ($(v4-4)!0.5!(v4-3)$) to[out=140,in=40] ($(v4-3)!0.5!(v4-2)$);
\draw[darrow] ($(v4-3)!0.5!(v4-2)$) to[out=140,in=40] ($(v4-2)!0.5!(v4-1)$);
\draw[darrow] ($(v4-2)!0.5!(v4-1)$) to[out=140,in=40] ($(v4-1)!0.5!(v4-0)$);

\end{tikzpicture}
\end{minipage}
    \caption{Illustrations of two temporal patterns, both of which have a $24$-vertex path as underlying graph, depicted with solid edges. The left pattern exhibits grid-like temporal constraints, depicted with dashed arrows, and we will show that a generalisation of such constraints to larger patterns yields intractability. The right pattern exhibits temporal constraints that follow the path, and our main algorithmic result (Theorem~\ref{thm:intro_main_algo}) implies that counting homomorphisms from such patterns is fixed-parameter tractable. Note that for both patterns, the dashed arrows only depict the Hasse diagram of the partial order; the full partial order is obtained by taking the transitive closure.}
    \label{fig:temp_patterns_example}
\end{figure}

In what follows, we present each contribution in detail.

\paragraph*{A Temporal Lov{\'{a}}sz Theorem}
One of Lov{\'{a}}sz's most famous results, often referred to as ``Lov{\'{a}}sz's Theorem'' states that the isomorphism type of a graph $G$ is determined by both infinite vectors $(\#\homs{H}{G})_{H}$ and $(\#\homs{G}{H})_{H}$, where $\#\homs{H}{G}$ denotes the number of graph homomorphisms from $H$ to $G$. For this work, the ``left-hand side'' result involving the vector $(\#\homs{H}{G})_{H}$ plays a more important role, and we can state the corresponding version of Lov{\'{a}}sz's Theorem formally as follows
\begin{theorem}[Lov{\'{a}}sz's Theorem; cf.\ Chapter 5.4 in~\cite{Lovasz12}]
    Two graphs $G_1$ and $G_2$ are isomorphic if and only if $\#\homs{H}{G_1}=\#\homs{H}{G_2}$ for each graph $H$.
\end{theorem}
The above version of Lov{\'{a}}sz's Theorem, while foundational on its own right, has witnessed a renaissance over the last decade in the context of the theory of homomorphism indistinguishability which has turned out to be a bridge between motif counting, descriptive complexity theory, graph neural networks, and isomorphism heuristics~\cite{Dvorak10,DellGR18,Morrisetal19}.

We establish a temporal version of Lov{\'{a}}sz's Theorem. For the statement however, we have to first discuss the notion of an isomorphism between temporal graphs. The literature contains various proposals, featured below. For what follows, given a temporal graph $\Gamma=(G,\tau)$, we write $G_i$ for the subgraph of $G$ only containing edges present at time $i$.
\begin{itemize}
    \item Two temporal graphs $\Gamma_1=(G,\tau_G)$ and $\Gamma_2=(H,\tau_H)$ are \emph{pointwise isomorphic}~\cite{BeddarWiesingetal24,WalegaR25} if $G_i$ and $H_i$ are isomorphic for each $i$.
    \item Let $\Gamma_1=(G,\tau_G)$ and $\Gamma_2=(H,\tau_H)$, and let $\tau_G(E(G))=\{t_1,\dots,t_n\}$ with $t_1<\dots<t_n$ and $\tau_H(E(H))=\{t'_1,\dots,t'_m\}$ with $t'_1<\dots<t'_m$. Then $\Gamma_1$ and $\Gamma_2$
    are \emph{timewise isomorphic}~\cite{WalegaR25} if $n=m$, $t_{i+1}-t_i = t'_{i+1}-t'_{i}$ and there is a mapping $\pi: V(G) \to V(H)$ such that $\pi$ is an isomorphism from $G_{t_i}$ to $H_{t'_i}$ for each $i$.
    \item The authors of~\cite{Heegetal25} introduce a third version of temporal isomorphism that interpolates between the former two and that is motivated by the preservation of temporal paths.\footnote{The details of the definition of temporal isomorphisms in~\cite{Heegetal25} are not important for the current paper.}
\end{itemize}
We claim that, in the context of the current work and temporal motif counting, the first version, i.e., pointwise isomorphism, is too weak as it allows to change the vertex identification of the isomorphism at any timestep, which will not be invariant for connected temporal patterns that enforce multiple time steps. We furthermore claim that the second version, i.e., timewise isomorphism, is too strong as it would distinguish e.g.\ the temporal graphs shown in Figure~\ref{fig:intro_isomorphisms}. Finally, the third version provided in~\cite{Heegetal25}, while enabling invariance under the number of homomorphisms from temporal walks, fails to capture global consistency of the order of the time steps. For the purpose of this work we therefore propose a natural notion of temporal isomorphisms between $\Gamma_1=(G,\tau_G)$ and $\Gamma_2=(H,\tau_H)$ as global mappings from $V(G)$ to $V(H)$ that preserve the order of the time steps, which we call \emph{order-isomorphisms}; an example is provided in Figure~\ref{fig:intro_isomorphisms}.

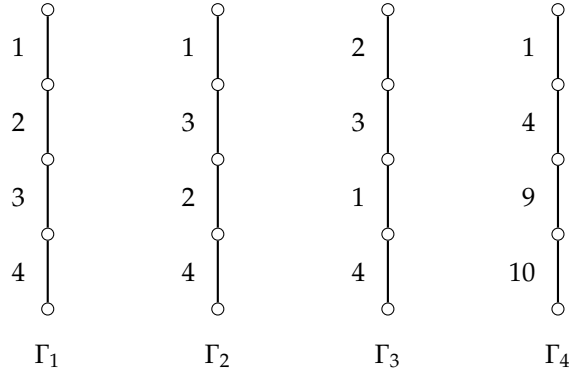
\begin{figure}[t]
    \centering
\begin{tikzpicture}[
  scale=0.9,
  vertex/.style={circle, draw, fill=white, inner sep=1.6pt},
  edge/.style={line width=0.8pt}
]

\def\dy{1.1}


\foreach \i in {0,...,4}
  \node[vertex] (a\i) at (0,-\i*\dy) {};
\foreach \i in {0,...,3}
  \draw[edge] (a\i)--(a\the\numexpr\i+1\relax);
\foreach \i/\lab in {0/1,1/2,2/3,3/4}
  \node[left=5pt] at ($(a\i)!0.5!(a\the\numexpr\i+1\relax)$) {$\lab$};

\foreach \i in {0,...,4}
  \node[vertex] (b\i) at (2.5,-\i*\dy) {};
\foreach \i in {0,...,3}
  \draw[edge] (b\i)--(b\the\numexpr\i+1\relax);
\foreach \i/\lab in {0/1,1/3,2/2,3/4}
  \node[left=5pt] at ($(b\i)!0.5!(b\the\numexpr\i+1\relax)$) {$\lab$};

\foreach \i in {0,...,4}
  \node[vertex] (c\i) at (5,-\i*\dy) {};
\foreach \i in {0,...,3}
  \draw[edge] (c\i)--(c\the\numexpr\i+1\relax);
\foreach \i/\lab in {0/2,1/3,2/1,3/4}
  \node[left=5pt] at ($(c\i)!0.5!(c\the\numexpr\i+1\relax)$) {$\lab$};

\foreach \i in {0,...,4}
  \node[vertex] (d\i) at (7.5,-\i*\dy) {};
\foreach \i in {0,...,3}
  \draw[edge] (d\i)--(d\the\numexpr\i+1\relax);
\foreach \i/\lab in {0/1,1/4,2/9,3/10}
  \node[left=5pt] at ($(d\i)!0.5!(d\the\numexpr\i+1\relax)$) {$\lab$};

\foreach \v/\lab in {a4/1,b4/2,c4/3,d4/4}
  \node[below=10pt] at (\v) {$\Gamma_{\lab}$};
\end{tikzpicture}
    \caption{Four temporal graphs; the numerical labels represent the function $\tau$, that is, the time steps of the edges. $\Gamma_1$, $\Gamma_2$, and $\Gamma_3$ are all pointwise isomorphic, but none of them is pointwise isomorphic to $\Gamma_4$. All four temporal graphs are pairwise distinct w.r.t.\ timewise isomorphism. Moreover, the pairs $\Gamma_1$ and $\Gamma_4$, as well as $\Gamma_2$ and $\Gamma_3$ are isomorphic w.r.t.\ the notion introduced in~\cite{Heegetal25}. For order-isomorphism (Definition~\ref{def:order_iso_intro}), we have $\Gamma_1 \orderiso \Gamma_4$, but no other pair is order-isomorphic, illustrating that we aim for a global notion of isomorphism that preserves the order of the time steps, but not their precise values.}
    \label{fig:intro_isomorphisms}
\end{figure}

\begin{definition}[Order-Isomorphisms between temporal graphs]\label{def:order_iso_intro}
Let $\Gamma_1=(G_1,\tau_1)$ and $\Gamma_2=(G_2,\tau_2)$ be temporal graphs. An \emph{order-isomorphism} from $\Gamma_1$ to $\Gamma_2$ is a pair $(\pi,t)$ of an isomorphism\footnote{As our temporal graphs are allowed to contain parallel edges, isomorphisms must not only specify a bijection from vertices to vertices $(\pi_\nu)$, but also a bijection from edges to edges $(\pi_\xi)$; see Section~\ref{sec:prelimns} for the formal definition of isomorphisms between static graphs with parallel edges.} $\pi=(\pi_\nu,\pi_\xi)$ from $G_1$ to $G_2$ and a bijection $t:\tau_1(E(G_1))\to \tau_2(E(G_2))$ such that
\begin{enumerate}
    \item for all $e \in E(G_1)$ we have $t(\tau_1(e)) = \tau_2(\pi_\xi(e))$, and
    \item $t$ is strongly monotonically increasing, i.e., $x<y$ implies $t(x)<t(y)$.
\end{enumerate}
We write $\Gamma_1 \orderiso \Gamma_2$ if $\Gamma_1$ and $\Gamma_2$ are order-isomorphic.
\end{definition}

To improve intuition, we provide an alternative way to define order-isomorphisms via normal forms of temporal graphs:
A temporal graph $\Gamma=(G,\tau)$ is in \emph{ordered normal form} (``ONF'') if $\tau(E(G))=[k]$ for some $k \in \mathbb{N}$, that is, the edges of $G$ contain precisely $k$ distinct time steps, and there is at least one edge per time step. This allows us to encode $\Gamma$ in a snap-shot representation $\Gamma=(G_1,\dots,G_k)$ where $V(G_i)=V(G)$ for each $i \in [k]$, and $G_i$ is, as before, the subgraph of $G$ containing only edges $e$ with $\tau(e)=i$; note that, specifically, none of the $G_i$ will have parallel edges as we enforced that parallel edges have different times in the definition of temporal graphs. We observe that each temporal graph $\Gamma=(G,\tau)$ with $\tau(E(G))=\{t_1,\dots,t_k\}$ and $t_1<\dots<t_k$ can be transformed into an ordered normal form by reassigning times $t_i \mapsto i$. In that way, it is easy to see that two temporal graphs $\Gamma=(G,\tau)$ and $\Gamma'=(G',\tau')$ with normal forms $(G_1,\dots,G_k)$ and $(G'_1,\dots,G'_k)$ are order-isomorphic if and only if there is a mapping $\pi: V(G)\to V(G')$ such that $\pi$ is an isomorphism from $G_i$ to $G'_i$ for each $i \in [k]$.

We are now able to present our first main result:
\begin{theorem}[A Temporal Lov{\'{a}}sz's Theorem]\label{thm:lovasz_improved_intro}
    Two temporal graphs $\Gamma_1$ and $\Gamma_2$ are order-isomorphic if and only if for all temporal patterns $P$ we have $\#\homs{P}{\Gamma_1}=\#\homs{P}{\Gamma_2}$.\qed
\end{theorem}

\paragraph*{An FPT Algorithm for Temporal Patterns of Bounded Toadwidth}
In the second part of this work, we develop an FPT algorithm for counting homomorphisms from temporal patterns of bounded toadwidth. We start by introducing the latter, which requires us to consider graphs containing both directed and undirected edges.

\begin{definition}[Mixed Graph]
    A \emph{mixed graph} is a graph that may contain both undirected and directed edges. For convenience, we call the undirected edges simply ''edges'' and the directed edges ''arcs.'' For a mixed graph $M$, we use $E(M)$ to denote the set of edges of $M$, and $A(M)$ to denote the set of arcs of $M$. For vertices $u, v \in V(M)$, we denote the undirected edge between $u$ and $v$ by $uv$ or $\set{u, v}$, and we denote the arc directed from $u$ to $v$ by $(u, v)$. 
\end{definition}

A \emph{labelled mixed graph} is a mixed graph $M$ equipped with a labeling function $\fn{\mu}{V(M)}{\mathbb{N}}$. For $\kw \in \mathbb{N}$, we say that a labelled mixed graph $(M, \mu)$ is $\kw$-labelled if $\mu(v) \leq \kw$ for every $v \in V(M)$. 
As a shorthand, we write $\mu(M)$ to denote the set $\mu(V(M)) = \set{\mu(v) ~|~ v \in V(M)}$. 

Let $(M_1, \mu_1)$ and $(M_2, \mu_2)$ be two labelled mixed graphs such that $V(M_1) \cap V(M_2) = \emptyset$. By the disjoint union of $(M_1, \mu_1)$ and $(M_2, \mu_2)$, we mean the labelled mixed graph $(M, \mu)$ where $V(M) = V(M_1) \cup V(M_2)$, $E(M) = E(M_1) \cup E(M_2)$ and $A(M) = A(M_1) \cup A(M_2)$, and the labeling function $\fn{\mu}{V(M)}{\mathbb{N}}$ is defined as follows: for $v \in V(M)$, we have $\mu(v) = \mu_1(v)$ if $v \in V(M_1)$, and $\mu(v) = \mu_2(v)$ if $v \in V(M_2)$. 

Next, we introduce the cliquewidth of mixed graphs; the definition follows the standard definitions w.r.t.\ both undirected and directed graphs, see e.g.\ \cite{CourcelleO00} (see also the surveys of Courcelle and Engelfriet~\cite{CourcelleE12}, and of Dabrowski, Johnson and Paulusma~\cite{DabrowskiJP19}). Since mixed graphs contain both directed and undirected edges, we introduce separate join operations for edges and arcs.
\begin{definition}[Cliquewidth of Mixed Graphs]
Consider the following operations on labelled~graphs. 
\begin{enumerate}
        \item Introduce operation: Create a vertex $v$ with label $i \in \mathbb{N}$. We denote this operation by $\intro_i(v)$. 
        
        \item Union operation: Disjoint union of two labelled mixed graphs $(M_1, \mu_1)$ and $(M_2, \mu_2)$ such that  $\mu_1(M_1) \cap \mu_2(M_2) = \emptyset$. We denote this operation by $\union((M_1, \mu_1), (M_2, \mu_2))$. 
        
        \item Relabeling operation: For a labelled mixed graph $(M, \mu)$, and for distinct $i, j \in \mathbb{N}$, change the label of all vertices with label $i$ to $j$. We denote this operation by $\relabel_{i, j}(M, \mu)$. 

        \item Edge-join operation: For a labelled mixed graph $(M, \mu)$, and for distinct $i, j \in \mathbb{N}$, add an edge between every vertex of label $i$ and every vertex of label $j$. 
        We denote this operation by $\ejoin_{i, j}(M, \mu)$. 

        \item Arc-join operation: For a labelled mixed graph $(M, \mu)$, and for distinct $i, j \in \mathbb{N}$, add an arc directed from every vertex of label $i$ to every vertex of label $j$. We denote this operation by $\ajoin_{i, j}(M, \mu)$.  
\end{enumerate}
A sequence of operations, in which each operation is one of the five operations above, is what we call a \emph{clique-expression}. For a clique-expression $\sigma$, we will use $(M^{\sigma}, \mu^{\sigma})$ to denote the labelled mixed graph constructed by $\sigma$. Let $\kw$ be a positive integer. We say that a clique-expression $\sigma$ is a \emph{$\kw$-expression} if at most $\kw$ labels are used in the construction of $(M^{\sigma}, \mu^{\sigma})$. 
Consider a mixed graph $M$. We say that $M$ admits a $\kw$-expression if there is a $\kw$-expression $\sigma$ such that $M = M^{\sigma}$; in this case, we also say that $\sigma$ is a $\kw$-expression for $M$. The \emph{cliquewidth} of $M$ is the least integer $\kw$ for which $M$ admits a $\kw$-expression. 
\end{definition}

The toadwidth of a temporal pattern will be based on the cliquewidth of a mixed graph associated with the pattern that takes into account both the underlying graph as well as the temporal constraints. \pagebreak

\begin{definition}[Order-Augmented Duals]
    Let $P=(H,R,\vartheta)$ be a temporal pattern. The \emph{order-augmented dual} (``oad'') of $P$ is a mixed graph $M$ defined as follows:
    \begin{itemize}
        \item $V(M)=E(H)$.
        \item $E(M)=\{\{e,e'\} \in V(M)^{(2)} \mid e \neq e' ~\wedge~e \cap e' \neq \emptyset\}$, that is, we add an edge between $e$ and $e'$ if $e \cap e' \neq \emptyset$.
        \item $A(M)=\{(e,e') \in V(M)^{2} \mid e \neq e' ~\wedge~\vartheta(e)\leq_R\vartheta(e')\}$, that is, we add an arc from $e$ to $e'$ if $\vartheta(e)\leq_R\vartheta(e')$. 
    \end{itemize}
\end{definition}
Consult Figure~\ref{fig:mixed_graph_intro} for an example of the construction of an order-augmented dual.

We are finally able to formally define the toadwidth.
\begin{definition}[Toadwidth]
    Let $P=(H,R,\vartheta)$ be a temporal pattern. The \textbf{t}emporally \textbf{o}rder-\textbf{a}ugmented \textbf{d}ual width (the ``\emph{toadwidth}'') of $P$ is defined as the cliquewidth of the order-augmented dual of $P$. We say that a class $\mathcal{C}$ of temporal patterns has \emph{bounded toadwidth} if there is a constant $B$ such that the toadwidth of each member of $\mathcal{C}$ is at most $B$.
\end{definition}

\begin{figure}
    \centering
    \usetikzlibrary{arrows.meta,calc}

\begin{tikzpicture}[
  scale=0.75,
  vertex/.style={circle, draw, fill=white, inner sep=1.6pt},
  edge/.style={line width=0.8pt},
  darrow/.style={dashed, -{Stealth[length=2.5mm]}, line width=0.8pt},
  darrowdashdotted/.style={dashdotted, -{Stealth[length=2.5mm]}, line width=0.8pt},
  lab/.style={font=\normalsize}
]

\node[vertex] (1) at (5,4.5) {};
\node[vertex] (2) at (3,4.5) {};
\node[vertex] (3) at (1,3.5) {};
\node[vertex] (4) at (0,2.5) {};
\node[vertex] (5) at (2,2.5) {};
\node[vertex] (6) at (0,1) {};
\node[vertex] (7) at (2,1) {};
\node[vertex] (8) at (1,0) {};
\node[vertex] (9) at (3,-1.0) {};
\node[vertex] (10) at (5,-1.0) {};

\draw[edge] (1) -- (2) -- (3) -- (4);
\draw[edge] (7) -- (8) -- (9) -- (10);
\draw[edge] (3) -- (5);
\draw[edge] (6) -- (8);

\draw[darrow] (0.5,2.9) -- (0.5,0.6);
\draw[darrow] (1.5,2.9) -- (1.5,0.6);
\draw[darrow] (4,4.4) -- (4,-0.9);

\node[lab,above=1pt]  at (4,4.5) {$d$};
\node[lab,below=1pt]  at (4,-1) {$h$};
\node[lab,above=1pt]  at (2,4) {$b$};
\node[lab,below=1pt]  at (2,-0.5) {$g$};
\node[lab,above left]  at (0.5,3) {$a$};
\node[lab,above right]  at (1.5,3) {$c$};
\node[lab,below left]  at (0.5,0.5) {$e$};
\node[lab,below right=-2pt]  at (1.5,0.5) {$f$};

\begin{scope}[xshift=9cm]

\node[vertex] (a) at (0,3) {};
\node[vertex] (b) at (1.0,4.5) {};
\node[vertex] (c) at (2,3.0) {};
\node[vertex] (d) at (4,4.5) {};
\node[vertex] (e) at (0.0,0.5) {};
\node[vertex] (f) at (2,0.5) {};
\node[vertex] (g) at (1.0,-1.0) {};
\node[vertex] (h) at (4.0,-1.0) {};

\draw[edge] (a)--(b)--(c)--(a);
\draw[edge] (b)--(d);
\draw[edge] (e)--(f)--(g)--(e);
\draw[edge] (g)--(h);

\draw[darrowdashdotted] (a) -- (e);
\draw[darrowdashdotted] (c) -- (f);
\draw[darrowdashdotted] (d) -- (h);

\node[lab,left=5pt]  at (a) {$a$};
\node[lab,above=4pt] at (b) {$b$};
\node[lab,right=5pt] at (c) {$c$};
\node[lab,above=4pt] at (d) {$d$};
\node[lab,left=5pt]  at (e) {$e$};
\node[lab,right=5pt] at (f) {$f$};
\node[lab,below=4pt] at (g) {$g$};
\node[lab,below=4pt] at (h) {$h$};

\end{scope}

\end{tikzpicture}
    \caption{\emph{(Left:)} A temporal pattern $P=(H,R,\vartheta)$. The underlying graph $H$ is depicted with solid edges, and the poset $R$ has groundset $\{a,c,d,e,f,h,\bot\}$ such that $a\leq_R e$, $c \leq_R f$, and $d \leq_R h$, depicted by dashed arrows. Moreover, $\vartheta$ is the identity on $\{a,c,d,e,f,h\}$ and maps the remaining edges to $\bot$. \emph{(Right:)} The order-augmented dual $M$ of $P$. Edges are depicted as solid lines and arcs are depicted as dash-dotted arrows. The toadwidth of $P$ is equal to the cliquewidth of $M$.}
    \label{fig:mixed_graph_intro}
\end{figure}
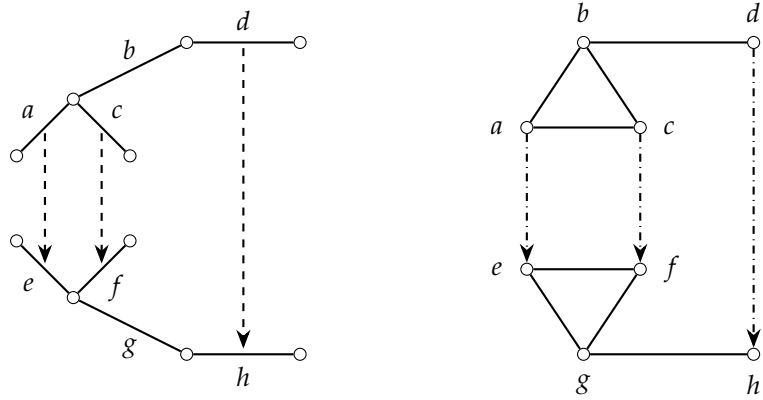

The central algorithmic result of this paper states that counting homomorphisms from temporal patterns of bounded toadwidth is fixed-parameter tractable. The restriction of the patterns is formalised by specifying the class of allowed patterns in the problem definition as follows; note that this is standard for counting problems in parameterised complexity theory~\cite{DalmauJ04,CurticapeanDM17,Roth26}.

\begin{definition}[$\#\textsc{TemporalHom}(\mathcal{C})$]
Let $\mathcal{C}$ be a class of temporal patterns. We define the following parameterised counting problem.
\begin{mdframed}
\#\textsc{TemporalHom}$(\mathcal{C})$ \\
\textbf{Input:} A temporal pattern $P\in \mathcal{C}$, and a temporal graph $(G,\tau)$\\
\textbf{Parameter:} $|P|$ \\
\textbf{Output:} $\#\homs{P}{(G,\tau)}$
\end{mdframed} 
\end{definition}

We are now able to state our main algorithmic result. Note that we restrict ourselves to temporal patterns without parallel edges for the sake of simplicity in our presentation, and we leave the case of parallel edges for future work.
\begin{restatable}{theorem}{MAINALGO}\label{thm:intro_main_algo}
    Let $\mathcal{C}$ be a class of temporal patterns without parallel edges. If $\mathcal{C}$ has bounded toadwidth, then $\#\textsc{TemporalHom}(\mathcal{C})$ is fixed-parameter tractable.\qed
\end{restatable}

\noindent \textbf{Toadwidth vs.\ Treewidth}\\
Readers familiar with parameterised counting might be aware that, in the static (non-temporal) setting, the complexity of counting homomorphisms is determined by the treewidth of the pattern~\cite{DalmauJ04,Marx10}. We argue that toadwidth is not just an arbitrary cliquewidth based width measure that happens to work for the upper bound, but that toadwidth is a very natural measure arising in the context of temporal constraints: First, it is well-known that a class of graphs has bounded treewidth if and only if the class of the associated line graphs has bounded cliquewidth~\cite{GurskiW07}. As a consequence, using the classification of Dalmau and Jonsson~\cite{DalmauJ04}, counting homomorphisms to non-temporal graphs is fixed-parameter tractable if and only if the line graphs of the allowed patterns have bounded cliquewidth; the only-if direction relies on standard lower bound assumptions in parameterised complexity theory. Now, in the temporal setting, the order augmented dual is identical to the line graph of the pattern with added arcs for the temporal constraints --- as the temporal constraints involve edges, and not vertices, it is unsurprising that we need to operate on the line graph. For this reason, we view our FPT algorithm in Theorem~\ref{thm:intro_main_algo}, which operates via dynamic programming along toadwidth, as a natural temporal counterpart of the well-established dynamic programming algorithm for counting static homomorphisms along treewidth. In particular, we believe that toadwidth will contribute to the active search for finding an appropriate temporal version of treewidth~\cite{TempTreewidth}.

\paragraph*{A Complexity Dichotomy for Totally Ordered Temporal Patterns}
In the final part of the paper, we ask whether the fixed-parameter tractability result in Theorem~\ref{thm:intro_main_algo} is optimal. While we do not answer this question in that generality, we analyse a modification of $\#\textsc{TemporalHom}(\mathcal{C})$ for which we only restrict the underlying graphs of the temporal patterns and for which we consider only total orders. For this special case, formally introduced below, we establish an exhaustive and explicit complexity dichotomy, the tractability criterion of which implies a bound on the toadwidth.

In what follows, we say that a temporal pattern $P=(H,R,\vartheta)$ is \emph{totally ordered} if $R$ is a total order and $\vartheta$ is a bijection. To avoid notational clutter, we allow ourselves to write $P=(H,\preccurlyeq)$, where $\preccurlyeq$ is a total order on $E(H)$; moreover, we write $e \prec e'$ if $e \preccurlyeq e'$ and $e \neq e'$.  
\begin{definition}
Let $\mathcal{H}$ be a class of graphs. We define the following parameterised counting problem.
\begin{mdframed}
$\#\textsc{TemporalHom}_\mathrm{TO}(\mathcal{H})$ \\
\textbf{Input:} A totally ordered temporal pattern $(H,\preccurlyeq)$ with $H \in \mathcal{H}$, and a temporal graph $(G,\tau)$\\
\textbf{Parameter:} $|(H,\preccurlyeq)|$ \\
\textbf{Output:} $\#\homs{(H,\preccurlyeq)}{(G,\tau)}$
\end{mdframed} 
\end{definition}

We are able to classify $\#\textsc{TemporalHom}_\mathrm{TO}(\mathcal{H})$ completely along the \emph{semi-induced matching number} of the line graphs of $\mathcal{H}$:
\begin{definition}[Semi-Induced Matchings]
A \emph{semi-induced matching} of size $k$ in a graph $L$ is a set of $k$ pairwise disjoint edges $\{u_1,v_1\},\dots,\{u_k,v_k\}$ such that there are no edges in $L$ between $u_i$ and $v_j$ for $i \neq j$.

The \emph{semi-induced matching number} of $L$ is the maximum size of a semi-induced matching in $L$, and we say that a class of graphs has \emph{bounded semi-induced matching number} if there is a constant $B$ such that each member of the class has semi-induced matching number at most $B$.
\end{definition}

The third main result of this work reads as follows.
\begin{restatable}{theorem}{CLASSIFICATION}\label{thm:intro_main_classification}
    Let $\mathcal{H}$ be a recursively enumerable\footnote{We note that the restriction to recursively enumerable classes is standard in parameterised complexity theory as, without this assumption, the parameter dependency in the running time of the reductions could not be bounded by a computable function.} class of graphs. Assuming that $\mathrm{FPT}\neq \mathrm{W}[1]$, the following are equivalent:
    \begin{itemize}
        \item $\#\textsc{TemporalHom}_\mathrm{TO}(\mathcal{H})$ is fixed-parameter tractable.
        \item the class of line graphs of $\mathcal{H}$ has bounded semi-induced matching number.\qed
    \end{itemize}
\end{restatable}

The assumption $\mathrm{FPT}\neq \mathrm{W}[1]$ is the standard hardness assumption in parameterised complexity theory (see Section~\ref{sec:prelimns} for a concise introduction). While it is unknown whether $\mathrm{FPT}= \mathrm{W}[1]$ implies $\mathrm{P}=\mathrm{NP}$, it is known that $\mathrm{FPT}= \mathrm{W}[1]$ refutes the Exponential Time Hypothesis~\cite{ETH,Chenetal05,Chenetal06}, leading to subexponential time algorithms for $3$-$\textsc{SAT}$.

Finally, we note that the tractability part of Theorem~\ref{thm:intro_main_classification} relies on our algorithm from Theorem~\ref{thm:intro_main_algo}, and connects toadwidth with semi-induced matching number, via the following bound, which might be of independent interest.

\begin{restatable}{lemma}{SIMBOUNDSTOADWIDTH}\label{lem:sim-cw-bound}
    Let $\simsize$ be a positive integer. 
    Let $H$ be a simple, undirected  graph. If every semi-induced matching in $L(H)$ has size at most $\simsize$, then for every total order $\ord$ on $E(H)$, the toadwidth of $(H,\ord)$ is at most $\simcw$.
\end{restatable}


\section{Preliminaries}\label{sec:prelimns}
Given a positive integer $n$, we set $[n]:=\{1,\dots,n\}$, and we set $\nzero{n}:=\{0,\dots,n-1\}$. For a finite set $S$, we write $|S|$ or $\#S$ do denote the cardinality of $S$. For a set $S$, we write $S^2$ to denote the set of all ordered pairs of elements of $S$, i.e., $S^2 = \set{(x, y) ~|~ x, y \in S}$; and we write $S^{(2)}$ to denote the the set of all unordered pairs of elements of $S$, i.e., $S^{(2)} = \set{\set{x, y} ~|~ x, y \in S}$. We use $\uplus$ to denote the union of two disjoint sets; that is, for disjoint sets $A$ and $B$, we write $A \uplus B$ to denote $A \cup B$. 

\subsection{Graph Theory}
Graphs in this work are undirected, and do not contain self-loops unless stated otherwise, but we allow multi-edges for technical reasons --- among others, this is to enable the presence of the same edge at multiple times in a temporal graph. For an undirected graph $H$, we use $V(H)$ and $E(H)$ to denote the vertex set and the edge set of $H$, respectively. 
A \emph{homomorphism} from a graph $H$ to a graph $G$ is a pair $\varphi=(\nu,\xi)$ where $\nu: V(H) \to V(G)$ and $\xi: E(H) \to E(G)$ such that for each $u,v \in V(H)$ and for each edge $e$ with endpoints $u$ and $v$, we have that $\nu(u)$ and $\nu(v)$ are the endpoints of $\xi(e)$.\footnote{Note that, for simple graphs without multi-edges, a homomorphism is defined as just an edge-preserving mapping $\nu:V(H) \to V(G)$, i.e., $\{u,v\}\in E(H)$ implies $\{\nu(u),\nu(v)\}\in E(G)$; this fully specifies the mapping $\xi$ from $E(H)$ to $E(G)$. However, for graphs with multi-edges, we need to specify $\xi$ explicitly.} A homomorphism $\varphi=(\nu,\xi)$ is called an \emph{embedding} if $\nu$ and $\xi$ are injective, and an embedding is an \emph{isomorphism} if $(\nu^{-1},\xi^{-1})$ is an embedding from $G$ to $H$.

The \emph{line graph} of a graph $G$ contains has vertex set $E(G)$, that is, every edge of $G$ becomes a vertex, and two distinct vertices $e,e' \in E(G)$ are made adjacent if $e \cap e'\neq \emptyset$.

\subsection{Temporal Graphs, Patterns, and Homomorphisms}

\begin{definition}[Temporal Graphs]
    A temporal graph is a pair $\Gamma=(G,\tau)$ of a graph $G$ without selfloops, but with parallel edges allowed, and a mapping $\tau: E(G)\to \mathbb{N}$ such that for every pair of parallel edges $e \neq e'$ we have $\tau(e)\neq \tau(e')$.
\end{definition}

As discussed in the introduction we will focus on the following type of isomorphisms between temporal graphs, called order-isomorphisms.

\begin{definition}[Order-Isomorphisms]
Let $\Gamma_1=(G_1,\tau_1)$ and $\Gamma_2=(G_2,\tau_2)$ be temporal graphs. An order-isomorphism from $\Gamma_1$ to $\Gamma_2$ is a pair $(\pi,t)$ of an isomorphism $\pi=(\pi_\nu,\pi_\xi)$ from $G_1$ to $G_2$ and a bijection $t:\tau_1(E(G_1))\to \tau_2(E(G_2))$ such that
\begin{enumerate}
    \item for all $e \in E(G_1)$ we have $t(\tau_1(e)) = \tau_2(\pi_\xi(e))$, and
    \item $t$ is strongly monotonically increasing, i.e., $x<y$ implies $t(x)<t(y)$.
\end{enumerate}
We write $\Gamma_1 \orderiso \Gamma_2$ if $\Gamma_1$ and $\Gamma_2$ are order-isomorphic.
\end{definition}

\begin{definition}[Temporal Pattern]
    A temporal pattern is a triple $P=(H,R,\vartheta)$ where $H$ is a graph, $R$ is a finite poset, and $\vartheta$ is a surjection from $E(H)\to R$ such that for ever pair $e,e'$ of parallel edges of $H$ we have $\vartheta(e)\neq \vartheta(e')$.
\end{definition}

To avoid notational clutter, given a temporal pattern $P=(H,R,\varphi)$, we might interchangeably use $R$ to denote the poset as well as its groundset. Moreover, we use the symbols $<_R$ and $\leq_R$ for the relation of $R$, that is, we write $a\leq_R b$ if $(a,b)$ is a tuple of the partial ordering relation, and we write $a <_R b$ if $a \leq_R b$ and $a \neq b$.

\begin{definition}[Homomorphisms from Temporal Patterns]
    Let $P=(H,R,\vartheta)$ be a temporal pattern, and let $\Gamma=(G,\tau)$ be a temporal graph. A homomorphism from $P$ to $\Gamma$ is a homomorphism $\varphi=(\nu,\xi)$ from $H$ to $G$ such that for all $e,e' \in E(H)$ we have
    $\vartheta(e)\leq_R \vartheta(e') \Rightarrow \tau(\xi(e)) \leq \tau(\xi(e'))$ (implying that 
    $\vartheta(e)= \vartheta(e') \Rightarrow \tau(\xi(e)) = \tau(\xi(e'))$). We write $\#\homs{P}{\Gamma}$ for the set of all homomorphisms from~$P$ to~$\Gamma$.
\end{definition}

\paragraph*{Temporal Patterns with Total Orders}
For our lower bounds, we will mainly rely on \emph{totally ordered} temporal patterns, that is, $P=(H,R,\vartheta)$ such that $R$ is a total order and $\vartheta$ is a bijection. To avoid notational clutter, we allow ourselves to write $P=(H,\preccurlyeq)$, where $\preccurlyeq$ is a total order on $E(H)$; moreover, we write $e \prec e'$ if $e \preccurlyeq e'$ and $e \neq e'$. 

A homomorphism from a totally ordered temporal pattern $(H,\preccurlyeq)$ to $(G,\tau)$ is then a homomorphism $\varphi=(\nu,\xi)$ from $H$ to $G$ such that, for all $e,e' \in E(H)$ we have $e \prec e'$ if and only if $\tau(\xi(e))<\tau(\xi(e'))$.

\subsection{Parameterised Complexity Theory}
We provide a brief introduction to parameterised complexity theory and refer the reader to the standard textbook of Flum and Grohe~\cite{FlumG06} for a detailed exposition (cf.\ \cite[Chapter 14]{FlumG06} for a treatment on parameterised counting).

A \emph{parameterised counting problem} is a pair $(P,\kappa)$ of a function $P:\{0,1\}^\ast \to \mathbb{Q}$ and a polynomial-time computable\footnote{We note that, in some literature, the requirement of the parameterisation to be computable is omitted; for our paper, this distinction does not make a difference.} \emph{parameterisation} $\kappa: \{0,1\}^\ast \to \mathbb{N}$. Parameterised decision problems are defined similarly, with the only exception that the co-domain of $P$ is $\{0,1\}$. In what follows, we refer to both parameterised decision and counting problems as \emph{parameterised problems}.

A parameterised problem $(P,\kappa)$ is \emph{fixed-parameter tractable} (``FPT'') if there is a computable function $f$ and an algorithm $\mathbb{A}$ such that $\mathbb{A}$, on input $x$, computes $P(x)$ in time
\[f(\kappa(x))\cdot |x|^{O(1)}\,.\]
We call $\mathbb{A}$ an \emph{FPT algorithm} for $(P,\kappa)$.

A \emph{parameterised Turing-reduction} from $(P_1,\kappa_1)$ to $(P_2,\kappa_2)$ is an FPT algorithm $\mathbb{A}$ for $(P_1,\kappa_1)$ which is given access to an oracle for $P_2$ and which satisfies that, on input $x$, each oracle query $y$ satisfies $\kappa_2(y)\leq g(\kappa_1(x))$, where $g$ is a computable function independent of $x$. In other words, the parameter of each oracle call depends only on the parameter of the input. We write $(P_1,\kappa_1)\fptred (P_2,\kappa_2)$ if a parameterised Turing-reduction exists. 

A parameterised problem $(P,\kappa)$ is $\mathrm{W}[1]$-\emph{hard} if it reduces from the parameterised clique problem $\textsc{Clique}$ (defined below), that is $\textsc{Clique}\fptred (P,\kappa)$.

\paragraph*{Further Parameterised Problem Definitions}
\begin{mdframed}
\textsc{Clique} \\
\textbf{Input:} A graph $G$ and a positive integer $k$ \\
\textbf{Parameter:} $k$ \\
\textbf{Output:} $1$ is $G$ contains a complete $k$-vertex subgraph, and $0$ otherwise.
\end{mdframed} 

For the following problem definition, $\mathcal{C}$ denotes a class of temporal patterns.

\begin{mdframed}
\#\textsc{TemporalHom}$(\mathcal{C})$ \\
\textbf{Input:} A temporal pattern $P\in \mathcal{C}$, and a temporal graph $(G,\tau)$\\
\textbf{Parameter:} $|P|$ \\
\textbf{Output:} $\#\homs{P}{(G,\tau)}$
\end{mdframed} 

For the following problem definitions, we do not restrict the entire temporal pattern, but only the underlying graph. Moreover, we also restrict to totally ordered ``$\mathrm{TO}$'' patterns. We write $\mathcal{H}$ to denote a recursively enumerable class of graphs. \pagebreak

\begin{mdframed}
$\#\textsc{TemporalHom}_\mathrm{TO}(\mathcal{H})$ \\
\textbf{Input:} A totally ordered temporal pattern $(H,\preccurlyeq)$ with $H \in \mathcal{H}$, and a temporal graph $(G,\tau)$\\
\textbf{Parameter:} $|(H,\preccurlyeq)|$ \\
\textbf{Output:} $\#\homs{(H,\preccurlyeq)}{(G,\tau)}$
\end{mdframed} 

\begin{mdframed}
$\#\textsc{ColTemporalHom}_\mathrm{TO}(\mathcal{H})$ \\
\textbf{Input:} A totally ordered temporal pattern $(H,\preccurlyeq)$ with $H \in \mathcal{H}$, and an $H$-coloured temporal graph $(G,\tau,c)$\\
\textbf{Parameter:} $|(H,\preccurlyeq)|$ \\
\textbf{Output:} $\#\colhom{(H,\preccurlyeq)}{(G,\tau,c)}$
\end{mdframed}

\section{Order-Isomorphisms and a Temporal Lov{\'{a}}sz Theorem}

Our goal is to prove the following classification of temporal isomorphisms via homomorphism indistinguishability.

\begin{theorem}[Theorem~\ref{thm:lovasz_improved_intro}; restated]\label{thm:lovasz_improved}
    Two temporal graphs $\Gamma_1$ and $\Gamma_2$ are order-isomorphic if and only if for all temporal patterns $P$ we have $\#\homs{P}{\Gamma_1}=\#\homs{P}{\Gamma_2}$.\qed
\end{theorem}

We start with the easy direction.

\begin{lemma}\label{lem:iso_easy_direction}
    Let $\Gamma_1=(G_1,\tau_1)$ and $\Gamma_2=(G_2,\tau_2)$ and let $P=(H,R,\vartheta)$ be a temporal pattern. If $\Gamma_1 \orderiso \Gamma_2$ then $\#\homs{P}{\Gamma_1}=\#\homs{P}{\Gamma_2}$.
\end{lemma}
\begin{proof}
    Let $(\pi,t)$ with $\pi=(\pi_\nu,\pi_\xi)$ be an order-isomorphism from $\Gamma_1$ to $\Gamma_2$. There is a canonical mapping $b$ from $\homs{P}{\Gamma_1}$ to $\homs{P}{\Gamma_2}$: Given $\varphi=(\hat{\nu},\hat{\xi})\in \homs{P}{\Gamma_1}$, set $b(\varphi):=(\nu',\xi')$ where $\nu'(v)=\pi_\nu(\hat{\nu}(v))$ and $\xi'(e)=\pi_\xi(\hat{\xi}(e))$, that is, we obtain the homomorphism $b(\varphi)$ from $P$ to $\Gamma_2$ by composing $\varphi$ with the order-isomorphism. This shows immediately that $b(\varphi)$ is a homomorphism from $H$ to $G_2$. Now let $e,e' \in E(H)$ such that $\vartheta(e)\leq_R \vartheta(e')$. Since $\varphi$ is a homomorphism from $P$ to $\Gamma_1$, we have that $\tau_1(\hat{\xi}(e))\leq \tau_1(\hat{\xi}(e'))$ and hence $t(\tau_1(\hat{\xi}(e)))\leq t(\tau_1(\hat{\xi}(e')))$. Thus
    \[\tau_2(\xi'(e)) = \tau_2(\pi_\xi(\hat{\xi}(e))) = t(\tau_1(\hat{\xi}(e)))\leq t(\tau_1(\hat{\xi}(e'))) = \tau_2(\pi_\xi(\hat{\xi}(e))) = \tau_2(\xi'(e'))\,.\]
    Next let $e,e' \in E(H)$ such that $\vartheta(e) = \vartheta(e')$. Since $\varphi$ is a homomorphism from $P$ to $\Gamma_1$, we have that $\tau_1(\hat{\xi}(e))= \tau_1(\hat{\xi}(e'))$ and hence $t(\tau_1(\hat{\xi}(e)))= t(\tau_1(\hat{\xi}(e')))$. Thus
    \[\tau_2(\xi'(e)) = \tau_2(\pi_\xi(\hat{\xi}(e))) = t(\tau_1(\hat{\xi}(e)))= t(\tau_1(\hat{\xi}(e'))) = \tau_2(\pi_\xi(\hat{\xi}(e))) = \tau_2(\xi'(e'))\,.\]
    Hence $b(\varphi)$ is indeed a homomorphism from $P$ to $\Gamma_2$. 

    Symmetrically, we define $b'$ from $\homs{P}{\Gamma_2}$ to $\homs{P}{\Gamma_1}$, using the inverse of $(\pi,t)$. It is straightforward to check that $b\circ b'$ and $b' \circ b$ are the identity functions on the respective homomorphism sets. Hence $b$ is a bijection and the claim follows.
\end{proof}

For the second, more interesting direction, we consider \emph{strict} homomorphisms in an intermediate step.

\begin{definition}[Strict Homomorphisms from Temporal Patterns]
    Let $P=(H,R,\vartheta)$ be a temporal pattern, and let $\Gamma=(G,\tau)$ be a temporal graph. A homomorphism $\varphi=(\nu,\xi)$ from $P$ to $\Gamma$ is called strict if
    $\vartheta(e)<_R \vartheta(e') \Rightarrow \tau(\xi(e)) < \tau(\xi(e'))$. We write $\#\stricthoms{P}{\Gamma}$ for the set of all strict homomorphisms from $P$ to~$\Gamma$.
\end{definition}

\begin{lemma}\label{lem:strict_to_nostrict}
    Let $\Gamma_1=(G_1,\tau_1)$ and $\Gamma_2=(G_2,\tau_2)$ and assume that, for all temporal patterns $P$, we have that $\#\homs{P}{\Gamma_1}=\#\homs{P}{\Gamma_2}$. Then $\#\stricthoms{P}{\Gamma_1}=\#\stricthoms{P}{\Gamma_2}$ for all temporal patterns $P$.
\end{lemma}
\begin{proof}
    We prove this claim by expressing the number of strict homomorphisms as a linear combination of (not necessarily strict) homomorphisms via inclusion-exclusion over the violated strict inequality constraints.
    Let $P=(H,R,\vartheta)$ be a temporal pattern, and let $\mathbb{H}$ be the Hasse diagram of $R$. Recall that, formally, $\mathbb{H}$ contains all strict inequalities $(a < b)$ for direct successors in the partial order $R$. Let $J_P$ be the set of all pairs of edges $e,e'$ of $H$ such that $(\vartheta(e) <\vartheta(e')) \in \mathbb{H}$.

    Observe that, for any temporal graph $\Gamma=(G,\tau)$, we have
    \begin{align}\#\stricthoms{P}{\Gamma} &= \#\homs{P}{\Gamma} - \#\{(\nu,\xi) \in\homs{P}{\Gamma} \mid \exists (e,e')\in J_P: \tau(\xi(e)) = \tau(\xi(e')) \} \label{eq:strict-nostrict1}\\
    ~&= \sum_{A \subseteq J_P} (-1)^{|A|}\cdot \#\{(\nu,\xi) \in\homs{P}{\Gamma} \mid \forall (e,e')\in A: \tau(\xi(e)) = \tau(\xi(e')) \}  \label{eq:strict-nostrict2}
    \end{align}
    Note that it suffices to consider $J_P$ in (\ref{eq:strict-nostrict1}): Given a pair $e,e'$ of edges of $H$ with $a:=\vartheta(e)<\vartheta(e')=:b$ such that $a < b$ is not an element of the Hasse diagram, any homomorphism $\varphi=(\nu,\xi)$ with $\tau(\xi(e))=\tau(\xi(e'))$ would map all edges $\hat{e}$ with $a<\vartheta(\hat{e})<b$ also to an edge with time $\tau(\xi(e))$. Since $\vartheta$ is surjective, this implies that, in particular, there would be an edge $f$ with $(\vartheta(e)<\vartheta(f))\in \mathbb{H}$ and $\tau(\xi(e))=\tau(\xi(f))$.

    The transformation (\ref{eq:strict-nostrict2}) is then just the inclusion-exclusion principle.

    Now, for $A \subseteq J_P$ we define the temporal pattern $P_A=(H,R_A,\vartheta_A)$ obtained from $P$ by identifying, in the poset $R$, the elements $\vartheta(e)$ and $\vartheta(e')$ for all $(e,e')\in A$. Observe that 
    \[\#\{(\nu,\xi) \in\homs{P}{\Gamma} \mid \forall (e,e')\in A: \tau(\xi(e)) = \tau(\xi(e')) \} = \#\homs{P_A}{\Gamma} \,.\]
    Consequently, we have
    \[\#\stricthoms{P}{\Gamma} = \sum_{A \subseteq J_P} (-1)^{|A|}\cdot \#\homs{P_A}{\Gamma}\,.\]
    In particular, for all temporal patterns $P$ we have
    \[\#\stricthoms{P}{\Gamma_1}=  \sum_{A \subseteq J_P} (-1)^{|A|}\cdot \#\homs{P_A}{\Gamma_1} \stackrel{(\ast)}{=}  \sum_{A \subseteq J_P} (-1)^{|A|}\cdot \#\homs{P_A}{\Gamma_2} =\#\stricthoms{P}{\Gamma_2}\,,\]
    where $(\ast)$ is the premise of the lemma. This concludes the proof.
\end{proof}

Next, we need to introduce quotients of temporal graphs. To this end we first recall the definition of quotients of static graphs with parallel edges: given $H$ and a partition $\rho$ of $V(H)$, the quotient $H/\rho$ contains one vertex for each block $B\in \rho$. Furthermore, for each edge $\{u,v\}\in E(H)$, we create an edge $\{B_u,B_v\}$ in $H/\rho$, where $B_u$ and $B_v$ denote, respectively, the blocks of $\rho$ that contain $u$ and $v$. Note that $B_u=B_v$ is possible - in that case, we create a self-loop.\footnote{While we technically allow self-loops for quotients, we note that quotients with self-loops will not influence our proofs as temporal graphs do not have self-loops.}

\begin{definition}[Quotients of Temporal Patterns]
    Let $P=(H,R,\vartheta)$ be a temporal pattern, and let $\rho$ be a partition of $V(H)$. The quotient $P/\rho$ is a temporal pattern defined has follows
    \begin{enumerate}
        \item We construct the quotient graph $H/\rho$, but we do not (yet) delete newly created parallel edges. 
        \item While there is a pair of parallel edges $e,e'$ of $H/\rho$ such that $\vartheta(e)=\vartheta(e')$, we delete one of $e$ or $e'$. The resulting graph is denoted by $(H/\rho)_\downarrow$.
    \end{enumerate}
    Then we set $P/\rho := ((H/\rho)_\downarrow,R,\vartheta|_{E((H/\rho)_\downarrow)})$.
\end{definition} \pagebreak

Note that, in the definition above, the image of the restriction of $\vartheta$ to the edges of $(H/\rho)_\downarrow$ is still $R$ (i.e., it is still surjective), since we did not delete any edge $e$ with a unique $\vartheta(e)$.

\begin{lemma}\label{lem:temp_hom_basis}
    Let $P=(H,R,\vartheta)$ be a temporal pattern, let $\Gamma$ be a temporal graph, and write $\#\mathsf{Inj}_{<}(P\to \Gamma)$ for the set of all strict homomorphisms $(\nu,\xi)\in \stricthoms{P}{\Gamma}$ such that $\nu$ is injective. Then
    \[ \#\mathsf{Inj}_{<}(P \to \Gamma) = \sum_\rho \mu(\rho) \cdot \#\stricthoms{P/\rho}{\Gamma} \,,\]
    where the sum is over all partitions of $V(H)$ and $\mu$ is the M\"obius function of the partition lattice of $V(H)$.
\end{lemma}
\begin{proof}
    The proof follows almost verbatim the classical proof of this transformation on graphs due to Lov{\'{a}}sz~\cite[Chapter 5.2.3]{Lovasz12}. The only subtlety we need to address is that, in the definition of quotients of temporal patterns, we had to delete (one of each pair of) parallel edges $e,e'$ with identical temporal constraints $\vartheta(e)=\vartheta(e')$ since those parallel edges would not be admissible in our definition of temporal patterns.
    However, since the temporal graph $\Gamma=(G,\tau)$ does not have parallel edges $e,e'$ with $\tau(e)=\tau(e')$ by definition, even if we kept the parallel edges of the pattern, it would not change the number of homomorphisms.
\end{proof}

\begin{lemma}\label{lem:total_injectivity}
    Let $P=(H,R,\vartheta)$ be a temporal pattern, and let $\Gamma=(G,\tau)$ be a temporal graph. Let $\varphi=(\nu,\xi)\in \stricthoms{P}{\Gamma}$. If $R$ is a total order and if $\nu$ is injective, then $\xi$ is injective as well. 
\end{lemma}
\begin{proof}
    If $\nu$ is injective, then the only way for $\xi$ to be not injective is the existence of two parallel edges $e$ and $e'$ with $\xi(e)=\xi(e')$. By definition of temporal patterns, we have $\vartheta(e)\neq \vartheta(e')$ since $e$ and $e'$ are parallel edges. Since $R$ is a total order, we can assume w.l.o.g.\ that $\vartheta(e)< \vartheta(e')$, which implies $\tau(\xi(e))<\tau(\xi(e'))$ and thus $\xi(e)\neq \xi(e')$ which yields a contradiction.
\end{proof}


For the next step, we will consider a normal form of temporal graphs under order isomorphisms.
\begin{definition}[Ordered Normal Forms (ONF)]
    A temporal graph $\Gamma=(G,\tau)$ is in \emph{ordered normal form} if $\tau(E(G))=[k]$ for some positive integer $k$; we call $\Gamma$ an \emph{ONF}.
    We will encode ONFs as a tuple of simple graphs without multi-edges $\Gamma=(G_1,\dots,G_k)$ where, for all $i \in [k]$ we have $V(G_i)=V(G)$ and 
    \[ E(G_i):= \{e \in E(G) \mid \tau(e)=i\}\,.\]
\end{definition}

Note that each temporal graph $(G,\tau)$ is order-isomorphic to a an ONF $\Gamma=(G_1,\dots,G_k)$ where $k=|\tau(E(G))|$ and $G=(V(\Gamma),\bigcup_{i \in[k]}E(G_i))$. Moreover, it is easy to verify the following property.

\begin{fact}\label{fact:compressed_isomorphism}
    Let $\Gamma^1=(G^1_1,\dots,G^1_k)$ and $\Gamma^2=(G^2_1,\dots,G^2_\ell)$ be ONFs. Then $\Gamma^1$ and $\Gamma^2$ are order-isomorphic if and only if $k=\ell$ and there is a bijection $\beta:V(\Gamma^1)\to V(\Gamma^2)$ such that, for each $i \in [k]$, $\beta$ is a graph isomorphism from $G^1_i$ to $G^2_i$.
\end{fact}


The final objects we need to introduce in this section are temporal patterns associated with ONFs.

\begin{definition}
    Let $\Gamma=(G_1,\dots,G_k)$ be an ONF. We define the temporal pattern $P[\Gamma]=(H,\R,\vartheta)$ as follows:
    \begin{enumerate}
        \item $H=(V(\Gamma),\bigcup_{i\in[k]}E(G_i))$.
        \item $R=([k],\leq)$, that is, $R$ is just the total order on the first $k$ natural numbers.
        \item $\forall i \in[k], \forall e \in E_i: \vartheta(e)=i$.
    \end{enumerate}
\end{definition}

\begin{lemma}\label{lem:inj_to_iso}
    Let $\Gamma^1=(G^1_1,\dots,G^1_k)$ and $\Gamma^2=(G^2_1,\dots,G^2_\ell)$ be ONFs such that $\#\mathsf{Inj}_{<}(P[\Gamma^1]\to \Gamma_2)>0$ and $\#\mathsf{Inj}_{<}(P[\Gamma^2]\to \Gamma_1)>0$. Then $\Gamma^1 \orderiso \Gamma^2$.
\end{lemma}
\begin{proof}
    Since the injective homomorphisms are strict, we obtain immediately that $k=\ell$. By Fact~\ref{fact:compressed_isomorphism} it suffices to construct a bijection $\beta:V(\Gamma^1)\to V(\Gamma^2)$ such that, for each $i \in [k]$, $\beta$ is a graph isomorphism from $G^1_i$ to $G^2_i$.

    Let $(\nu,\xi)\in  \mathsf{Inj}_{<}(P[\Gamma^1]\to \Gamma_2)>0$ and, for $i\in\{1,2\}$ write $(H^i,R^i,\vartheta^i)$ for $P[\Gamma^i]$. We claim that $\nu$ can be chosen as the desired function $\beta$. 
    
    Since $(\nu,\xi)$ is strict, and since $R^1=([k],\leq)$ and $k=\ell$, we have that $\xi(\vartheta^{-1}(i))\subseteq E(G_i)$ . Moreover, $\nu$ is injective by definition of $\mathsf{Inj}$, and since $R^1$ is a total order, we get by Lemma~\ref{lem:total_injectivity} that $\xi$ is injective as well. Consequently, $\nu$ induces a subgraph embedding from $H[E(G^1_i)]=G^1_i$ to $G^2_i$ for all $i\in [k]$. In particular, this implies that $G^1_i$ is a subgraph of $G^2_i$.
    With a symmetric argument, we obtain that $G^2_i$ is a subgraph of $G^1_i$.

    However, in combination, this implies that, for each $i\in[k]$, the function $\nu$ is a subgraph embedding from $G^1_i$ to $G^2_i$ \emph{and} $G^2_i$ is a subgraph of $G^1_i$. Hence $\nu$ must be an isomorphism, which concludes the proof.
\end{proof}

We are now able to prove the backwards direction of our isomorphism theorem.

\begin{lemma}\label{lem:iso_hard_direction}
     Two temporal graphs $\Gamma^1$ and $\Gamma^2$ are order-isomorphic if for all temporal patterns $P$ we have $\#\homs{P}{\Gamma^1}=\#\homs{P}{\Gamma^2}$.
\end{lemma}
\begin{proof}
    We can assume w.l.o.g.\ that $\Gamma^1$ and $\Gamma^2$ are ONFs as each temporal graph has an order-isomorphic ONF and the number of homomorphisms does not change under order-isomorphic temporal patterns by Lemma~\ref{lem:iso_easy_direction}.

    By Lemma~\ref{lem:inj_to_iso} it suffices to show that $\#\mathsf{Inj}_{<}(P[\Gamma^1]\to \Gamma_2)>0$ and $\#\mathsf{Inj}_{<}(P[\Gamma^2]\to \Gamma_1)>0$. We only show the first inequality, since the second one is fully symmetric. First, note that the identity mapping induces a strict injective homomorphism from $P[\Gamma^1]$ to $\Gamma^1$; therefore
    \begin{equation}\label{eq:iso_hard_helper}
        \mathsf{Inj}_{<}(P[\Gamma^1] \to \Gamma^1)\neq \emptyset 
    \end{equation}
    We hence have
    \begin{align}
         \#\mathsf{Inj}_{<}(P[\Gamma^1] \to \Gamma^2) &= \sum_\rho \mu(\rho) \cdot \#\stricthoms{P[\Gamma^1]/\rho}{\Gamma^2} \tag{Lemma~\ref{lem:temp_hom_basis}}\\
         ~&= \sum_\rho \mu(\rho) \cdot \#\stricthoms{P[\Gamma^1]/\rho}{\Gamma^1} \tag{Lemma~\ref{lem:strict_to_nostrict} \& premise of this lemma} \\
         ~&= \#\mathsf{Inj}_{<}(P[\Gamma^1] \to \Gamma^1) > 0 \tag{Lemma~\ref{lem:temp_hom_basis} \& (\ref{eq:iso_hard_helper})}\,.
    \end{align}
\end{proof}

\begin{proof}[Proof of Theorem~\ref{thm:lovasz_improved}]
    The if-direction if is proved in Lemma~\ref{lem:iso_hard_direction}, and the only-if-direction is proved in Lemma~\ref{lem:iso_easy_direction}.
\end{proof}

\section{An FPT Algorithm for Temporal Patterns of Bounded Toadwidth}

In this section, we prove our main algorithmic result, i.e., Theorem~\ref{thm:intro_main_algo}, which says that $\counttemphomsprob{(\ca{C})}$ is fixed-parameter tractable if $\ca{C}$ has bounded toadwidth. To do this, we design an involved dynamic programming algorithm along toadwidth.

The remainder of this section is dedicated to designing this algorithm, and we assume throughout this part that temporal patterns do not have parallel edges. We begin by introducing notation and proving a number of preparatory results. 

\subsection{Further Notation and Terminology}
 For $E' \subseteq E(H)$, we use $V(E')$ to denote the set of vertices of $H$ that are the endpoints of the edges in $E'$; that is, $V(E') = \set{v \in V(H) ~|~ \text{ there exists } e \in E' \text{ such that } e \text{ is incident with } v}$. For $E' \subseteq E(H)$, the subgraph of $H$ induced by $E'$, denoted by $H[E']$, is the subgraph of $H$ whose vertex set is $V(E')$ and edge set is $E'$. For $V' \subseteq V(H)$, we use $E(V')$ to denote the set of edges that are incident with a vertex in $V'$; that is, $E(V') = \set{e \in E(H) ~|~ \text{ there exist } v \in V' \text{ and } u \in V(H) \text{ such that } e = uv}$. For a graph $H$, we say that $X \subseteq V(H)$ is an \emph{independent vertex cover of $H$} if $X$ is an independent set and a vertex cover of $H$. 

\paragraph*{Viable graphs}

 We say that a simple, undirected graph $H$ is a \emph{star} if $H$ has no isolated vertices, and there exists a vertex $w \in V(H)$ such that every edge in $H$ is incident with $w$; we call such a vertex $w$ a \emph{centre} of $H$. Notice  that if $H$ is a star with at least two edges, then $H$ has a unique centre, and otherwise $H$ consists of a single edge, say edge $e = w_1 w_2$, and both $w_1$ and $w_2$ are centres of $H$. 

Let $H$ be a simple, undirected graph.  
 We say that $H$ is  \emph{viable} if $H$ has no isolated vertices, and $H$ has an independent vertex cover of size at most $2$. We say that $H$ is \emph{non-viable} if $H$ has no isolated vertices and $H$ is not viable. 
We now show that a viable graph $H$ has a unique independent vertex cover of size at most $2$, unless $H$ has at most two edges, or $H$ is a path or a cycle with exactly $4$ vertices, or $H$ has a connected component consisting of a single edge. See Figures~\ref{fig:viable-unique-vc} and~\ref{fig:viable-non-unique-vc}. 
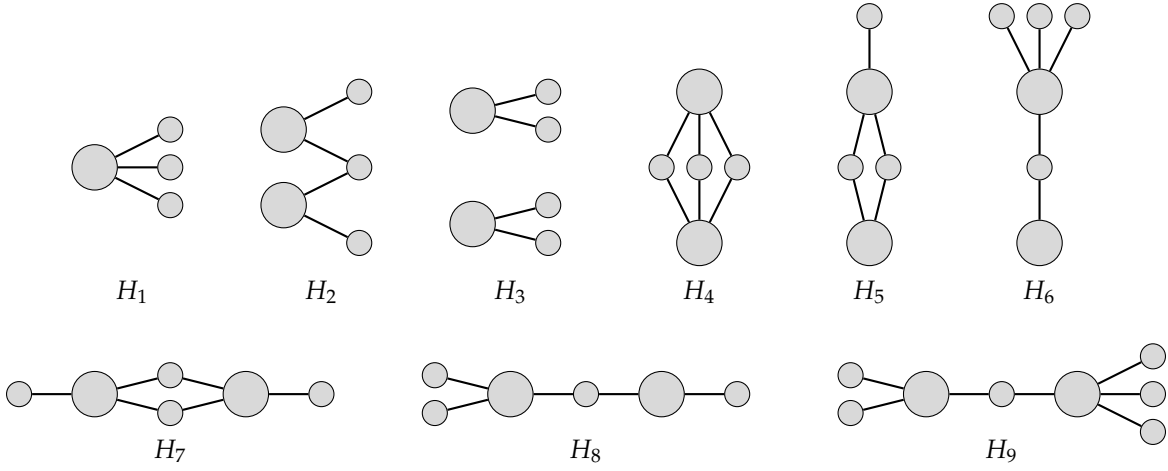
\begin{figure}
    \centering
    \begin{tikzpicture}[scale=0.5, every node/.style={circle, minimum size=2.5mm}, every edge quotes/.style = {auto, font=\footnotesize, sloped}]

\node[minimum size=6mm, draw,fill=gray!30] (u) at (0.0, 9.0) {};
\node[draw, fill=gray!30] (x2) at (2.0, 10.0) {};
\node[draw, fill=gray!30] (x3) at (2.0, 9.0) {};
\node[draw, fill=gray!30] (x4) at (2.0, 8.0) {};

\draw[thick] (u) to (x4);
\draw[thick] (u) to (x3);
\draw[thick] (u) to (x2);

\node[] (H1) at (1.0, 5.7) {$H_1$};

\node[minimum size=6mm, draw, fill=gray!30] (u) at (5.0, 10.0) {};
\node[draw, fill=gray!30] (x1) at (7.0, 11.0) {};
\node[draw, fill=gray!30] (x2) at (7.0, 9.0) {};
\node[minimum size=6mm, draw, fill=gray!30] (v) at (5.0, 8.0) {};
\node[draw, fill=gray!30] (x3) at (7.0, 7.0) {};

\draw[thick] (v) to (x3);
\draw[thick] (v) to (x2);
\draw[thick] (u) to (x2);
\draw[thick] (u) to (x1);

\node[] (H2) at (6.0, 5.7) {$H_2$};

\node[minimum size=6mm, draw, fill=gray!30] (u) at (10.0, 10.5) {};
\node[draw, fill=gray!30] (x1) at (12.0, 11.0) {};
\node[draw, fill=gray!30] (x3) at (12.0, 10.0) {};

\draw[thick] (u) to (x3);
\draw[thick] (u) to (x1);

\node[minimum size=6mm, draw, fill=gray!30] (v) at (10.0, 7.5) {};
\node[draw, fill=gray!30] (y1) at (12.0, 8.0) {};
\node[draw, fill=gray!30] (y3) at (12.0, 7.0) {};

\draw[thick] (v) to (y3);
\draw[thick] (v) to (y1);

\node[] (H3) at (11.0, 5.7) {$H_3$};


\node[minimum size=6mm, draw, fill=gray!30] (u) at (16.0, 11.0) {};

\node[draw, fill=gray!30] (p1) at (15.0, 9.0) {};
\node[draw, fill=gray!30] (p2) at (16.0, 9.0) {};
\node[draw, fill=gray!30] (p3) at (17.0, 9.0) {};

\node[minimum size=6mm, draw, fill=gray!30] (v) at (16.0, 7.0) {};

\draw[thick] (u) to (p1);
\draw[thick] (u) to (p2);
\draw[thick] (u) to (p3);

\draw[thick] (v) to (p1);
\draw[thick] (v) to (p2);
\draw[thick] (v) to (p3);

\node[] (H4) at (16.0, 5.7) {$H_4$};


\node[draw, fill=gray!30] (u1) at (20.5, 13.0) {};
\node[minimum size=6mm, draw, fill=gray!30] (u2) at (20.5, 11.0) {};

\node[draw, fill=gray!30] (p1) at (20.0, 9.0) {};
\node[draw, fill=gray!30] (p2) at (21.0, 9.0) {};

\node[minimum size=6mm, draw, fill=gray!30] (v) at (20.5, 7.0) {};

\draw[thick] (u1) to (u2);
\draw[thick] (u2) to (p1);
\draw[thick] (u2) to (p2);
\draw[thick] (v) to (p1);
\draw[thick] (v) to (p2);

\node[] (H5) at (20.5, 5.7) {$H_5$};


\node[draw, fill=gray!30] (x1) at (24.0, 13.0) {};
\node[draw, fill=gray!30] (x2) at (25.0, 13.0) {};
\node[draw, fill=gray!30] (x3) at (26.0, 13.0) {};
\node[minimum size=6mm, draw, fill=gray!30] (u) at (25.0, 11.0) {};

\node[draw, fill=gray!30] (p1) at (25.0, 9.0) {};

\node[minimum size=6mm, draw, fill=gray!30] (v) at (25.0, 7.0) {};

\draw[thick] (u) to (x1);
\draw[thick] (u) to (x2);
\draw[thick] (u) to (x3);
\draw[thick] (u) to (p1);
\draw[thick] (v) to (p1);

\node[] (H6) at (25.0, 5.7) {$H_6$};


\node[draw, fill=gray!30] (u1) at (-2.0, 3.0) {};
\node[minimum size=6mm, draw, fill=gray!30] (u2) at (0.0, 3.0) {};

\node[draw, fill=gray!30] (p1) at (2.0, 3.5) {};
\node[draw, fill=gray!30] (p2) at (2.0, 2.5) {};

\node[minimum size=6mm, draw, fill=gray!30] (v1) at (4.0, 3.0) {};
\node[draw, fill=gray!30] (v2) at (6.0, 3.0) {};

\draw[thick] (u1) to (u2);
\draw[thick] (u2) to (p1);
\draw[thick] (u2) to (p2);
\draw[thick] (v1) to (p1);
\draw[thick] (v1) to (p2);
\draw[thick] (v1) to (v2);

\node[] (H7) at (2.0, 1.5) {$H_7$};


\node[minimum size=6mm, draw, fill=gray!30] (u) at (11.0, 3.0) {};
\node[draw, fill=gray!30] (x1) at (9.0, 3.5) {};
\node[draw, fill=gray!30] (x2) at (9.0, 2.5) {};

\node[draw, fill=gray!30] (p1) at (13.0, 3.0) {};
\node[minimum size=6mm, draw, fill=gray!30] (v1) at (15.0, 3.0) {};
\node[draw, fill=gray!30] (v2) at (17.0, 3.0) {};

\draw[thick] (u) to (x2);
\draw[thick] (u) to (x1);

\draw[thick] (u) to (p1);
\draw[thick] (v1) to (p1);
\draw[thick] (v1) to (v2);

\node[] (H8) at (13.0, 1.5) {$H_8$};


\node[minimum size=6mm, draw, fill=gray!30] (u) at (22.0, 3.0) {};
\node[draw, fill=gray!30] (x1) at (20.0, 3.5) {};
\node[draw, fill=gray!30] (x2) at (20.0, 2.5) {};

\node[draw, fill=gray!30] (p1) at (24.0, 3.0) {};
\node[minimum size=6mm, draw, fill=gray!30] (v) at (26.0, 3.0) {};
\node[draw, fill=gray!30] (y1) at (28.0, 4.0) {};
\node[draw, fill=gray!30] (y2) at (28.0, 3.0) {};
\node[draw, fill=gray!30] (y3) at (28.0, 2.0) {};

\draw[thick] (u) to (x2);
\draw[thick] (u) to (x1);

\draw[thick] (u) to (p1);
\draw[thick] (v) to (p1);

\draw[thick] (v) to (y1);
\draw[thick] (v) to (y2);
\draw[thick] (v) to (y3);
\node[] (H9) at (24.0, 1.5) {$H_9$};
\end{tikzpicture}
    \caption{Examples of viable graphs with a unique independent vertex cover of size at most 2. In graphs $H_2, H_3,\ldots, H_9$, the two ``big''vertices constitute the unique independent vertex cover of size at most 2. The graph $H_1$ is a star and has an independent vertex cover of size 1, which is also its unique independent vertex cover of size at most 2.}
    \label{fig:viable-unique-vc}
\end{figure}

\begin{lemma}\label{lem:viable-ind-vc}
    Let $H$ be a viable graph. Then $H$ is covered by one of the following four cases: 
    \begin{description} 
        \item[(Viable-1).] $H$ has at most 2 edges; 
        \item[(Viable-2).] $H$ is a path or a cycle on exactly $4$ vertices; 
        \item[(Viable-3).] $H$ has a connected component with a single edge; 
        \item[(Viable-4).] $H$ has a unique independent vertex cover of size at most 2. 
    \end{description}
\end{lemma}
\begin{proof}
     First, since $H$ is viable, $H$ has no isolated vertices, and therefore $H$ contains at least one edge, and consequently, every vertex cover of $H$ must have size at least $1$. Suppose first that $H$ has a vertex cover of size exactly $1$, i.e., $H$ is a star. Now, if $H$ has at most $2$ edges, then $H$ is covered by case Viable-1. Otherwise, $H$ is a star with at least $3$ edges, and therefore $H$ has a unique centre, say $w \in V(H)$; then $\set{w}$ is the unique independent vertex cover of size at most 2. 
     
     So assume from now on that $H$ has no vertex cover of size exactly 1. Since $H$ is viable, $H$ has an independent vertex cover of size at most 2; let $\set{w_1, w_2} \subseteq V(H)$ be such a vertex cover. As $H$ has no isolated vertices, every vertex of $V(H) \setminus \set{w_1, w_2}$ is adjacent to at least one (and possibly both) of $w_1$ and $w_2$. Notice also that as $V(H) \setminus \set{w_1, w_2}$ is an independent set, for every vertex $x \in V(H) \setminus \set{w_1, w_2}$, we have $N(x) \subseteq \set{w_1, w_2}$; in particular, we have $N(x) = \set{w_1}$ or $N(x) = \set{w_2}$ or $N(x) = \set{w_1, w_2}$. For $i \in \set{1, 2}$, let $N_i$ be the set of all vertices $x$ in $V(H) \setminus \set{w_1, w_2}$ with $N(x) = \set{w_i}$, and  let $N_{12}$ be the set of all vertices in $V(H) \setminus \set{w_1, w_2}$ with $N(x) = \set{w_1, w_2}$. 
     We now split the analysis into cases depending on which of the three sets $N_1, N_2$ and $N_{12}$ are empty. Notice that by definition, the sets $N_1, N_2$ and $N_{12}$ are pairwise-disjoint. 

     Suppose first that $N_{12} = \emptyset$. Notice then that $H$ has exactly two connected components---one containing $w_1$ and the other containing $w_2$, and both components are stars. If $\card{N(w_1)} = 1$, then the component containing $w_1$ is simply an edge, and in this case $H_1$ is covered by case Viable-3. Identical reasoning applies if $\card{N(w_2)} = 1$. So assume that $\card{N(w_1)} > 1$ and $\card{N(w_2)} > 1$. Then every vertex cover of size at most $2$ must contain both $w_1$ and $w_2$, and we can again conclude that $H$ has a unique independent vertex cover of size at most $2$. 

     Suppose now that $N_{12} \neq \emptyset$. Now consider the following three sub-cases. 
     \begin{itemize}
         \item Both $N_{1}$ and $N_{2}$ are empty. Now if $\card{N_{12}} = 1$, say $N_{12} = \set{x}$, then $H$ has exactly two edges, $w_1 x$ and $w_2 x$, and so $H$ is covered by case Viable-1. If $\card{N_{12}} = 2$, say $N_{12} = \set{x, y}$, then $E(H) = \set{w_1 x, w_2 x, w_1 y, w_2 y}$; thus $H$ is a cycle on 4 vertices, and hence $H$ is covered by case Viable-2. If $\card{N_{12}} \geq 3$, then both $w_1$ and $w_2$ have degree at least 3, and consequently every vertex cover of size at most 2 must contain both $w_1$ and $w_2$. So we can again conclude that $\set{w_1, w_2}$ is the unique independent vertex cover of size at most $2$. 
         \item Both $N_{1}$ and $N_{2}$ are non-empty; fix $x_1 \in N_1$, $x_2 \in N_2$ and $y \in N_{12}$. Notice then  that $x_1 w_1 y w_2 x_2$ is a path on 5 vertices, and in this case, every vertex cover of size at most 2 must contain both $w_1$ and $w_2$. We can therefore conclude that $\set{w_1, w_2}$ is the unique independent vertex cover of size at most 2. 
         \item Exactly one of $N_{1}$ and $N_{2}$ is non-empty. Assume without loss of generality that $N_1 \neq \emptyset$ and $N_2 = \emptyset$; the other case is symmetric. Recall that we are under the assumption that $\card{N_{12}} \neq \emptyset$;  fix $y \in N_{12}$.   
         \begin{itemize}
              \item Suppose that either $\card{N_1} \geq 2$ or $\card{N_{12}} \geq 2$.  Then the vertex $w_1$ has degree at least 3 (notice that $N(w_1) = N_1 \cup N_{12}$), and therefore $w_1$ must belong to every vertex cover of size at most $2$. Hence $y$ cannot belong to any independent vertex cover of size at most 2 (as $w_1 y \in E(G)$), which implies that $w_2$ must belong to every independent vertex cover of size at most $2$ (to cover the edge $y w_2$). We have thus shown that every independent vertex cover of size at most $2$  must contain both $w_1$ and $w_2$, and hence $\set{w_1, w_2}$ is the unique independent vertex cover of size at most 2.  
              \item The only remaning case is when $\card{N_1} = 1$ and $\card{N_{12}} = 1$, say $N_1 = \set{x}$ (and $N_{12} = \set{y}$). Then   $E(H) = \set{x w_1, w_1 y, y w_2}$; thus $H$ is a path on 4 vertices, and hence covered by case Viable-2. \qedhere
         \end{itemize}
     \end{itemize}
\end{proof}

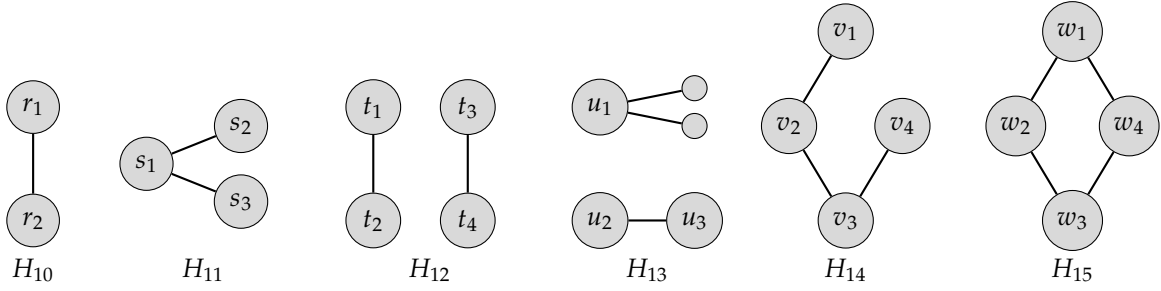
\begin{figure}
    \centering
    \begin{tikzpicture}[scale=0.5, every node/.style={circle, minimum size=2.5mm}, every edge quotes/.style = {auto, font=\footnotesize, sloped}]

\node[minimum size=6mm, draw, fill=gray!30] (r1) at (-10.5, 10.0) {$r_1$};
\node[minimum size=6mm, draw, fill=gray!30] (r2) at (-10.5, 7.0) {$r_2$};
\draw[thick] (r1) to (r2);
\node[] (H10) at (-10.5, 5.7) {$H_{10}$};
\node[minimum size=6mm, draw, fill=gray!30] (s1) at (-7.5, 8.5) {$s_1$};
\node[minimum size=6mm, draw, fill=gray!30] (s2) at (-5.0, 9.5) {$s_2$};
\node[minimum size=6mm, draw, fill=gray!30] (s3) at (-5.0, 7.5) {$s_3$};

\draw[thick] (s1) to (s2);
\draw[thick] (s1) to (s3);

\node[] (H12) at (-6.0, 5.7) {$H_{11}$};

\node[minimum size=6mm, draw, fill=gray!30] (t1) at (-1.5, 10.0) {$t_1$};
\node[minimum size=6mm, draw, fill=gray!30] (t3) at (1.0, 10.0) {$t_3$};
\node[minimum size=6mm, draw, fill=gray!30] (t2) at (-1.5, 7.0) {$t_2$};
\node[minimum size=6mm, draw, fill=gray!30] (t4) at (1.0, 7.0) {$t_4$};

\draw[thick] (t1) to (t2);
\draw[thick] (t3) to (t4);

\node[] (H12) at (0.0, 5.7) {$H_{12}$};
\node[minimum size=6mm, draw, fill=gray!30] (u) at (4.5, 10.0) {$u_1$};
\node[draw, fill=gray!30] (x2) at (7.0, 10.5) {};
\node[draw, fill=gray!30] (x3) at (7.0, 9.5) {};

\draw[thick] (u) to (x3);
\draw[thick] (u) to (x2);

\node[minimum size=6mm, draw, fill=gray!30] (v) at (4.5, 7.0) {$u_2$};
\node[minimum size=6mm, draw, fill=gray!30] (y1) at (7.0, 7.0) {$u_3$};

\draw[thick] (v) to (y1);

\node[] (H13) at (5.75, 5.7) {$H_{13}$};

\node[minimum size=6mm, draw, fill=gray!30] (v1) at (11.0, 12.0) {$v_1$};

\node[minimum size=6mm, draw, fill=gray!30] (v2) at (9.5, 9.5) {$v_2$};
\node[minimum size=6mm, draw, fill=gray!30] (v4) at (12.5, 9.5) {$v_4$};

\node[minimum size=6mm, draw, fill=gray!30] (v3) at (11.0, 7.0) {$v_3$};

\draw[thick] (v1) to (v2);
\draw[thick] (v2) to (v3);
\draw[thick] (v3) to (v4);

\node[] (H14) at (11, 5.7) {$H_{14}$};


\node[minimum size=6mm, draw, fill=gray!30] (u) at (17.0, 12.0) {$w_1$};

\node[minimum size=6mm, draw, fill=gray!30] (p1) at (15.5, 9.5) {$w_2$};
\node[minimum size=6mm, draw, fill=gray!30] (p3) at (18.5, 9.5) {$w_4$};

\node[minimum size=6mm, draw, fill=gray!30] (v) at (17.0, 7.0) {$w_3$};

\draw[thick] (u) to (p1);
\draw[thick] (u) to (p3);

\draw[thick] (v) to (p1);
\draw[thick] (v) to (p3);

\node[] (H15) at (17.0, 5.7) {$H_{15}$};

\end{tikzpicture}
    \caption{Examples of viable graphs with more than one independent vertex cover of size 2. Each of them has  at least two independent vertex covers of size 2. In $H_{10}$, both $\set{r_1}$ and $\set{r_2}$ are independent vertex covers of size at most 2. In $H_{11}$, both $\set{s_1}$ and $\set{s_2, s_3}$ are independent vertex covers of size at most 2; In $H_{12}$, $\set{t_1, t_3}$, $\set{t_1, t_4}$, $\set{t_2, t_3}$ and $\set{t_2, t_4}$ are all  independent vertex covers of size 2. In $H_{13}$, both $\set{u_1, u_2}$ and $\set{u_1, u_3}$ are independent vertex covers of size 2; in $H_{14}$, both $\set{v_1, v_3}$ and $\set{v_2, v_4}$ are independent vertex covers of size 2; in $H_{15}$, both $\set{w_1, w_3}$ and $\set{w_2, w_4}$ are independent vertex covers of size 2.}
    \label{fig:viable-non-unique-vc}
\end{figure}


\subparagraph*{Centre-set of a viable graph.} Consider a viable graph $H$. The \emph{centre-set} of $H$, denoted by $W(H)$, is the set of all vertices in $H$ that belong to an independent vertex cover of size at most $2$; that is,
\[
W(H) = \set{w \in V(H) ~ \Big\rvert ~ 
\begin{matrix}
    \text{ there exists } X \subseteq V(H) \text{ such that } X \text{ is an independent vertex cover of }H, \\
    ~\card{X} \leq 2, \text{ and } w \in X
\end{matrix} }.
\]
Naturally, if $H$ has a unique independent vertex cover of size at most 2, then the centre-set is precisely this independent vertex cover. However, there is some variability for the corner cases (Viable-1 to Viable-4) of the previous lemma. For technical reasons, we provide explicitly the properties of the centre-sets in the following result. 
\begin{corollary}\label{cor:centre-set}
    Consider a viable graph $H$. Then  $\card{W(H)} \leq 4$. In particular, the following assertions are true. 
    \begin{description}
        \item[Case 1.] If $H$ has at most two edges, then $\card{V(H)} \leq 4$ and $W(H) = V(H)$. 
        
        \item[Case 2.] If $H$ is a path or a cycle on exactly 4 vertices, then $W(H) = V(H)$.
        
        \item[Case 3.] If $H$ has  at least 3 edges, and $H$ has a connected component $C$ with a single edge $e = uv$,  then $H$ has exactly one other component $C'$, and  $C'$ is star with at least two edges. Moreover, $W(H) = \set{u, v, w}$, where $w$ is the unique centre of $C'$. 
        
        \item[Case 4.] Suppose $H$ is a star (or equivalently if $H$ has a vertex cover of size exactly $1$) with at least $ 3$ edges, then  $W(H)$ is the singleton set consisting of the unique centre of $H$. 
        
        \item[Case 5.] If $H$ is not covered by any of the three previous cases, then $H$ has no vertex cover of size $1$, and $H$ has a unique independent vertex cover $X$ of size exactly 2, and consequently $W(H) = X$.  
    \end{description}
\end{corollary}
\begin{proof}
        Observe first that the five cases in the statement of the lemma together show that $\card{W(H)} \leq 4$. Let us now see that each of the cases is indeed true. 

        Suppose first that $H$ has at most two edges; recall that as $H$ is viable, $H$ has no isolated vertices, and hence $H$ has either exactly one edge or exactly two edges. Suppose $H$ has exactly one edge, say $E(H) = \set{uv}$, then $V(H) = \set{u, v}$ and both $\set{u}$ and $\set{v}$ are independent vertex covers of size at most $2$, and hence $W(H) = V(H) = \set{u, v}$. On the other hand, suppose $H$ has exactly two edges, say $E(H) = \set{e, e'}$. Then there are two cases: either $e$ and $e'$ share an endpoint or they do not. In the former case, say $e = uv$ and $e' = u v'$, we have $V(H) = \set{u, v, v'}$, and the sets $\set{u}$ and $\set{v, v'}$ are both independent vertex covers of size at most $2$, and consequently $W(H) = \set{u, v, v'}$. In the latter case, say $e = uv$ and $e' = u'v'$, where $u, v, u', v'$ are all distinct, we have  $V(H) = \set{u, v, u', v'}$, and the sets $\set{u, u'}$ and $\set{v, v'}$ are both independent vertex covers of size at most $2$, and consequently $W(H) = \set{u, v, u', v'}$.  
        
        If $H$ is path on 4 vertices, say $V(H) = \set{v_1, v_2, v_3, v_4}$ and $E(H) = \set{v_1 v_2, v_2 v_3, v_3 v_4}$, then  notice that both $\set{v_1, v_3}$ and $\set{v_2, v_4}$ are independent vertex covers of $H$, and consequently, $v_1, v_2, v_3, v_4 \in W(H)$. Similarly, if $H$ is a cycle on exactly $4$ vertices, with say, $V(H) = \set{w_1, w_2, w_3, w_4}$ and $E(H) = \set{w_1 w_2, w_2 w_3, w_3 w_4, w_4 w_1}$, then again, $\set{w_1, w_3}$ and $\set{w_2, w_4}$ are independent vertex covers of $H$, and consequently $w_1, w_2, w_3, w_4 \in W(H)$. 

    Suppose now that $H$ has at least 3 edges, and $H$ has a  connected component $C$ with a single edge $e = uv$. First of all, observe that $H$ has exactly two connected components. To see this, notice that as $H$ has no isolated vertex, each connected component of $H$ contains an edge, and therefore every vertex cover must contain at least one vertex from each connected component. And as $H$ is viable and therefore has a vertex cover of size at most 2, we can conclude that $H$ has at most two connected components. But as $H$  has at least 3 edges, $H$ must have a component $C' \neq C$ with at least two edges. Thus $C$ and $C'$ are the only two connected components of $H$. And notice that every vertex cover of $H$ of size 2 must contain exactly one vertex from $C$ and exactly one vertex from $C'$. Therefore $C'$ has a vertex cover of size exactly $1$, or equivalently $C'$ is a star. As $C'$ is a star with at least two edges, $C'$ has a unique centre $w$, and notice that $\set{w}$ is the only vertex cover of $C'$  of size 1. Notice now that both $\set{u, w}$ and $\set{v, w}$ are independent vertex covers of $H$, and they are the only independent vertex covers of $H$ of size 2. Thus $W(H) = \set{u, v, w}$. 

    Suppose now that $H$ is a star with at least 3 edges. Then $H$ has a unique centre; let $w' \in V(H)$ be the unique centre. First of all, $\set{w'}$ is an independent vertex cover of $H$ of size at most 2, and hence $w' \in W(H)$. Now, as $w'$ has at least 3 edges incident with it, we have $\card{N(w')} \geq 3$, and therefore, any vertex cover not containing $w'$ must contain $N(w')$, and consequently must have size at least 3. Thus $\set{w'}$ is the only vertex cover of $H$ of size at most 2. Hence $W(H) = \set{w'}$. 

    Suppose now that $H$ is not covered by any of the three cases considered above. Since Cases 1, 2 and 3 do not apply,  by  Lemma~\ref{lem:viable-ind-vc}, $H$ has a unique independent vertex cover $X$ of size at most 2. And since Case 4, does not apply, $H$ does not have a vertex cover of size $1$, and therefore we can conclude that $H$ has a unique independent vertex cover $X$ of size exactly 2. Then by the definition of $W(H)$, we have $W(H) = X$.
\end{proof}

\begin{observation}\label{obs:non-viable-large} 
    Observe that a graph without isolated vertices and with at most two edges is viable.    
    And recall that by definition, a non-viable graph has no isolated vertices. Every non-viable graph, therefore, has at least three edges. 
\end{observation}

\subsubsection{Cliquewidth of a mixed graph}
\begin{definition}[Mixed graph]
    A mixed graph is a graph that may contain both undirected and directed edges. For convenience, we call the undirected edges simply ''edges'' and the directed edges ''arcs.'' For a mixed graph $M$, we use $E(M)$ to denote the set of edges of $M$, and $A(M)$ to denote the set of arcs of $M$. For vertices $u, v \in V(M)$, we denote the undirected edge between $u$ and $v$ by $uv$ or $\set{u, v}$, and we denote the arc directed from $u$ to $v$ by $(u, v)$. 
\end{definition}

By a labelled mixed graph, we mean a mixed graph $M$ equipped with a labeling function $\fn{\mu}{V(M)}{\mathbb{N}}$; we may not always mention $\mu$ explicitly and instead simply write that $M$ is a labelled mixed graph. For $\kw \in \mathbb{N}$, we say that a labelled mixed graph $(M, \mu)$ is $\kw$-labelled if $\mu(v) \leq \kw$ for every $v \in V(M)$. 
As a shorthand, we write $\mu(M)$ to denote the set $\mu(V(M)) = \set{\mu(v) ~|~ v \in V(M)}$. 
Let $(M_1, \mu_1)$ and $(M_2, \mu_2)$ be two labelled mixed graphs such that $V(M_1) \cap V(M_2) = \emptyset$. By the disjoint union of $(M_1, \mu_1)$ and $(M_2, \mu_2)$, we mean the labelled mixed graph $(M, \mu)$ where $V(M) = V(M_1) \cup V(M_2)$, $E(M) = E(M_1) \cup E(M_2)$ and $A(M) = A(M_1) \cup A(M_2)$, and the labeling function $\fn{\mu}{V(M)}{\mathbb{N}}$ is defined as follows: for $v \in V(M)$, we have $\mu(v) = \mu_1(v)$ if $v \in V(M_1)$, and $\mu(v) = \mu_2(v)$ if $v \in V(M_2)$.

We consider the following five operations involving labelled graphs. 
\begin{enumerate}
        \item Introduce operation: Create a vertex $v$ with label $i \in \mathbb{N}$. We denote this operation by $\intro_i(v)$. 
        
        \item Union operation: Disjoint union of two labelled mixed graphs $(M_1, \mu_1)$ and $(M_2, \mu_2)$ such that  $\mu_1(M_1) \cap \mu_2(M_2) = \emptyset$. We denote this operation by $\union((M_1, \mu_1), (M_2, \mu_2))$. 
        
        \item Relabeling operation: For a labelled mixed graph $(M, \mu)$, and for distinct $i, j \in \mathbb{N}$, change the label of all vertices with label $i$ to $j$. We denote this operation by $\relabel_{i, j}(M, \mu)$. 

        \item Edge-join operation: For a labelled mixed graph $(M, \mu)$, and for distinct $i, j \in \mathbb{N}$, add an edge between every vertex of label $i$ and every vertex of label $j$. 
        We denote this operation by $\ejoin_{i, j}(M, \mu)$. 

        \item Arc-join operation: For a labelled mixed graph $(M, \mu)$, and for distinct $i, j \in \mathbb{N}$, add an arc directed from every vertex of label $i$ to every vertex of label $j$. We denote this operation by $\ajoin_{i, j}(M, \mu)$.  
\end{enumerate}

A sequence of operations, in which each operation is one of the five operations above, is what we call a \emph{clique-expression}. For a clique-expression $\sigma$, we will use $(M^{\sigma}, \mu^{\sigma})$ to denote the labelled mixed graph constructed by $\sigma$. 

Consider a clique-expression $\sigma$ and the labelled mixed graph $(M^{\sigma}, \mu^{\sigma})$ constructed by $\sigma$. We can associate $\sigma$ with a rooted binary tree $T^{\sigma}$ in a natural way: Each node $z \in V(T^{\sigma})$ corresponds to an operation in $\sigma$, and in particular, each leaf of $T^{\sigma}$ corresponds to an introduce operation. Notice that each node of $T^{\sigma}$ that corresponds to an edge-join, arc-join or a relabeling operation has exactly one child, whereas a node corresponding to a disjoint union operation has exactly two  children. For a node $z$ of $T^{\sigma}$, let $\sigma_z$ be the sub-expression of $\sigma$ rooted at $z$, by which we mean the subsequence of $\sigma$ corresponding to the nodes in the subtree rooted at $z$. Consider  $(M^{\sigma_z}, \mu^{\sigma_z})$, the labelled mixed graph constructed by $\sigma_{z}$. Notice that $M^{\sigma_z}$ is a subgraph of $M$ consisting of all the vertices, edges and arcs  introduced in the subtree rooted at $z$. In particular, for the root $\hat z$ of $T^{\sigma}$, we have $(M^{\sigma_{\hat z}}, \mu^{\sigma_{\hat z}}) = (M^{\sigma}, \mu^{\sigma})$. 
For a node $z \in V(T^{\sigma})$ and a label $i \in [\kw]$, let $V(M^{\sigma_z, i})$ be the set of vertices of $(M^{\sigma_z}, \mu^{\sigma_z})$ with label $i$, i.e., $V(M^{\sigma_z, i}) = \set{v \in V(M^{\sigma_z}) ~|~ \mu^{\sigma_z}(v) = i}$. 

Let $\kw$ be a positive integer. We say that a clique-expression $\sigma$ is a \emph{$\kw$-expression} if $(M^{\sigma_z}, \mu^{\sigma_z})$ is $\kw$-labelled for every node $z \in V(T^{\sigma})$. 
Consider a mixed graph $M$. We say that $M$ admits a $\kw$-expression if there is a $\kw$-expression $\sigma$ such that $M = M^{\sigma}$; in this case, we also say that $\sigma$ is a $\kw$-expression for $M$. The \emph{cliquewidth} of $M$ is the least integer $\kw$ for which $M$ admits a $\kw$-expression. 

From now on, we will only work with a fixed $\kw$-expression $\sigma$ and its sub-expressions $\sigma_z$ for $z \in V(T^{\sigma})$. So to avoid clutter, we may omit $\sigma$ from the superscript in the notation introduced above, and simply write $(M, \mu)$ for $(M^{\sigma}, \mu^{\sigma})$, $T$ for $T^{\sigma}$,  $(M^z, \mu^z)$ for $(M^{\sigma_z}, \mu^{\sigma_z})$, and $V(M^{z, i})$ for $V(M^{\sigma_z, i})$.  
\begin{remark}[The disjoint labels requirement in the union operation and assumptions about the other operations in the definition of $\kw$-expression]\label{rem:assumptions}
    We note that in the $\union((M_1, \mu_1), (M_2, \mu_2))$ operation, the requirement that the labels of $M_1$ and $M_2$ be disjoint, i.e., $\mu_1(M_1) \cap \mu_2(M_2) = \emptyset$, is  usually not included in the standard definition of cliquewidth. But this requirement is harmless, as imposing $\mu_1(M_1) \cap \mu_2(M_2) = \emptyset$ would only increase the cliquewidth by a multiplicative factor of $2$. We also make a few implicit assumptions about the relabeling, edge-join and arc-join operations: We assume that each execution of these operations is  \emph{not redundant}. So, for example, we perform a $\relabel_{i, j}(\cdot)$ operation only if there already exist  at least one vertex of label $i$ and at least one vertex of label $j$. Thus for a $\kw$-expression $\sigma$, if node $z$ of $T^{\sigma}$ corresponds to the operation $\relabel_{i, j}((M, \mu))$ and $z'$ is the unique child of $z$, then $V(M^{z', i}) \neq \emptyset$ and $V(M^{z', j}) \neq \emptyset$. As for the two join operations, we preform them only if the  edges/arcs inserted by these operations are not already present in the mixed graph. That is, we perform the operation $\ejoin_{i, j}(\cdot)$ only if there does not exist any  edge between a vertex of label $i$ and a vertex of label $j$; similarly, we perform the operation $\ajoin_{i, j}(\cdot)$ only if there does not exist any arc from a vertex of label $i$ to a vertex of label $j$. We also assume that each  edge-join (resp. arc-join) operation adds at least one edge (resp. arc) to the mixed graph.  Again, these assumptions can all be made without loss of generality as any $\kw$-expression can be transformed in polynomial time into an equivalent one that satisfies all these requirements. These assumptions have previously been used  in the literature; see, for example~\cite{DBLP:journals/tcs/BergougnouxK19,DBLP:conf/esa/HegerfeldK23}.  
\end{remark}

\begin{remark}[Number of nodes in the tree $T$]\label{rem:nodes-in-T}
    Consider a \kw-expression for a graph $M$ and the corresponding tree $T$. We will rely on the known fact  that the number of nodes in $T$ is $\cO(\kw^2 \cdot \card{M})$~\cite{CourcelleO00}; it is not difficult to verify that the \kw-expression $\sigma$ contains $\card{M}$ introduce operations, at most $\card{M} - 1$ union operations, and $\cO(\kw^2 \cdot \card{M})$ other operations. 
\end{remark}

\subsubsection*{Defining the \lgp\ of a temporal pattern} For an undirected graph $H$, we use $L(H)$ to denote the line graph of $H$, which is the intersection graph of the edges of $H$; that is, the vertex  set of $L(H)$ is precisely the edge set of $H$, i.e., $V(L(H)) = E(H)$, and for vertices $e, e' \in V(L(H))$, the line graph $L(H)$ contains the edge $ee'$ if and only if the edges $e$ and $e'$ share an endpoint in $H$, i.e., $e \cap e' \neq \emptyset$. 
When talking about a graph $H$ and its line graph $L(H)$, we may refer to $H$ as the primal graph. 
Let $P = (H, R, \vartheta)$ be a temporal pattern. 
We define the \lgp\ of $P$, denoted by $\lplus(P)$,  as follows. 
First of all, $\lplus(P)$ is a mixed graph, and $\lplus(P)$ is a supergraph of $L(H)$. 
We have $V(\lplus(P)) = V(L(H))$, and we have $E(\lplus(P)) = E(L(H))$. In addition, $\lplus(P)$ contains arcs: For distinct vertices $e, e' \in V(\lplus(P))$,  $\lplus(P)$ contains the arc $(e, e')$ if and only if $\vartheta(e) \leq_{R} \vartheta(e')$ in $P$.

\subsection{Algorithm for computing $\counthoms{P}{\Gamma}$} 
The following lemma is the main technical contribution of this section.  
\begin{lemma}\label{lem:dp}
    There is an algorithm that, given a temporal graph $\Gamma = (G, \tau)$, a temporal pattern $P = (H, R, \vartheta)$ of toadwidth at most  $\kw$, and a $\kw$-expression $\sigma$ and the corresponding tree $T^{\sigma}$ for the \lgp\ of $P$ as input, runs in time $\cwDPruntime$, and outputs $\#\homs{P}{\Gamma}$. \qed
\end{lemma}

\subsubsection*{Notation}
We use the following notation throughout this section. 
Let $\Gamma = (G, \tau)$ be a temporal graph and $P = (H, R, \vartheta))$ be a temporal pattern. Let $\lt$ denote the lifetime of $\Gamma$, i.e., $\lt = \max_{e \in E(G)} \tau(e)$. 
Let $M = \lplus(P)$, the \lgp\ of $P$, and let $\kw$ be the cliquewidth of $M$; that is, $\kw$ is the toadwidth of $P$.  Let $\sigma$ be a $\kw$-expression for $M$; in particular, let $(M, \mu)$ denote  the labelled graph constructed by $\sigma$. Also, let $T = T^{\sigma}$, the rooted tree corresponding to $\sigma$, and let $\hat z \in V(T)$ denote the root of $T$. We assume that we are given $\sigma$ along with $\Gamma$ and $P$. For a node $z$ of $T$, recall that $M^z$ is the subgraph of $M$ consisting of all the vertices, edges and arcs introduced in the subtree rooted at $z$, and for a label $i \in [\kw]$, recall that $V(M^{z, i})$ is the set of all vertices of $M^z$ with label $i$. Recall also that the elements of $V(M^{z, i})$ are edges in the primal graph $H$, i.e., $V(M^{z, i}) \subseteq E(H)$. Thus the subgraph of $H$ induced by $V(M^{z, i})$ is well-defined.\footnote{For $E' \subseteq E(H)$, the subgraph of $H$ induced by $E'$, denoted by $H[E']$, is defined as follows: The vertex set $H[E']$ contains all those vertices of $H$ that are the endpoints of the edges in $E'$, and the edge set of $H[E']$ is precisely $E'$.} We denote this subgraph by $H^{z, i}$. In particular, $V(H^{z, i})$ is the set of all vertices of $H$ that are the endpoints of the primal-graph-edges in $V(M^{z, i})$, or more formally, $V(H^{z, i}) \subseteq V(H)$ is the set of all  vertices $v \in V(H)$ such that $uv \in V(M^{z, i})$ for some $u \in V(H)$. We also define $H^z$ to be the union of $H^{z, i}$s; that is, $H^z$ is the subgraph of $H$ with $V(H^z) = \bigcup_{i \in [\kw]} V(H^{z, i})$ and $E(H^z) = \bigcup_{i \in [\kw]} E(H^{z, i}) = V(M^{z})$. Similarly, we define $P^z$ to be the temporal pattern $(H^z, R^z, \vartheta^z)$, where $R^z$ is sub-poset of $R$ whose ground-set is precisely the set $\vartheta(E(H^z)) = \set{t \in R ~|~ \text{ there exists } e \in E(H^z) \text{ such that } \vartheta(e) = t}$, and the partial order  $\leq_{R^z}$ is simply the restriction of $\leq_R$ to $E(H^z)$, and similarly $\vartheta^z$ is the restriction of $\vartheta$ to $E(H^z)$. We may omit the superscript $z$ from $R^z$ and $\vartheta^z$, unless it is absolutely necessary to specify them.   

As $V(M^{z, i})$ is a set of vertices in $\lplus(P) = M$ and a set of edges in $H$, to avoid any confusion caused by using the words vertices and edges, we will instead refer to the elements of $V(M^{z, i})$ as simply ``elements.''  

 

In what follows, we design a dynamic programming algorithm to compute $\counthoms{P}{\Gamma}$, which will prove Lemma~\ref{lem:dp}. 

\subsubsection*{Outline of the DP}
Given a temporal graph $\Gamma$, a  temporal pattern $P$, a $\kw$-expression $\sigma$ for the \lgp\ of $P$ and the corresponding tree $T$, we do a DP over $T$ as follows. At each node $z$ of $T$, we count the number of homomorphisms from $P^z$ to $\Gamma$. To do this,  
we exploit the properties imposed by the line graph structure of $M$. In particular, we will argue that for every ``active'' label class $V(M^{z, i})$, the corresponding subgraph $H^{z, i}$ is a viable graph; by an active label class $V(M^{z, i})$, we mean that all elements of $V(M^{z, i})$ may be involved in a future $\ejoin_{\cdot,
\cdot}(\cdot)$ operation. That is, $V(M^{z, i})$ is active if there exists an 
ancestor $z'$ of $z$ in the tree $T$ such that $z'$ corresponds to an $\ejoin_{\cdot,
\cdot}(\cdot)$ operation, and for every $e \in V(M^{z, i})$, a subset of the line-graph-edges incident with $e$ is inserted at node $z'$. 
And if $H^{z, i}$ is not viable 
then the elements of $V(M^{z, i})$ cannot be involved in any future $\ejoin_{\cdot, \cdot}(\cdot)$ operation. 
Using this observation, we define a small subset of primal-graph-vertices $W(z) \subseteq V(H^{z})$ as follows: $W(z)$ contains all the vertices the centre-set of $H^{z, i}$ for all active (and hence viable) $H^{z, i}$. 
We will have $\card{W(z)} \leq 4 \kw$. Now, to count the number of homomorphisms, we guess how the vertices of $W(z)$ are mapped by a homomorphism. For each $i \in [\kw]$, we also guess the earliest and the latest time-steps to which an edge of $H^{z, i}$ is mapped by a homomorphism. And for each combination of these two guesses, we count the number of homomorphisms from $P^z$ to $\Gamma$ that are consistent with the guesses. 

\subsubsection*{Ingredients for the DP}
Consider a node $z$ of $T$ and a label $i \in [\kw]$. Recall that $H^{z, i}$ is the subgraph of $H$ induced by $V(M^{z, i})$. Notice that by definition, $H^{z, i}$ has no isolated vertex and hence $\card{V(H^{z, i})} \leq 2 \card{E(H^{z, i})} = 2 \card{V(M^{z, i})}$. 
With a slight abuse of terminology, we say that $(z, i)$ is \emph{viable} if $H^{z, i}$ is viable, and that $(z, i)$ is \emph{non-viable} if $H^{z, i}$ is non-viable. 
Consider a viable pair $(z, i)$. Recall that $W(H^{z, i})$ denotes the centre-set of $H^{z, i}$; by Lemma~\ref{cor:centre-set}, we have $\card{W(H^{z, i})} \leq 4$.  

The next lemma says that if $(z, i)$ is non-viable, then all the edges of $M$ incident with $V(M^{z, i})$ have already been inserted in the sub-tree rooted at $z$. Equivalently, no more edge-join operations involving $V(M^{z, i})$ can be performed at an ancestor of $z$. 
\begin{lemma}\label{lem:non-viable}
    Let $z \in V(T)$ and $i \in [\kw]$. 
    Consider the set $F$ of edges of $M$ incident with $V(M^{z, i})$,  i.e., $F = \set{f \in E(M)  ~|~ f \text{is incident with  an element in } V(M^{z, i})}$. If $(z, i)$ is non-viable, then $F \subseteq E(M^z)$.  
\end{lemma}

\begin{proof}
    Assume that $(z, i)$ is non-viable, 
    and assume for a contradiction that $F \nsubseteq E(M^z)$. Let $ee' \in F \setminus E(M^z)$, where $e \in V(M^{z, i})$.  Then the tree $T$ must contain a node $z'$ at which the line-graph-edge $ee'$ is inserted through an edge-join operation. In particular, there is a node $z' \in V(T)$ such that $z'$ is an ancestor of $z$, and $z'$ corresponds to the operation $\ejoin_{i', j'}(\cdot)$ for distinct $i', j' \in [\kw]$ with $e \in V(M^{z', i'})$ and $e' \in V(M^{z', j'})$. 
    Notice that $i'$ need not be equal to $i$; in particular, on the unique path in $T$ from $z$ to $z'$, there could be a sequence of relabel operations, say $(\relabel_{i, i_1}(\cdot), \relabel_{i_1, i_2}(\cdot),\ldots, \relabel_{i_{p - 1}, i_p}(\cdot), \relabel_{i_p, i'}(\cdot)$ so that the elements of $V(M^{z, i})$ have the label $i'$ at node $z'$, and thus $V(M^{z, i}) \subseteq V(M^{z', i'})$. 
    Now, as $z'$ corresponds to $\ejoin_{i',j'}(\cdot)$, every element of $V(M^{z', i'})$ is adjacent to every element of $V(M^{z', j'})$. Then, as $V(M^{z, i}) \subseteq V(M^{z', i'})$ and as $e' \in V(M^{z', j'})$, every element of $V(M^{z, i})$ is adjacent to $e'$. 
    
    Now, consider $H^{z, i}$, the subgraph of $H$ induced by $V(M^{z, i})$. We will show that there exists an edge $\hat e \in E(H^{z, i})$ such that $\hat e e' \notin E(M)$, a contradiction to the earlier claim that every element of $V(M^{z, i}) = E(H^{z, i})$ is adjacent to $e'$; to do this, notice that we only need to show that there exists $\hat e \in E(H^{z, i})$ such that $\hat e$ and $e'$ do not have any common endpoint. Let $e' = u'v'$. If  neither $u'$ nor $v'$ are in $V(H^{z, i})$  then no edge in $E(H^{z, i})$ has a common endpoint with $e'$, and therefore $\hat e e' \notin E(M)$ for every $\hat e \in E(H^{z, i})$; notice that $E(H^{z, i}) \neq \emptyset$ as $H^{z, i}$ non-viable and hence has at least 3 edges, and therefore such an $\hat e$ indeed exists. So assume from here on that either $u' \in V(H^{z, i})$ or $v' \in V(H^{z, i})$ (or both). We now invoke the fact that $H^{z, i}$ is non-viable, and in particular the fact that $H^{z, i}$ has at least 3 edges, but no independent vertex cover of size at most 2.  Suppose exactly one of $u'$ and $v'$ belongs to $V(H^{z, i})$, say $u' \in V(H^{z, i})$ and $v' \notin V(H^{z, i})$. Notice then that  there must exist  an edge $\hat e \in E(H^{z, i})$ such that $\hat e$ is not incident with $u'$; if not, then every edge in $H^{z, i}$ must be incident with $u'$, and hence $H^{z, i}$ would be a star with centre $u'$, and consequently $H^{z, i}$ would be viable. Fix such an $\hat e$.  As $\hat e \in E(H^{z, i})$, both endpoints of $\hat e$ belong to $V(H^{z, i})$, and by assumption $v' \notin V(H^{z, i})$; thus $\hat e$ is not incident with $v'$ either, and therefore $\hat e$ and $e'$ do not have any common endpoint. 
    Suppose now that both $u'$ and  $v'$ belong to $V(H^{z, i})$.  Notice that in this case $u'v' \notin E(H^{z, i})$, as $u'v' = e' \notin V(M^{z, i}) = E(H^{z, i})$. Thus $\set{u', v'}$ is an independent set in $H^{z, i}$, and as $H^{z, i}$ is non-viable, we can conclude that $\set{u', v'}$ is not a vertex cover for $H^{z, i}$, which implies that  there exists an edge $\hat e \in E(H^{z, i})$ such that $\hat e$ is not incident with $u'$ or $v'$.  Thus $\hat e$ and $e'$ do not have any common endpoint in this case either. Therefore in either case, $\hat e e' \notin E(M)$, a contradiction.  

    We have thus shown that the assumption $F \nsubseteq E(M^z)$ leads to a contradiction, and hence the lemma follows.
\end{proof}

Lemma~\ref{lem:non-viable} tells us that if we must add line-graph-edges incident with $V(M^{z, i})$ at a node $z'$, which is an ancestor of $z$, then $(z, i)$ has to be viable. Not just that, we can say more: Recall that for $e \in V(M^{z, i})$ and $e' \in V(M^{z, j})$, we add the line-graph-edge $ee'$ (at an ancestor $z'$ of $z$) only if $e$ and $e'$ share an endpoint, say $v \in V(H)$. In this case, we can show that $v$ is in the centre-set $W(H^{z, i})$. 

\begin{lemma}\label{lem:common-endpoint}
    Let $z \in V(T)$ and $i \in [\kw]$, and let $e = uv \in V(M^{z, i})$. If there exists $e' = u'v \in V(M) \setminus V(M^{z, i})$ and $ee' \notin E(M^{z})$, then $(z, i)$ is viable and $v \in W(H^{z, i})$. 
\end{lemma}
\begin{proof}
    Suppose there exists $e' = u'v \in V(M) \setminus V(M^{z, i})$ and $ee' \notin E(M^{z})$. As $e$ and $e'$ have the common endpoint $v$, we have $ee' \in E(M)$. Observe now that as $ee' \in E(M) \setminus E(M^{z})$, there exists an ancestor $z'$ of $z$, such that $z'$ corresponds to the $\ejoin_{\cdot, \cdot}(\cdot)$ operation at which the (line-graph-edge) $ee'$ is inserted.  Then, as $e \in V(M^{z, i})$, we can conclude that $e'$ is adjacent in $M^{z'}$ to every element of $V(M^{z, i})$; that is, for every $\hat e \in V(M^{z, i})$, we must have $\hat e e' \in E(M)$, and therefore $\hat e$ and $e'$ must have a common endpoint; we will use this fact later in the proof. 
    
    Let us now see that $(z, i)$ is viable. 
    As $ee' \in E(M)$ and $e \in V(M^{z, i})$, the line-graph-edge $ee'$ is incident with an element in $V(M^{z, i})$; also, $ee' \notin E(M^z)$. Then, Lemma~\ref{lem:non-viable} implies that $(z, i)$ is viable. Hence $W(H^{z, i})$ is well-defined; recall that $W(H^{z, i})$ is the centre-set of the graph $H^{z, i}$, the subgraph of $H$ induced by $V(M^{z, i})$.  
    
    Let us now argue that $v \in W(H^{z, i})$. We will consider several cases, and invoke Lemma~\ref{cor:centre-set} to conclude that $v \in W(H^{z, i})$ in each case. 
 
    If $H^{z, i}$ has at most two edges, or if $H^{z, i}$ is a path or a cycle on exactly 4 vertices, then $W(H^{z, i}) = V(H^{z, i})$, and in particular, $v \in W(H^{z, i})$ (Corollary~\ref{cor:centre-set}-Cases 1 and 2). From now on, assume that $H^{z, i}$ has at least three edges. Suppose $H^{z, i}$ has a connected component $C_1$ with a single edge. Then by Lemma~\ref{cor:centre-set}-Case 3, $H$ has exactly one other component $C_2$, and $C_2$ is a star with at least 2 edges. Let $e_1$ be the single edge in the component $C_1$. Now, if $e_1  = e (= uv)$, then by Lemma~\ref{cor:centre-set}-Case 3, we have $v \in W(H^{z, i})$. Suppose this is not the case.  Then $e = uv$ belongs to the component $C_2$; we will argue that $v$ is the unique centre of $C_2$, which again will imply that $v \in W(H^{z, i})$. As $C_2$ has at least two edges, there exists an edge $e_2$ in $C_2$ such that $e_2 \neq e$.  
    As observed earlier, for every $\hat e \in V(M^{z, i})$, and in particular, for $\hat e \in \set{e, e_1, e_2}$, the edges $\hat e$ and $e'$ have a common endpoint; recall also that the common endpoint of $e$ and $e'$ is $v$. We argue that the only way this is possible is if the centre of $C_2$ is precisely $v$, so that $v$ is the common endpoint of $e_2$, $e$ and $e'$, (and $u'$ is an endpoint of $e_1$). If $v$ were not the centre of the star $C_2$, then the centre must be $u$, which implies that $e_2 = u v_2$ for some vertex $v_2 \in V(H^{z, i})$. Then, as $e_2 = u v_2$ and $e' = u' v$ share an endpoint, (and as $u \neq u'$), we must have $v_2 = u'$. Thus both the endpoints of $e'$ (i.e., $v$ and $u'$) are present in the component $C_2$. But then $e'$ and $e_1$ (the single edge in the component $C_1$) do not have any common endpoint, which is a contradiction.  
    Again by Corollary~\ref{cor:centre-set}-Case 3, we have $v \in W(H^{z, i})$. 

    Suppose $H^{z, i}$ is a star. This case is similar to the previous one; here also, it must be be the case that $v$ is the unique centre of $H^{z, i}$. To see this, notice that as $H^{z, i}$ has at least 3 edges, there exist two distinct edges $e_3, e_4 \in E(H^{z, i}) \setminus \set{e}$ such that for each $\hat e \in \set{e, e_3, e_4}$, $\hat e$ and $e'$ have a common endpoint; recall again that the common endpoint of $e$ and $e'$ is $v$.  Again, the only way this is possible is if $v$ is the centre of $H^{z, i}$, and thus $v \in W(H^{z, i})$ (by Corollary~\ref{cor:centre-set}-Case 4). 

    Now, if none of the previous cases occur, then by Corollary~\ref{cor:centre-set}-Case 5, $H^{z, i}$ has no vertex cover of size $1$, and $H^{z, i}$ has a unique independent vertex cover $X$ of size exactly 2, and consequently $W(H^{z, i}) = X$. We will show that $X = \set{u', v}$ (recall that $e' = u'v$), which will imply that $v \in W(H^{z, i})$. 
    Suppose $X \neq \set{u', v}$. Notice that the vertex $u'$ may or may not belong to the graph $H^{z, i}$.  In either case, as $u'v = e' \notin V(M^{z, i})$, the set $\set{u', v} \cap V(H^{z, i})$ is an independent set in $H^{z, i}$. Then, as $X$ is the unique independent vertex cover of $H^{z, i}$ of size at most $2$, we can conclude that $\set{u', v} \cap V(H^{z, i})$ is not a vertex cover of $H^{z, i}$. And hence there exists an edge $\tilde e \in E(H^{z, i})~ (= V(M^{z, i}))$ such that $\tilde e$ is not incident with $u'$ or $v$; thus, $\tilde e$ and $e'$ do not have any common endpoint, and hence $\tilde e e' \notin E(M)$, a contradiction. We can therefore conclude that $X = \set{u', v}$ and hence $v \in W(H^{z, i})$. 
\end{proof}

Consider a node $z \in V(T)$. To design our DP, we first need to formalize our ``guesses'' at  $z$. And we will guess the following details of every homomorphism $\varphi = (\nu, \xi)$: (i) how $\nu$ maps $W(H^{z, i})$ for each viable $(z, i)$, and (ii) the earliest and latest time-steps to which elements of $V(M^{z, i})$ are mapped by $\tau \circ \xi$ for each label $i$. For each choice of guesses, we count the number of homomorphisms that correspond to the choice. 
With this in mind, we now define three sets that we will  use to formalize the guesses. 

\subparagraph*{Definition of $W(z)$.} Consider a node $z \in V(T)$. 
We define $W(z) \subseteq V(H)$ as follows:
\[
W(z) = \bigcup_{\substack {i \in [\kw]\\ (z, i) \text{ is viable}}} W(H^{z, i}).
\]


\subparagraph*{Definition of $\maps(z)$.} 
For a node $z \in V(T)$, let $\maps(z)$ be the set of all maps $\fn{\vmap}{W(z)}{V(G)}$. 

Notice that the set $W(z)$ may be empty, in which case $\maps(z)$ is the singleton set containing the empty function. Thus $\maps(z) \neq \emptyset$ for every $z \in V(T)$. 

\subparagraph*{Definition of $\tms(z)$.} For $z \in V(T)$, we  define a set $\tms(z)$ of tuples as follows. 
Let $\tms(z)$ be the set of all tuples 
\[((\hat t_{1}, \tilde t_{1}), (\hat t_{2}, \tilde t_{2}),\ldots, (\hat t_{\kw}, \tilde t_{\kw})),
\]
where for every $i \in [\kw]$, the following conditions hold: (a) $\hat t_{i}, \tilde t_{i} \in [\lt] \cup \set{\nnn}$, (b) for each $x \in \set{\hat t_i, \tilde t_i}$, we have $x = \nnn$ if and only if $V(M^{z, i}) = \emptyset$, and (c) if $\hat t_i, \tilde t_i \in [\lt]$ (equivalently, if $V(M^{z, i}) \neq \emptyset$), then  $\hat t_i \leq \tilde t_i$. 
Notice that it may be the case that $\hat t_i = \tilde t_i$. From now on, we may write $(\hat t_i, \tilde t_i)_{i \in [\kw]}$ as a shorthand for the $\kw$-tuple $((\hat t_{1}, \tilde t_{1}), (\hat t_{2}, \tilde t_{2}),\ldots, (\hat t_{\kw}, \tilde t_{\kw}))$. And we use  boldface letters like $\bt, \bt'$, etc. to denote the elements of $\tms(z)$.


\begin{observation}\label{obs:bounds}
    Let $z$ be a node of $T$.  Observe that the following bounds hold. 
    \begin{enumerate}
        \item First, $\card{W(z)} \leq 4 \kw$. To see this, consider $i \in [\kw]$. 
        If $(z, i)$ is non-viable, then the index $i$ does not contribute any vertex to $W(z)$. 
        If $(z, i)$ is viable, then by Corollary~\ref{cor:centre-set}, we have $\card{W(H^{z, i})} \leq 4$. Thus each $i \in [\kw]$ contributes at most $4$ vertices of $H$ to $W(z)$. 


        
        \item Second, $\card{\maps(z)} \leq \card{G}^{4 \kw}$. This holds because $\card{W(z)} \leq 4 \kw$, and a function $\vmap \in \maps(z)$ may map each element of $W(z)$ to any one of the $\card{G}$ vertices of $G$.  

        \item Third, $\card{\tms(z)} \leq \lt ^{2 \kw}$. This holds because for each $i \in [\kw]$, the pair $(\hat t_i, \tilde t_i)$ has (i) at most $\lt^2$ choices if $V(M^{z, i}) \neq \emptyset$, and $(ii)$ exactly one choice, namely $(\hat t_i, \tilde t_i) = (\nnn, \nnn)$, if $V(M^{z, i}) = \emptyset$.  
    \end{enumerate}
\end{observation}

\begin{definition}[$(z, \vmap, \bt)$-compatible homomorphism]\label{def:compatible}
    Consider a node $z \in V(T)$, a map  $\vmap \in \maps(z)$, and  a tuple $\bt = (\hat t_i, \tilde t_i)_{i \in [\kw]} \in \tms(z)$, and a temporal homomorphism $\varphi = (\nu, \xi)$ from $P^z = (H^z, R^z, \vartheta^z)$ to $\Gamma = (G, \tau)$. We say that $\varphi$ is \emph{$(z, \vmap, \bt)$-compatible} if the following conditions hold:
\begin{description}
    \item[(CC1).]\label{CC1} $\nu{|}_{W(z)} = \vmap$,  
    \item[(CC2).]\label{CC2} there exist $\hat e_i, \tilde e_i \in V(M^{z, i})$ such that $\tau(\xi(\hat e_i)) = \hat t_i$ and $\tau(\xi(\tilde e_i)) = \tilde t_i$ for every $i \in [k]$ with $V(M^{z, i}) \neq \emptyset$, and 
    \item[(CC3).]\label{CC3} $\hat t_i \leq \tau(\xi(e)) \leq \tilde t_i$ for every $e \in V(M^{z, i})$ and for every $i \in [\kw]$ with $V(M^{z, i}) \neq \emptyset$.\qedhere 
\end{description}
\end{definition}

Let $\comhom(z, \vmap, \bt)$ denote the set of all $(z, \vmap, \bt)$-compatible homomorphisms from $P^z$ to $\Gamma$. 
Before proceeding further, let us see how we can use Definition~\ref{def:compatible} to count the number of homomorphisms from $P = (H, R, \vartheta)$ to $\Gamma = (G, \tau)$. 

\begin{lemma}\label{lem:unique}
    Let $z \in V(T)$. For every homomorphism $\varphi = (\nu, \xi)$ from $P^z = (H^z, R^{z}, \vartheta^z)$ to $\Gamma =(G, \tau)$, there exists a unique pair $(\vmap, \bt)$, where $\vmap \in \maps(z)$ and $\bt \in \tms(z)$  such that $\varphi$ is an $(z, \vmap, \bt)$-compatible homomorphism. 
\end{lemma}
\begin{proof}
    Consider a homomorphism $\varphi = (\nu, \xi)$ from $P^z = (H^z, R^{z}, \vartheta^z)$ to $\Gamma =(G, \tau)$. We will show that there exist $\vmap \in \maps(z)$ and $t \in \tms(z)$as required by the statement of the lemma. 

    We define $\vmap$ to be $\nu|_{W(z)}$, the restriction of $\nu$ to $W(z)$, i.e., $\vmap(v) = \nu(v)$ for every $v \in W(z)$; if $W(z) = \emptyset$, then we take $\vmap$ to be the empty function. 

    We define the tuple $\bt = (\hat t_i, \tilde t_i)_{i \in [\kw]}$ as follows: For $i \in [\kw]$, 
    \[
    \begin{matrix}
        \hat t_i = \begin{cases}
            \min_{e \in V(M^{z, i})} \tau(\xi(e)), \text{ if } V(M^{z, i}) \neq \emptyset, \text{ and} \\
            \nnn, \text{ otherwise;} 
        \end{cases} &
        \tilde t_i = \begin{cases}
            \max_{e \in V(M^{z, i})} \tau(\xi(e)), \text{ if } V(M^{z, i}) \neq \emptyset, \text{ and} \\
            \nnn, \text{ otherwise.}
        \end{cases}
    \end{matrix}
    \]

Notice that for each $i \in [k]$, we have $\hat t_i, \tilde t_i \in [\lt]$ if $V(M^{z, i}) \neq \emptyset$, and $\hat t_i = \tilde t_i = \nnn$ otherwise. Thus by definition $\tms(z)$, we have $\bt \in \tms(z)$. 

By definition, $(\vmap, \bt)$ satisfies conditions (CC1), (CC2) and (CC3) in Definition~\ref{def:compatible}, and thus $\varphi$ is an $(z, \vmap, \bt)$-compatible homomorphism. 

Let us now see that $(\vmap, \bt)$ is the unique pair in $\maps(z) \times \tms(z) $ for which $\varphi$ is $(z, \vmap, \bt)$-compatible. Suppose not, and assume that there exists a pair $(\vmap', \bt') \in \maps(z) \times \tms(z)$ such that $(\vmap', \bt') \neq (\vmap, \bt)$ and $\varphi = (\nu, \xi)$ is $(z, \vmap', \bt')$-compatible. Then condition (CC1) implies that $\vmap' = \nu|_{W(z)}$. But by the definition of $\vmap$, we have  $\vmap = \nu|_{W(z)}$, and thus $\vmap' = \vmap$. Then, as $(\vmap', \bt') \neq (\vmap, \bt)$, we must have $\bt' \neq \bt$. Let $\bt' = (\hat t'_i, \tilde t'_i)_{i \in [\kw]}$. As $\bt' \neq \bt$,  there exists $i \in [\kw]$ such that $(\hat t'_i, \tilde t'_i) \neq (\hat t_i, \tilde t_i)$; fix such an $i \in [\kw]$. Notice that $V(M^{z, i}) \neq \emptyset$, as otherwise, we would have $(\hat t'_i, \tilde t'_i) = (\nnn, \nnn) = (\hat t_i, \tilde t_i)$. Condition (CC2) now implies that there exist $\hat e_i, \tilde e_i \in V(M^{z, i})$ such that $\tau(\xi(\hat e_i)) = \hat t'_i$ and $\tau(\xi(\tilde e_i)) = \tilde t'_i$. Let us first argue that $\hat t'_i = \hat t_i$. By the definition of $\hat t_i$, we have $\hat t_i = \min_{e \in V(M^{z, i})} \tau(\xi(e)) \leq \tau(\xi(\hat e_i)) = \hat t'_i$. Condition (CC3) implies that $\hat t'_i \leq \tau(\xi(e))$ for every $e \in V(M^{z, i})$, which implies that $\hat t'_i \leq \min_{e \in V(M^{z, i})} \tau(\xi(e)) = \hat t_i$. We have thus shown that $\hat t_i \leq \hat t'_i$ and $\hat t'_i \leq \hat t_i$, and thus $\hat t'_i = \hat t_i$. Symmetric arguments will show that $\tilde t'_i = \tilde t_i$, and thus $(\hat t'_i, \tilde t'_i) = (\hat t_i, \tilde t_i)$, a contradiction. 
\end{proof}


In light of Lemma~\ref{lem:unique}, we can define an equivalence relation $\sim_z$ on the set $\homs{P^z}{\Gamma}$ of all homomorphisms from $P^z$ to $\Gamma$ as follows: For homomorphisms $\phi, \phi' \in \homs{P^z}{\Gamma}$, we say that $\phi$ and $\phi'$ are equivalent under $\sim_z$ if there exists $(\vmap, \bt) \in \maps(z) \times \tms(z)$ such that both $\phi$ and $\phi'$ are $(z, \vmap, \bt)$-compatible. It is straightforward to verify that $\sim_z$ is indeed an equivalence relation, and that the sets $\comhom(z, \vmap, \bt)$ are precisely the equivalence classes under $\sim_z$. This observation immediately implies the following results. 

\begin{corollary}\label{cor:unique-sum}
    Let $z \in V(T)$. Then $\homs{P^z}{\Gamma} = \biguplus_{(\vmap, \bt) } \comhom(z, \vmap, \bt)$, and $\#\homs{P^z}{\Gamma} = \sum_{(\vmap, \bt)} \card{\comhom(z, \vmap, \bt)}$, where both the union and the summation are over all pairs $(\vmap, \bt) \in  \maps(z) \times \tms(z)$.  
\end{corollary}
Corollary~\ref{cor:unique-sum} tells us that to compute $\#\homs{P^z}{\Gamma}$, it is enough to compute $\card{\comhom(z, \vmap, \bt)}$ for every $(\vmap, \bt)$. In particular, we can compute $\#\homs{P}{\Gamma}$ by computing $\card{\comhom(\hat z, \vmap, \bt)}$ for every $(\vmap, \bt) \in \maps(\hat z) \times \tms(\hat z)$, where $\hat z$ is the root of the tree $T$ (and hence $P^{\hat z} = P$). 
We now describe our dynamic programming algorithm, which will do precisely this.  
 
\subsubsection*{Definition of the states of the DP}

We are now ready to define the states of our DP. For every node $z$ of $T$, every map $\vmap \in \maps(z)$, and every tuple $\bt \in \tms(z)$ we define 
\[
\cwDP[z, \vmap, \bt] = \card{\comhom(z, \vmap, \bt)}. 
\]

\subsubsection*{Computation of the DP table entries}

We compute the entries of the DP table in a bottom-up fashion over $T$. For each node $z$ of $T$, assuming we have computed all entries corresponding to all descendants $z'$ of $z$, we can compute the entries corresponding to $z$. We split the computation into cases depending on what kind of operation $z$ corresponds to. 

Consider a node $z \in V(T)$, a map $\vmap \in \maps(z)$ and a tuple $\bt = (\hat t_q, \tilde t_q)_{q \in [\kw]} \in \tms(z)$. We discuss below how we can compute $\cwDP[z, \vmap, \bt]$.  

\subparagraph*{Introduce operation.} Suppose $z$ corresponds to the operation $\intro_i(e)$; that is, we introduce the vertex $e$ of $M = \lplus(P)$ with label $i \in [\kw]$. Let $e = uv$ for $u, v \in V(H)$. 
Recall that $z$ is a leaf in the tree $T$. 
Then $e \in V(M^{z, i})$, and more important, $V(M^{z, i}) = \set{e}$ and $V(M^{z, j}) = \emptyset$ for every $j \in [\kw] \setminus \set{i}$. Notice also that $H^{z}$ is the graph with two vertices $u$ and $v$ and the single edge $uv$, and so by Corollary~\ref{cor:centre-set}, we have $W(z) = \set{u, v}$. 
Now, consider the tuple  $\bt = (\hat t_q, \tilde t_q)_{q \in [\kw]}$. As $V(M^{z, j}) = \emptyset$ for every $j \in [\kw] \setminus \set{i}$, we must have $\hat t_j = \tilde t_j = \nnn$ for every $j \in [\kw] \setminus \set{i}$.


 Now, suppose $\varphi =(\nu, \xi)$ is an $(z, \vmap, \bt)$-compatible homomorphism from $P^z$ to $\Gamma = (G, \tau)$. As $W(z) = \set{u, v} = V(H^z)$, condition (CC1) (in  Definition~\ref{def:compatible}) implies that $\nu = \alpha$. And the definition of a homomorphism implies that there exists an edge $e^{\star}$ in $G$ between the vertices $\vmap(u) = \nu(u)$ and $\vmap(v) = \nu(v)$ with   $\xi(e) = e^{\star}$.  Finally, as  $P^z$ is a pattern with a single edge $e = uv$, condition (CC2)  implies that $\hat t_i = \tau(\xi(e)) = \tilde t_i$. Recall also that by the definition of a temporal graph, there exists at most one edge $e^{\star}$ between $\vmap(u)$ and $\vmap(v)$ with $\tau(e^{\star}) = \hat t_i = \tilde t_i$. In other words, an $(z, \vmap, \bt)$-compatible homomorphism is uniquely determined by $\vmap$ and $\tau$. 
 We thus have
 \[
 \cwDP[z, \vmap, \bt] = \begin{cases}
     1, \text{ if there exists an edge } e^{\star} \in E(G)  \text{ with }e^{\star} = \vmap(u) \vmap(v) \text{ and } \tau(e^{\star}) = \hat t_i = \tilde t_i, \\
     0, \text{ otherwise.}
 \end{cases}
 \]

 \subparagraph*{Union operation.} Suppose $z$ corresponds to the operation $\union(\cdot, \cdot)$. Let $z_1$ and $z_2$ be the two children of $z$. Then $(M^{z}, \mu^z)$ is the disjoint union of the two mixed graphs $(M^{z_1}, \mu^{z_1})$ and $(M^{z_2}, \mu^{z_2})$. Recall that by the definition of the $\union(\cdot, \cdot)$ operation, the set of labels used in $(M^{z_1}, \mu^{z_1})$ and the set of labels used in $(M^{z_2}, \mu^{z_2})$ are disjoint, i.e., $\mu^{z_1}(M^{z_1}) \cap \mu^{z_2}(M^{z_2}) = \emptyset$. This implies that for each $q \in [\kw]$, at most one of $V(M^{z_1, q})$ and $V(M^{z_2, q})$ is non-empty; we will  rely on this property a few times in the discussion that follows.  
 

 Now, as $V(M^{z_1}) \cap V(M^{z_2}) = \emptyset$, no edge of $H$ is contained in both $M^{z_1}$ and $M^{z_2}$. But there may be a vertex $v \in V(H)$ such that $v$ is an endpoint of $e_1 \in V(M^{z_1})$ and of $e_2 \in V(M^{z_2})$. Notice that $e_1 e_2$ is an edge of the line graph $L(H)$ and hence the \lgp\ $M = \lplus(H)$ in this case. And the edge $e_1 e_2$ will be inserted at a future $\ejoin_{\cdot, \cdot}(\cdot)$ operation.  
 We argue next that $v$ must belong to all three sets $W(z)$, $W(z_1)$ and $W(z_2)$. 
 \begin{claim}\label{claim:union}
     Consider a vertex $v \in V(H)$ and edges $e_1, e_2 \in E(H)$ such that $e_1 \in V(M^{z_{1}})$ and $e_2 \in V(M^{z_2})$ and $v$ is an endpoint of both $e_1$ and $e_2$. Then $v \in W(z) \cap W(z_1) \cap W(z_2)$. 
 \end{claim}
 \begin{proof} 
      Let $e_1 =  u_1 v$ and $e_2 = u_2 v$. 
      And let $i_1 = \mu^{z}(e_1)$ and $i_2 = \mu^z(e_2)$. Notice that as $(M^z, \mu^z)$ is the disjoint union of $(M^{z_1}, \mu^{z_1})$ and $(M^{z_2}, \mu^{z_2})$, we also have $\mu^{z_1}(e_1) = i_1$ and $\mu^{z_2}(e_2) = i_2$. Notice also that as $\mu^{z_1}(M^{z_1}) \cap \mu^{z_2}(M^{z_2}) = \emptyset$, and as $V(M^{z_1, i_i}) \neq \emptyset$ (because $e_1  \in V(M^{z_1, i_1})$), we have $V(M^{z_2, i_1}) = \emptyset$, and therefore we have $V(M^{z, i_1}) = V(M^{z_1, i_1})$. Using symmetric arguments, we have $V(M^{z_1, i_2}) = \emptyset$ and consequently,  $V(M^{z, i_2}) = V(M^{z_2, i_2})$. 
      
      Now, consider $p \in \set{1, 2}$. Recall that we defined $H^{z, i_p}$ to be the subgraph of $H$ induced by $V(M^{z, i_p})$. 
      As $v$ is an endpoint of $e_p \in V(M^{z, i_p})$, we have $v \in V(H^{z, i_p})$. 
      As $V(M^{z_p, i_p}) = V(M^{z, i_p})$, we have $H^{z_p, i_p} = H^{z, i_p}$, and thus, $v \in V(H^{z_p, i_p})$ as well. 
      Now, as $e_1$ and $e_2$ share the endpoint $v$, notice that $e_1 e_2$ is an edge of the line graph $L(H)$ (and hence of the \lgp\ $M$). But notice that $e_1 e_2 \notin E(M^{z})$, (as any $\ejoin_{\cdot, \cdot}(\cdot)$ operation between $e_1$ and $e_2$ is yet to be performed), and therefore there must exist an edge-join node $z'$ in $T$ such that $z'$ is an ancestor of $z$, and $e_1 e_2$ is inserted at $z'$. Then  Lemma~\ref{lem:common-endpoint} implies that (i) $(z,  i_1)$ is viable and $v \in W(H^{z, i_1})$, and that (ii) $(z, i_2)$ is viable and $v \in W(H^{z, i_2})$. 
      Now, by the definition of $W(z)$, we have $W(z) \supseteq W(H^{z, i_1})$, and thus $v \in W(z)$. 
      Also, as $V(M^{z_1, i_1}) = V(M^{z, i_1})$, the viability of $(z, i_1)$ implies that $(z_1, i_1)$ is viable.  Again, by the definition of $W(z_1)$, we have $W(z_1) \supseteq W(H^{z_1, i_1})$, and thus $v \in W(z_1)$. Identical reasoning shows that $v \in W(z_2)$. 

      We have thus shown that $v \in A^{z} \cap W(z_1) \cap W(z_2)$.
 \end{proof}


 Consider the map $\vmap$. For $p \in \set{1, 2}$, let $\vmap_p$ denote the  restriction of $\vmap$ to $W(z_p)$.  
 Now consider the tuple $\bt = (\hat t_q, \tilde t_q)_{q \in [\kw]}$. 
 We define the tuples $\bt_1 = (\hat t_{1q}, \tilde t_{1q})_{q \in [\kw]}$ and $\bt_2 = (\hat t_{2q}, \tilde t_{2q})_{q \in [\kw]}$ as follows: For $q \in [\kw]$, we define  $(\hat t_{1q}, \tilde t_{1q}) = (\hat t_q, \tilde t_q)$ if $V(M^{z_1, q}) \neq \emptyset$ and $(\hat t_{1q}, \tilde t_{1q}) = (\nnn, \nnn)$ otherwise; similarly, $(\hat t_{2q}, \tilde t_{2q}) = (\hat t_q, \tilde t_q)$ if $V(M^{z_2, q}) \neq \emptyset$ and $(\hat t_{2q}, \tilde t_{2q}) = (\nnn, \nnn)$ otherwise. 

 It is straightforward to show that if $\varphi = (\nu, \xi)$ is an $(z, \vmap, \bt)$-compatible homomorphism from $P^z$ to $\Gamma$, then for $p \in \set{1, 2}$, then $\varphi_p =(\nu_p, \xi)$ is an $(z_p, \vmap_p, \bt_p)$-compatible homomorphism from $P^{z_p}$ to $\Gamma$, where $\nu_p$ and $\xi_p$ respectively are the restrictions of $\nu$ and $\xi$ to $P^{z_p}$. We can also establish a reverse correspondence: For each $p \in \set{1, 2}$, suppose  $\varphi'_p = (\nu'_p,  \xi'_p)$ is an $(z_p, \vmap_p, \bt_p)$-compatible homomorphism from $P^{z_p}$ to $\Gamma$. Then $\varphi' = (\nu', \xi')$ is an $(z, \vmap, \bt)$-compatible homomorphism, where $\fn{\nu'}{V(H^z)}{V(G)}$ is the ``union'' of $\nu'_1$ and $\nu'_2$ defined as $\nu'(v) = \nu'_p(v)$ if $v \in V(H^{z_p})$, and similarly $\fn{\xi'}{E(H^{z})}{E(G)}$ is the ``union'' of $\xi'_1$ and $\xi'_2$ defined as $\xi'(e) = \xi'_p(e)$ if $e \in E(H^{z_p})$. Claim~\ref{claim:union} ensures that $\nu'$ is well-defined; in particular, if the vertex $v$ belongs to both $V(H^{z_1})$ and $V(H^{z_2})$, then we have $v \in W(z) \cap W(z_1) \cap W(z_2)$, and hence by condition (CC1) and the definitions of $\vmap_1$ and $\vmap_2$, we have $\nu(v) = \vmap(v) = \vmap_1(v) = \vmap_2(v)$. As  $E(H^{z_1}) = V(M^{z_1})$ and $E(H^{z_2}) = V(M^{z_2})$ are disjoint, $\xi'$ is also well-defined. 
 In short, we can show that there is a bijection between the sets $\comhom(z, \vmap, \bt)$ and $\comhom(z_1, \vmap_1, \bt_1) \times \comhom(z_2, \vmap_2, \bt_2)$. 
 We thus have 
 \[
 \cwDP[z, \vmap, \bt] = \cwDP[z_1, \vmap_1, \bt_1] \cdot \cwDP[z_2, \vmap_2, \bt_2]. 
 \]

 \subparagraph*{Relabeling operation.} Suppose $z$ corresponds to the operation $\relabel_{i, j}(\cdot)$, and let $z'$ be the unique child of $z$.  That is, at node $z$, all elements of $M^{z'}$ with label $i$ are  relabelled with label $j$; in particular, the elements of $V(M^{z', i})$ received label $j$ at $z$. Hence, $V(M^{z, i}) = \emptyset$ and $V(M^{z, j}) = V(M^{z', j}) \uplus V(M^{z', i})$. This is the only difference between $(M^z, \mu^z)$ and $(M^{z'}, \mu^{z'})$.  But notice that the patterns $P^{z}$ and $P^{z'}$ are the same, i.e., $P^{z} = P^{z'}$. We will argue that every $(z, \vmap, \bt)$-compatible homomorphism is an $(z', \vmap', \bt')$-compatible homomorphism for some $(\vmap', \bt') \in \maps(z') \times \tms(z')$; we  need to carefully identify these pairs $(\vmap', \bt')$, and the rest of the analysis is aimed at that. We first prove the following claim. 
 \begin{claim}\label{claim:W(z)}
     $W(z) \subseteq W(z')$. 
 \end{claim}
 \begin{claimproof}
     Recall first that $W(z) = \bigcup_{(z, q)} W(H^{z, q})$, where the union is over all \emph{viable} $(z, q)$, and $W(z') = \bigcup_{(z', q)} W(H^{z'}, q)$, where the union is over all \emph{viable} $(z', q)$. 
     To prove the claim, consider a vertex $v \in W(z)$. Then there exists an index $p \in [\kw]$ such that $(z, p)$ is viable and $v \in W(H^{z, p})$; in particular, $v$ is a vertex of the viable graph $H^{z, p}$, the subgraph of $H$ induced by the subset $V(M^{z, p})$ of $E(H)$. Notice that $p \neq i$ as $V(M^{z, i}) = \emptyset$. Now, recall that the relableing operation did not affect $V(M^{z', q})$ for any $q \in [\kw] \setminus \set{i, j}$. So if $p \notin \set{i, j}$, then $V(M^{z, p}) = V(M^{z', p})$, and consequently $H^{z, p} = H^{z', p}$, and hence $H^{z', p}$ is viable and in particular, $W(z) = W(z')$. Thus $v \in W(z')$ in this case. 
     
     Now suppose $p = j$. Then $v \in W(H^{z, j})$ and in particular, the graph $H^{z, j}$ is viable. The viability of $H^{z, j}$, along with the fact that  $v \in W(H^{z, j})$ implies that the graph $H^{z, j}$ has an independent vertex cover, say $X \subseteq V(H^{z, j})$ such that $\card{X} \leq 2$ and $v \in X$.   Recall that $V(M^{z, j}) = V(M^{z', i}) \uplus V(M^{z', j})$; that is, $V(M^{z, j})$ is a superset of both $V(M^{z', i})$ and  $V(M^{z', j})$, and hence $H^{z, j}$ is a supergraph of both $H^{z', i}$ and $H^{z', j}$; in particular, $H^{z, j}$ is the precisely the graph with vertex set $V(H^{z', i}) \cup V(H^{z', j})$ and edge set $E(H^{z', i}) \uplus E(H^{z', j})$. Thus $v$ is a vertex of at least one of the graphs $H^{z', i}$ and $H^{z', j}$. 
     Also, as the supergraph $H^{z, j}$ is viable, we can conclude that both $H^{z', i}$ and $H^{z', j}$ are also viable. In particular, $X \cap V(H^{z', i})$ is an independent vertex cover of size at most 2 for $H^{z', i}$, and $X \cap V(H^{z', j})$ is an independent vertex cover of size at most 2 for $H^{z', j}$. Now, as $v \in X$ and $v \in V(H^{z', i}) \cup V(H^{z', j})$, we can conclude that either $v$ belongs to the independent vertex cover $X \cap V(H^{z', i})$ in which case $v \in W(H^{z', i})$, or $v$ belongs to the independent vertex cover $X \cap V(H^{z', j})$ in which case $v \in W(H^{z', j})$. In either case, $v \in W(z')$. 

     We have thus shown that $W(z) \subseteq W(z')$. 
     \end{claimproof}

     Now in the next two claims, we explore how $\comhom(z, \vmap, \bt)$ and $\comhom(z', \vmap', \bt')$ are related. Let $\bt = (\hat t_q, \tilde t_q)_{q \in [\kw]}$. 
     
     \begin{claim}\label{claim:tuples}  
         Consider a pair  $(\vmap', \bt') \in \maps(z') \times \tms(z')$, where $\bt' = (\hat s_q, \tilde s_q)_{q \in [\kw]}$. Then none of $\hat s_i, \hat s_j, \hat t_j, \tilde s_i, \tilde s_j, \tilde t_j$ is $\nnn$. Also, if $\comhom(z, \vmap, \bt) \cap \comhom(z', \vmap', \bt') \neq \emptyset$, then it holds that $(\hat t_q, \tilde t_q) = (\hat s_q, \tilde s_q)$ for every $q \in [\kw] \setminus \set{i, j}$; and $\hat t_j = \min\set{\hat s_i, \hat s_j}$ and $\tilde t_j = \max\set{\tilde s_i, \tilde s_j}$.  
     \end{claim}
     \begin{claimproof}
            First of all, as $z$ corresponds to the operation $\relabel_{i, j}(\cdot)$, and $z'$ is the unique child of $z$, both  $V(M^{z', i})$ and $V(M^{z', j})$ are non-empty (because the relabel operation is not redundant; see Remark~\ref{rem:assumptions}), and consequently $V(M^{z, j}) = V(M^{z', i}) \uplus V(M^{z', j})$ is non-empty. These facts, along with the definitions of $\bt$ and $\bt'$, together imply the first  assertion in the statement of the claim.
            
           Now, to prove the second assertion, assume that $\comhom(z, \vmap, \bt) \cap \comhom(z', \vmap', \bt') \neq \emptyset$. 
           Fix a  homomorphism $\varphi = (\nu, \xi) \in \comhom(z', \vmap', \bt') \cap \comhom(z, \vmap, \bt)$. That is, $\varphi$ is both  $(z', \vmap', \bt')$-compatible and $(z, \vmap, \bt)$-compatible.
           First, since $\varphi$ is $(z, \vmap, \bt)$-compatible, conditions (CC2) and (CC3) together imply that for every $q \in [\kw]$ with $V(M^{z, q}) \neq \emptyset$, we have $\hat t_q = \min_{e \in V(M^{z, q})} \tau(\xi(e))$. Similarly, from the $(z', \vmap', \bt')$-compatibility of $\varphi$, we also have $\hat s_q = \min_{e \in V(M^{z', q})} \tau(\xi(e))$. Now, since $V(M^{z, q}) = V(M^{z', q})$ for every $q \in [\kw] \setminus \set{i, j}$, we get $\hat s_q = \hat t_q$ for $q \in [\kw] \setminus \set{i, j}$. Recall that $V(M^{z, j}) = V(M^{z', i}) \uplus V(M^{z', j})$. 
        Hence, as $V(M^{z, j}) \supseteq V(M^{z', i})$, we have $\tilde t_j = \min_{e \in V(M^{z, j})} \tau(\xi(e)) \leq \min_{e \in V(M^{z', i})} \tau(\xi(e)) = \hat s_i$; by similar reasoning, we get $\hat t_j \leq \hat s_j$. Thus $\hat t_j \leq \hat s_i$ and $\hat t_j \leq \hat s_j$, which together imply that $\hat t_j \leq \min\set{\hat s_i, \hat s_j}$. Also, since $\varphi$ is $(z, \vmap, \bt)$-compatible, by condition (CC2), there exists $\hat e_j \in V(M^{z, j})$ with $\hat t_j = \tau(\xi(\hat e_j))$. Now we use the fact that $\varphi$ is $(z', \vmap', \bt')$-compatible, specifically condition (CC3): As $\hat e_j \in V(M^{z, j}) = V(M^{z', i}) \uplus V(M^{z', j})$, we either have $\hat e_j \in V(M^{z', i})$, in which case $\hat s_i \leq \tau(\xi(\hat e_j)) = \hat t_j$, or we have $\hat e_j \in V(M^{z', j})$, in which case $\hat s_j \leq \tau(\xi(\hat e_j)) = \hat t_j$. In either case $\min\set{\hat s_i, \hat s_j} \leq \hat t_j$. We have thus shown that $\hat t_j = \min\set{\hat s_i, \hat s_j}$.  
        Symmetric argumentswill show that $\tilde s_q = \tilde t_q$ for $q \in [\kw] \setminus \set{i, j}$ and $\tilde s_j \leq \tilde t_j$. 
     \end{claimproof}

 \begin{claim}\label{claim:all-or-none}
     Consider a pair $(\vmap', \bt') \in \maps(z') \times \tms(z')$, where $\bt' = (\hat s_q, \tilde s_q)_{q \in [\kw]}$. Then either no $(z', \vmap', \bt')$-compatible homomorphism is also $(z, \vmap, \bt)$-compatible (i.e., $\comhom(z', \vmap', \bt')$ and  $\comhom(z, \vmap, \bt)$ are disjoint), or every $(z', \vmap', \bt')$-compatible homomorphism is also $(z, \vmap, \bt)$-compatible  (i.e., $\comhom(z', \vmap', \bt')$ is a subset of  $ \comhom(z, \vmap, \bt)$). Moreover, if $\comhom(z', \vmap', \bt')$ is a subset of  $ \comhom(z, \vmap, \bt)$, then the map $\vmap$ is the restriction of $\vmap'$ to $W(z)$, or equivalently $\vmap'$ is an extension of $\vmap$.  
 \end{claim}
 \begin{claimproof}
     Assume that $\comhom(z', \vmap', \bt')$ and  $\comhom(z, \vmap, \bt)$ are not disjoint, i.e., $\comhom(z', \vmap', \bt') \cap \comhom(z, \vmap, \bt) \neq  \emptyset$. Fix a   homomorphism $\varphi = (\nu, \xi) \in \comhom(z', \vmap', \bt') \cap \comhom(z, \vmap, \bt)$. That is, $\varphi$ is both  $(z', \vmap', \bt')$-compatible and $(z, \vmap, \bt)$-compatible. 
    Let us now see that $\comhom(z', \vmap', \bt') \subseteq \comhom(z, \vmap, \bt)$, and to that end consider $\varphi' = (\nu', \xi') \in \comhom(z', \vmap', \bt')$. We will show that $\varphi'$ satisfies the three conditions required to be an $(z, \vmap, \bt)$-compatible homomorphism. 
    \begin{itemize}
        \item Since both $\varphi$ and $\varphi'$ are $(z', \vmap', \bt')$-compatible, by condition (CC1) in Definition~\ref{def:compatible}, we have $\nu|_{W(z')} = \vmap'$ and $\nu'|_{W(z')} = \vmap'$. That is, the maps $\nu$ and $\nu'$ agree on the set $W(z')$. By Claim~\ref{claim:W(z)}, we have $W(z) \subseteq W(z')$, and  we can therefore conclude that $\nu$ and $\nu'$ agree on the set $W(z)$. Now as $\varphi$ is also $(z, \vmap, \bt)$-compatible, by (CC1), we further have $\nu|_{W(z)} = \vmap$, which implies that $\nu'|_{W(z)} = \vmap$. Thus $\varphi'$ satisfies condition (CC1). Notice that this also shows that $\vmap = \nu|_{W(z)} = \vmap'|_{W(z)}$, which proves the second assertion in the statement of the claim. 

        \item We now show that $\varphi'$ satisfies condition (CC2). Consider $q \notin \set{i, j}$ with $V(M^{z, q}) \neq \emptyset$. By Claim~\ref{claim:tuples}, we have $(\hat t_q, \tilde t_q) = (\hat s_q, \tilde s_q)$.  
        As $\varphi'$ is $(z', \vmap', \bt')$-compatible, by condition (CC2), there exist $\hat e_q, \tilde e_q \in V(M^{z', q})$ such that $\tau(\xi'(\hat e_q)) = \hat s_q$ and $\tau(\xi'(\tilde e_q)) = \tilde s_q$. Now, since $V(M^{z', q}) = V(M^{z, q})$, and as $(\hat t_q, \tilde t_q) = (\hat s_q, \tilde s_q)$, it follows that $\hat e_q, \tilde e_q \in V(M^{z, q})$, and $\tau(\xi'(\hat e_q)) = \hat t_q$ and $\tau(\xi'(\tilde e_q)) = \tilde t_q$. Thus (CC2) holds for $q \notin \set{i, j}$. 
        Note that (CC2) vacuously holds for $q = i$ as  $V(M^{z, i}) = \emptyset$. Now we consider $q = j$. Again, by Claim~\ref{claim:tuples}, we have $\hat t_j = \min\set{\hat s_i, \hat s_j}$. Hence either $\hat t_j = \hat s_i$ or $\hat t_j = \hat s_j$. Suppose $\hat t_j = \hat s_i$; the other case is identical. Now, we invoke the $(z', \vmap', \bt')$-compatibility of $\varphi'$, specifically condition (CC2), which implies that there exists $\hat e \in V(M^{z', i})$ such that $\hat s_i = \tau(\xi'(\hat e))$. As $V(M^{z, j}) \supseteq V(M^{z', i})$ and as $\hat t_j = \hat s_i$, we have $\hat e \in V(M^{z, j})$ and $\tau(\xi'(\hat e)) = \hat t_j$. Identical reasoning will show that there exists $\tilde e \in V(M^{z, j})$ with $\tau(\xi'(\tilde e)) = \tilde t_j$. Thus (CC2) holds for $q = j$ as well. 
        \item We now show that $\varphi'$ satisfies (CC3). Consider $q \in [\kw]$ with $V(M^{z, q}) \neq \emptyset$. Then Claim~\ref{claim:tuples} implies that $\hat t_q \leq \hat s_q$ and $\tilde t_q \leq \tilde s_q$. Now, as $\varphi'$ is $(z', \vmap', \bt')$-compatible, by condition (CC3), for every $e \in V(M^{z, q})$, we have $\hat s_q \leq \tau(\xi'(e)) \leq \tilde s_q$ which implies that $\hat t_q \leq \hat s_q \leq \tau(\xi'(e)) \leq \tilde s_q \leq \tilde t_q$. Thus condition (CC3) holds as well. 
    \end{itemize}
    We have thus shown that $\varphi'$ is $(z, \vmap, \bt)$-compatible. 
 \end{claimproof}

 Claim~\ref{claim:all-or-none} immediately leads to the following result. 
 \begin{claim}\label{claim:calI}
    There exists a subset $\ca{I}$ of $\maps(z') \times \tms(z')$ such that \[
    \comhom(z, \vmap, \bt) = \biguplus_{(\vmap', \bt') \in \ca{I}} \comhom(z', \vmap', \bt').
    \]
    In particular, $\ca{I}$ is precisely the set of all $(\vmap', \bt')$ in $\maps(z') \times \tms(z')$ such that $\comhom(z', \vmap', \bt') \subseteq \comhom(z, \vmap, \bt)$. 
 \end{claim}
 \begin{claimproof}
    Consider the set $\ca{I}$ as defined in the second assertion in the statement of the claim. By the definition of $\ca{I}$, we have $\bigcup_{(\vmap', \bt') \in \ca{I}} \comhom(z', \vmap', \bt') \subseteq \comhom(z, \vmap, \bt)$. Now the fact that the sets $\comhom(z', \vmap', \bt')$ are equivalence classes under the relation $\sim_{z'}$ imply that these sets are also pairwise-disjoint, and hence the union is indeed a disjoint union, i.e., $\comhom(z, \vmap, \bt) \subseteq \biguplus_{(\vmap', \bt') \in \ca{I}} \comhom(z', \vmap', \bt')$.

    Now, to prove that $\comhom(z, \vmap, \bt) \subseteq \biguplus_{(\vmap', \bt') \in \ca{I}} \comhom(z', \vmap', \bt')$, consider $\varphi \in \comhom(z, \vmap, \bt)$. As $P^z = P^{z'}$, every homomorphism from $P^{z}$ to $\Gamma$ is also a homomorphism from $P^{z'}$ to $\Gamma$; in particular, $\varphi$ is a homomorphism form $P^{z'}$ to $\Gamma$. 
    Recall also that by Lemma~\ref{lem:unique}, every homomorphism from $P^{z'}$ to $\Gamma$ is an $(z', \vmap', \bt')$-compatible homomorphism for exactly one pair $(\vmap', \bt') \in \maps(z') \times \tms(z')$. Let $(\vmap'', \bt'')$ be the unique pair in $\maps(z') \times \tms(z')$ such that $\varphi$ is $(z', \vmap'', \bt'')$-compatible. Thus $\varphi$ is both $(z, \vmap, \bt)$-compatible and $(z', \vmap'', \bt'')$-compatible, which implies that $\comhom(z, \vmap, \bt) \cap \comhom(z', \vmap'', \bt'') \neq \emptyset$. Then, by Claim~\ref{claim:all-or-none}, we can conclude that $\comhom(z', \vmap'', \bt'') \subseteq \comhom(z, \vmap, \bt)$, which implies that $(\vmap'', \bt'') \in \ca{I}$. This shows that $\comhom(z, \vmap, \bt) \subseteq \biguplus_{(\vmap', \bt') \in \ca{I}} \comhom(z', \vmap', \bt')$.
    \end{claimproof}

 Claim~\ref{claim:calI} shows  that to compute $\card{\comhom(z, \vmap, \bt)}$, we only need to identify the elements of the set $\ca{I}$ and sum up $\card{\comhom(z', \vmap', \bt')}$ for all $(\vmap', \bt') \in \ca{I}$. We do this below. In particular, we will define nine pairwise-disjoint sets $\ca{I}^1, \ca{I}^2,\ldots, \ca{I}^9 \subseteq \maps(z') \times \tms(z')$, and argue that $\ca{I}$ is the union of these nine sets; the pairwise-disjointness of the sets will ensure that we will not over-count when summing up $\card{\comhom(z', \vmap', \bt')}$ over the elements $(\vmap', \bt')$ of $\ca{I}$.  
    
 Before explicitly identifying the elements of $\ca{I}$, let us flesh out some of the properties of $\ca{I}$, and to that end consider a hypothetical $(\vmap', \bt') \in \ca{I}$. 

 Let us first focus on the map $\vmap'$. Recall that $\vmap'$ must be a map from $W(z')$ to $V(G)$. 
 Now, for a homomorphism $\varphi$ from $P^{z'} = P^z$ to $\Gamma$ to be both $(z, \vmap, \bt)$-compatible and $(z', \vmap', \bt')$-compatible, by the second assertion of Claim~\ref{claim:all-or-none}, we must have $\vmap = \vmap'|_{W(z)}$, i.e., $\vmap$ must be a restriction of $\vmap'$. With this in mind, we define the following set: 
 Let $\ext(\vmap, z')$ denote the set of all maps $\fn{\alpha'}{W(z')}{V(G)}$ that extend $\vmap$, or equivalently $\vmap'|_{W(z)} = \vmap$. Notice that $\ext(\vmap, z') \subseteq \maps(z')$. 
 
 Let us now focus on the tuple $\bt' \in \tms(z')$ for $(\vmap', \bt') \in \ca{I}$. 
 Recall that by Claim~\ref{claim:tuples}, we have $\hat t_j, \tilde t_j \neq \nnn$, and in particular, $\hat t_j, \tilde t_j \in [\lt]$. Now, consider an $(z, \vmap, \bt)$-compatible homomorphism $\varphi = (\nu, \xi)$. By definition, there exist $\hat e_j, \tilde e_j \in V(M^{z, j})$ such that $\hat t_j = \tau(\xi(\hat e_j))  = \min_{e \in V(M^{z, j})} \tau(\xi(e))$, and $\tilde t_j = \tau(\xi(\tilde e_j))  = \max_{e \in V(M^{z, j})} \tau(\xi(e))$; the edges $\hat e_j$ and $\tilde e_j$ are ``witnesses'' for the earliest and latest time-steps to which elements in $V(M^{z, j})$ are mapped by $\xi$. But notice that there may be more than one witness for each of $\hat t_j$ and $\tilde t_j$. Now, as $V(M^{z, j}) = V(M^{z', i}) \uplus V(M^{z', j})$, we have to  consider different possibilities depending on which of the two sets $V(M^{z', i})$ and $V(M^{z', j})$ contains a witness for $\hat t_j$ and a witness for $\tilde t_j$. Notice that it may also be the case that each of $\hat t_j$ and $\tilde t_j$ has witnesses in both $V(M^{z', i})$ and $V(M^{z', j})$; that is, there may exist $e_i \in V(M^{z', i})$ and $e_j \in V(M^{z', j})$ with $\tau(\xi(e_i)) = \tau(\xi(e_j)) = \hat t_j$. Thus, we will have to consider all possible cases where each of $\hat t_j$ and $\tilde t_j$ has witnesses in $V(M^{z', i})$ alone, in $V(M^{z', j})$ alone, and in both $V(M^{z', i})$ and $V(M^{z', j})$. 
 
 Before formally accounting for all these cases, let us consider, as an example, the case when $\hat t_j$ has a witness in $V(M^{z', i})$ alone and $\tilde t_j$ has a witness in both $V(M^{z', i})$ and $V(M^{z', j})$. Let $\bt' = (\hat s_q, \tilde s_q)_{q \in [\kw]}$, and consider a homomorphism $\varphi = (\nu, \xi)$ that is both $(z', \vmap', \bt')$-compatible and $(z, \vmap, \bt)$-compatible.  
 Now, as $\hat t_j$ has a witness in $V(M^{z', i})$, there exists $\hat e_j \in V(M^{z', i})$ such that $\tau(\xi(\hat e_j)) = \hat t_j = \min_{e \in V(M^{z', i})} \tau(\xi(e))$. We must therefore have $\hat s_i = \hat t_j$.  Also, as $\hat t_j$ has no witnesses in $V(M^{z', j})$, there does not exist any $e \in V(M^{z', j})$ with $\tau(\xi(e)) = \hat t_j$; in particular, we must have $\min_{e \in V(M^{z', j})} \tau(\xi(e)) \geq \hat t_j + 1$, and therefore we must have $\hat s_j \geq \hat t_j + 1$. By similar reasoning, as $\tilde t_j$ has witnesses in both $V(M^{z', i})$ and $V(M^{z', j})$, there exist $\tilde e_{j, i} \in V(M^{z', i})$ and $\tilde e_{j, j} \in V(M^{z', j})$ with $\tau(\xi(\tilde e_{j, i})) = \tau(\xi(\tilde e_{j, j})) = \tilde t_j = \max_{e \in V(M^{z', i})} \tau(\xi(e)) = \max_{e \in V(M^{z', j})} \tau(\xi(e))$. We  must therefore have $\tilde s_i = \tilde s_j = \tilde t_j$. 
 
 With the above discussion as well as Claim~\ref{claim:tuples} in mind, we first define the following set:
 \[
 \ca{S}^{\bt} = \set{(\hat s_q, \tilde s_q)_{q \in [\kw]} \in \tms(z') ~|~ \hat s_q = \hat t_q \text{ and } \tilde s_q = \tilde t_q \text{ for every } q \in [k] \setminus \set{i, j}}. 
 \]
 Notice that by definition, $\ca{S}^{\bt} \subseteq \tms(z')$, and by the definition of $\tms(z')$, for every $q \in [\kw]$, either $(\hat s_q, \tilde s_q) = (\nnn, \nnn)$ or $\hat s_q \leq \tilde s_q$.  We now define the following nine subsets of $\ca{S}^{\bt}$, each of which corresponds to a possible case depending on which of the two sets $V(M^{z', i})$ and $V(M^{z', j})$ contain witnesses for $\hat t_j$ and $\tilde t_j$. We label the definitions of these sets with ``tags'' that help us identify which set corresponds to which case. For example, if $V(M^{z', i})$ (resp. $V(M^{z', j})$) alone contains a witness for $\hat t_j$, then we use the tag $i$-min (resp. $j$-min), and if both $V(M^{z', i})$ and $V(M^{z', j})$ contain witnesses for $\hat t_j$, we use the tag $ij$-min. Similarly, we use the tags $i$-max, $j$-max, and $ij$-max when dealing with $\tilde t_j$.

 \begin{equation*}
     \ca{S}^{\bt, 1} = \set{ (\hat s_q, \tilde s_q)_{q \in [k]}\in \ca{S}^{\bt} ~~~ \Big\rvert ~~~ 
     \begin{matrix}
     \hat s_i = \hat t_j \text{ and } \tilde s_i = \tilde t_j, \text{ and } \\ 
     
      \hat t_j + 1 \leq  \hat s_j \leq \tilde s_j \leq \tilde t_j - 1
     \end{matrix}} \tag{$i$-min, $i$-max}
 \end{equation*}

 \begin{equation*}
     \ca{S}^{\bt, 2} = \set{(\hat s_q, \tilde s_q)_{q \in [k]}\in \ca{S}^{\bt} ~~~ \Big\rvert ~~~ 
     \begin{matrix}
     \hat s_i = \hat t_j \text{ and } \tilde s_i \leq \tilde t_j - 1, \text{ and } \\ 
     
     \hat s_j \geq \hat t_j + 1 \text{ and } \tilde s_j = \tilde t_j
     \end{matrix}} \tag{$i$-min, $j$-max}
 \end{equation*}


 \begin{equation*}
     \ca{S}^{\bt, 3} = \set{ (\hat s_q, \tilde s_q)_{q \in [k]}\in \ca{S}^{\bt} ~~~ \Big\rvert ~~~ 
     \begin{matrix}
     \hat s_i \geq \hat t_j + 1 \text{ and } \tilde s_i = \tilde t_j, \text{ and } \\
     
      \hat s_j = \hat t_j \text{ and } \tilde s_j \leq \tilde t_j - 1
     \end{matrix}} \tag{$i$-max, $j$-min}
 \end{equation*}

  \begin{equation*}
     \ca{S}^{\bt, 4} = \set{(\hat s_q, \tilde s_q)_{q \in [k]}\in \ca{S}^{\bt} ~~~ \Big\rvert ~~~ 
     \begin{matrix}
      \hat t_j + 1 \leq  \hat s_i \leq \tilde s_i \leq \tilde t_j - 1, \text{ and } \\
      
      \hat s_j = \hat t_j \text{ and } \tilde s_j = \tilde t_j 
     \end{matrix}} \tag{$j$-min, $j$-max}
 \end{equation*}

 \begin{equation*}
     \ca{S}^{\bt, 5} = \set{(\hat s_q, \tilde s_q)_{q \in [k]}\in \ca{S}^{\bt} ~~~ \Big\rvert ~~~ 
     \begin{matrix}
      \hat s_i = \hat t_j \text{ and } \tilde s_i = \tilde t_j, \text{ and } \\
      
      \hat s_j = \hat t_j \text{ and } \tilde s_j \leq \tilde t_j - 1
     \end{matrix}} \tag{$ij$-min, $i$-max}
 \end{equation*}

 \begin{equation}
     \ca{S}^{\bt, 6} = \set{(\hat s_q, \tilde s_q)_{q \in [k]}\in \ca{S}^{\bt} ~~~ \Big\rvert ~~~ 
     \begin{matrix}
      \hat s_i = \hat t_j \text{ and } \tilde s_i \leq \tilde t_j - 1, \text{ and } \\
      
      \hat s_j = \hat t_j \text{ and } \tilde s_j = \tilde t_j 
     \end{matrix}} \tag{$ij$-min, $j$-max}
 \end{equation}

 \begin{equation*}
     \ca{S}^{\bt, 7} = \set{(\hat s_q, \tilde s_q)_{q \in [k]}\in \ca{S}^{\bt} ~~~ \Big\rvert ~~~ 
     \begin{matrix}
      \hat s_i = \hat t_j \text{ and } \tilde s_i = \tilde t_j, \text{ and } \\
      
      \hat s_j \geq \hat t_j + 1 \text{ and } \tilde s_j = \tilde t_j 
     \end{matrix}} \tag{$i$-min, $ij$-max}
 \end{equation*}

 \begin{equation*}
     \ca{S}^{\bt, 8} = \set{(\hat s_q, \tilde s_q)_{q \in [k]}\in \ca{S}^{\bt} ~~~ \Big\rvert ~~~ 
     \begin{matrix}
      \hat s_i \geq \hat t_j + 1 \text{ and } \tilde s_i = \tilde t_j, \text{ and } \\
      
      \hat s_j = \hat t_j \text{ and } \tilde s_j = \tilde t_j 
     \end{matrix}} \tag{$j$-min, $ij$-max}
 \end{equation*}

 \begin{equation*}
     \ca{S}^{\bt, 9} = \set{(\hat s_q, \tilde s_q)_{q \in [k]}\in \ca{S}^{\bt} ~~~ \Big\rvert ~~~ 
     \begin{matrix}
      \hat s_i = \hat t_j \text{ and } \tilde s_i = \tilde t_j, \text{ and } \\
      
      \hat s_j = \hat t_j \text{ and } \tilde s_j = \tilde t_j 
     \end{matrix}} \tag{$ij$-min, $ij$-max}
 \end{equation*}

  For $\ell \in \set{1, 2,\ldots, 9}$, let $\ca{I}^{\ell} = \ext(\vmap, z') ~\times~ \ca{S}^{\bt, \ell}$. It is straightforward to verify that the sets $\ca{S}^{\bt, 1}, \ca{S}^{\bt, 2},\ldots, \ca{S}^{\bt, 9}$, and therefore the sets $\ca{I}^1, \ca{I}^2,\ldots, \ca{I}^9$, are pairwise-disjoint. And that for $(\vmap', \bt') \in \ca{I}$, there exists exactly one index $\ell \in \set{1, 2,\ldots, 9}$ such that $(\vmap', \bt') \in \ca{I}^{\ell}$. 
  In other words, we have $\ca{I} = \biguplus_{\ell = 1}^9 \ca{I}^{\ell}$, and thus by Claim~\ref{claim:calI}, $\comhom(z, \vmap, \bt) = \biguplus_{(\vmap', \bt') \in \ca{I}} \comhom(z', \vmap', \bt') = \biguplus_{\ell = 1}^9 \biguplus_{(\vmap', \bt') \in \ca{I}^{\ell}} \comhom(z', \vmap', \bt')$. 
 We thus have 
 \[
 \cwDP[z, \vmap, \bt] = \sum_{(\vmap', \bt') \in \ca{I}} \cwDP[z', \vmap', \bt'] = \sum_{\ell = 1}^9 \sum_{(\vmap', \bt') \in \ca{I}^{\ell}} \cwDP[z', \vmap', \bt'].
 \]

 \subparagraph*{Edge-join operation.} Suppose $z$ corresponds to the operation $\ejoin_{i, j}(\cdot)$ for distinct $i, j \in [\kw]$, and let $z'$ be the unique child of $z$. Notice that $P^{z} = P^{z'}$; that is, the addition of the line-graph-edges through the $\ejoin_{i, j}(\cdot)$ operation does not change the pattern. Notice also that  $V(M^{z, q}) = V(M^{z', q})$, and therefore $H^{z, q} = H^{z', q}$ for every $q \in [\kw]$. We thus have $W(z) = W(z')$, ~ $\maps(z) = \maps(z')$, ~  $\tms(z)  = \tms(z')$, and consequently $\comhom(z, \vmap, \bt) = \comhom(z', \vmap, \bt)$. Therefore,   
 \[
 \cwDP[z, \vmap, \bt] = \cwDP[z', \vmap, \bt]. 
 \]
 We note that while it may seem at this point  that the insertion of line-graph-edges through the $\ejoin_{i, j}(\cdot)$ operation does not matter for the computation of $\cwDP[z, \vmap, \bt]$, that is not the case. We used the existence of line-graph-edges while computing the DP table entries corresponding to  the union operation; specifically, Claim~\ref{claim:union}  relies on the fact that appropriate line-graph-edges will be inserted at a future node.  

 \subparagraph*{Arc-join operation.} Suppose $z$ corresponds to the operation $\ajoin_{i, j}(\cdot)$ for distinct $i, j \in [\kw]$, and let $z'$ be the unique child of $z$. Notice first that  $V(M^{z, q}) = V(M^{z', q})$, and therefore $H^{z, q} = H^{z', q}$ for every $q \in [\kw]$. We thus have $W(z) = W(z')$, ~ $\maps(z) = \maps(z')$, ~  $\tms(z)  = \tms(z')$. Notice now that the only difference between the patterns $P^{z}$ and $P^{z'}$ is that $P^{z}$ contains the additional constraints $\vartheta(e_i) \leq_{R} \vartheta(e_j)$ (or more precisely $\vartheta^z(e_i) \leq_{R^z} \vartheta^z(e_j))$ for every $e_i \in V(M^{z, i})$ and $e_j \in V(M^{z, j})$. Thus every $(z, \vmap, \bt)$-compatible homomorphism is trivially $(z', \vmap, \bt)$-compatible. 
 
 Now, consider an $(z', \vmap, \bt)$-compatible homomorphism $\varphi =(\nu, \xi)$. Let us see that $\varphi =(\nu, \xi)$ is $(z, \vmap, \bt)$-compatible if only if $\tilde t_i \leq \hat t_j$. Suppose $\varphi =(\nu, \xi)$ is $(z, \vmap, \bt)$-compatible. 
 Then by condition (CC2) Definition~\ref{def:compatible}, there must exist $\tilde e_i \in V(M^{z, i})$ with $\tau(\xi(\tilde e_i)) = \tilde t_i$, and there must exist $\hat e_j \in V(M^{z, j})$ with $\tau(\xi(\hat e_j)) = \hat t_j$. And as $\vartheta(e_i) \leq_R \hat \vartheta(e_j)$, by condition (CC3), we have $\tau(\xi(\tilde e_i)) \leq \tau(\xi(\hat e_j))$, which implies that $\tilde t_i = \tau(\xi(\tilde e_i)) \leq \tau(\xi(\hat e_j)) = \hat t_j$. 
 Conversely, suppose that $\tilde t_i \leq \hat t_j$. To show that $\varphi = (\nu, \xi)$ is $(z, \vmap, \bt)$-compatible, notice that we only need to show that $\tau(\xi(e_i)) \leq \tau(\xi(e_j))$ for every $e_i \in V(M^{z, i})$ and $e_j \in V(M^{z, j})$. By condition (CC3), for every $e_i \in V(M^{z, i})$, we have $\tau(\xi(e_i)) \leq \tilde t_i$, and for every  $e_j \in V(M^{z, j})$, we have $\hat t_j \leq \tau(\xi(e_j))$, and thus $\tau(\xi(e_i)) \leq \tilde t_i \leq \hat t_j \leq \tau(\xi(e_j))$. 

 These arguments show that $\comhom(z, \vmap, \bt) = \comhom(z', \vmap, \bt)$ if $\tilde t_i \leq \hat t_j$, and $\comhom(z, \vmap, \bt) = \emptyset$ otherwise. 
 We thus have
 \[
 \cwDP[z, \vmap, \bt] = \begin{cases} 
                        \cwDP[z', \vmap, \bt], ~~~~~~ \text{ if } \tilde t_i \leq \hat t_j,\text{ and } \\
                        0, ~~~~~~~~~ ~~~~~~~~~~~~~ ~ \text{ otherwise.}
                       \end{cases}
 \]

This completes the description of the DP. We are now ready to complete the proof of Lemma~\ref{lem:dp}.
\begin{proof}[Proof of Lemma~\ref{lem:dp}]
   The preceding  discussion shows that our recurrences for computing $\cwDP[z, \vmap, \bt]$ are correct. Therefore, we can indeed compute $\card{\comhom(z, \vmap, \bt)}$ for every node $z \in V(T)$, $\vmap \in \maps(z)$ and $\bt \in \tms(z)$; in particular, we can compute $\card{\comhom(\hat z, \vmap, \bt)}$, where $\hat z$ is root of the tree $T$.  Now, recall that $P^{\hat z} = P$, which,  along with  Corollary~\ref{cor:unique-sum}, implies that $\counthoms{P}{\Gamma} = \sum_{(\vmap, \bt)} \card{\comhom(\hat z, \vmap, \bt)} = \sum_{(\vmap, \bt)} \cwDP[\hat z, \vmap, \bt]$, where the summation is over all pairs $(\vmap, \bt) \in \maps(\hat z) \times \tms(\hat z)$.  

   The algorithm is now straightforward. Given $\Gamma$, $P$ and $T$, we compute all the entries in the DP table, and output $\sum_{(\vmap, \bt)} \cwDP[\hat z, \vmap, \bt]$. 

   As for the running time, notice that the number of entries in the DP table is $\sum_{z \in V(T)} \card{\maps(z)} \cdot \card{\tms(z)}$. By Remark~\ref{rem:nodes-in-T}, we have $\card{V(T)} = \cO(\kw^2 \cdot \card{M})$, and note that the number of vertices in the \lgp\ $M$ of $P$ is $\card{P}^{\cO(1)}$. By  Observation~\ref{obs:bounds}, we have $\card{\maps(z)} \leq \card{G}^{4 \kw}$ and $\card{\tms(z)} \leq \lt^{2 \kw}$, where $\lt$ is the lifetime of $\Gamma$, i.e., $\lt = \max_{e \in E(G)} \tau(e)$. Thus the total number of entries in the DP table is $ \card{P}^{\cO(1)} \cdot \card{G}^{\cO(\kw)} \cdot \lt^{2 \kw} = \cwDPruntime$, where we write $\card{\Gamma}$ as a shorthand for $\card{G} \cdot \lt$. 
   
   Notice also that we can compute each entry of the DP table in time $\cwDPruntime$. In particular, we can compute the table entries corresponding to the introduce operation in time $\card{\Gamma}^{\cO(1)}$. To compute the entries corresponding to the union, the edge-join and arc-join operations, we only need to look up at most two previously computed table entries. As for the relabeling operation, notice that we need to look up all the entries $\cwDP[z', \vmap', \bt']$ with $(\vmap', \bt') \in \ca{I}^{\ell} = \ext(\vmap, z') ~\times~ \ca{S}^{\bt, \ell}$ for all $\ell \in \set{1, 2,\ldots, 9}$. Now, as $\ext(\vmap, z') \subseteq \maps(z')$, by Observation~\ref{obs:bounds}, we have $\card{\ext(\vmap, z')} \leq \card{G}^{4\kw}$. And notice that for each $\ell$, the elements of $\ca{S}^{\bt, \ell}$ are determined only by $(\hat s_i, \tilde s_i)$ and $(\hat s_j, \tilde s_j)$, and each of these pairs has $\lt^2$ choices; every other pair $(\hat s_q,\tilde s_q)$ has exactly one choice, namely, $(\hat s_q, \tilde s_q) = (\hat t_q, \tilde t_q)$. We thus have $\card{\ca{I}^{\ell}} \leq \card{G}^{4 \kw} \cdot \lt^{4} = \card{\Gamma}^{\cO(\kw)}$. In short, we can compute the table entries corresponding to the relabeling operation by looking up  $\card{\Gamma}^{\cO(\kw)}$ many previously computed entries. 

   Finally, to complete the proof of the lemma, note that we also need to show that each DP table entry is not too large. To that end, observe that for each node $z$ of $T$, the number of homomorphisms $\varphi = (\nu, \xi)$ from $P^z = (H^z, R^z, \vartheta^z)$ to $\Gamma = (G, \tau)$ is at most $\card{G}^{\card{H}} \cdot  \lt^{\card{E(H)}}$; this is because $\nu$ has at most $\card{G}^{\card{H}}$ choices as each vertex of $H$ could be mapped to any of the $\card{G}$ vertices, and for each choice of $\nu$, each edge $uv$ of $H$ could be mapped to one of the at most $\lt$ many parallel edges between $\nu(u)$ and $\nu(v)$. In short, the number of homomorphisms, and therefore each entry of the DP table, can be encoded using $\cO(\card{H}^2 \log(\card{G} \cdot \lt)) = \card{P}^{\cO(1)} \cdot \log \Gamma$ bits. 

   We can thus conclude that the overall running time of our algorithm is upper bounded by $\cwDPruntime$. 
\end{proof}

\begin{proof}[Proof of Theorem~\ref{thm:intro_main_algo}]
    Follows from Lemma~\ref{lem:dp} and the fact that $T^\sigma$ can be computed in time only depending on the temporal pattern $P$.
\end{proof}

\section{A Classification for Totally Ordered Patterns}
In this section we prove Theorem~\ref{thm:intro_main_classification}, which establishes a complete classification of the parameterised complexity of $\#\textsc{TemporalHom}_\mathrm{TO}(\mathcal{H})$, where $\ca{H}$ is a recursively enumerable class of static graphs. 
We discuss the upper bound in Section~\ref{sec:upper-bound}, and the lower bound in Section~\ref{sec:lower-bound}.
Throughout this section, we only deal with totally ordered temporal patterns, i.e., patterns $P = (H, R, \vartheta)$ in which $R$ is a totally ordered set and $\vartheta$ is a bijection. And we denote a totally ordered temporal pattern by $(H, \ord)$, where $\ord$ is a total order on $E(H)$.  

\subsection{Bounding Toadwidth via Semi-Induced Matching Number}\label{sec:upper-bound}

In this section, we show that $\#\textsc{TemporalHom}_\mathrm{TO}(\mathcal{H})$ is fixed-parameter tractable if the class of line graphs of $\ca{H}$ has bounded semi-induced matching number. To do this, we will prove that if the class of line graphs of $\ca{H}$ has bounded semi-induced matching number, then  for every $H \in \ca{H}$ and every total order $\ord$ on $E(H)$, the toadwidth of the temporal pattern $(H, \ord)$ is bounded. The proof then follows from Theorem~\ref{thm:intro_main_algo}. 

Let $H$ be a simple, undirected graph, and let $\ord$ be a total order on $E(H)$; thus $(H, \ord)$ is a temporal pattern.   Consider the line graph $L(H)$ of $H$, and the \lgp\ $\lplus((H, \ord))$; from here on, we will omit the second pair of parentheses, and simply write $\lplus(H, \ord)$ instead of $\lplus((H, \ord))$.   Let $M = \set{e_1 e'_1, e_2 e'_2,\ldots, e_{\simsize} e'_{\simsize}}$ be a semi-induced matching in $L(H)$. We say that $M$ is an \emph{order-respecting semi-induced matching} (orsim, for short) in $\lplus(H, \ord)$ if $e_i \ord e'_j$ for every $i, j \in [\simsize]$. 

In this section, we prove Lemma~\ref{lem:sim-cw-bound}, which says  that the toadwidth of  $(H, \ord)$ is upper bounded by a function of the semi-induced matching number of $L(H)$. Let us recall the formal statement of the lemma. 
\SIMBOUNDSTOADWIDTH*

To prove Lemma~\ref{lem:sim-cw-bound}, we prove the following stronger claim, which says that the toadwidth of $(H, \ord)$ is bounded if every orsim in $\lplus(H, \ord)$ has bounded size. Notice that every orsim in $\lplus(H, \ord)$ is, by definition, a semi-induced matching in $L(H)$. Hence if every semi-induced matching in $L(H)$ has bounded size, then every orsim in $\lplus(H, \ord)$ has bounded size for every total order $\ord$ on $E(H)$. And therefore, Lemma~\ref{lem:orsim-cw-bound} does imply Lemma~\ref{lem:sim-cw-bound}. 

\begin{lemma}\label{lem:orsim-cw-bound}
    Let $\simsize$ be a positive integer. 
    Let $H$ be a simple, undirected graph, and $\ord$ a total order on $E(H)$. If every orsim in $\lplus(H, \ord)$ has size at most $b$, then the toadwidth of $(H, \ord)$ is at most $\simcw$. 
\end{lemma}
\begin{proof}
     Assume that every orsim in $\lplus(H, \ord)$ has size at most $b$. Let $m = \card{E(H)}$. 
     To prove the lemma, we will construct a clique-expression of $\lplus(H, \ord)$ using at most $\simcw$ labels. And to do this, we will define a sequence of $m$ equivalence relations $\sim_1, \sim_2,\ldots, \sim_m$, and bound the number of equivalence classes under each relation by $\simcwone$. And we will ensure that the number of labels used in our clique-expression is at most $1 + \max_{i \in [m]} \text{number of eq. classes under } \sim_i$.  
     
     Let us first define the equivalence relations. Let $(e_1, e_2,\ldots, e_m)$ be the ordering of the edges of $H$ imposed by $\ord$. For each $i \in [m]$, let $E_i = \set{e_1, e_2, \ldots, e_i}$. Consider the line graph $L(H)$. For each $i \in [m]$ and $p \in [i]$, let $N_{> i}(e_p)$ be the set of neighbors of $e_p$ in the set $\set{e_{i + 1}, e_{i + 2},\ldots, e_m}$; by a neighbor of $e_p$, we mean neighbor in the line graph $L(H)$. That is, $N_{> i}(e_p) = \set{e_q ~|~ q > i \text{ and } e_p e_q \in E(L(H))}$. For each $i \in [m]$, we define an equivalence relation $\sim_i$ on the set $E_i$ as follows: For $p, q \in [i]$, we define $e_p \sim_i e_q$ if and only if $N_{> i}(e_p) = N_{> i}(e_q)$. That is, $e_p$ and $e_q$ are equivalent under $\sim_i$ if and only if they have the same neighbors in $\set{e_{i + 1}, e_{i + 2},\ldots, e_m}$. 
     It is straightforward to verify that $\sim_i$ is indeed an equivalence relation.  
     The following claim is immediate from the definition of $\sim_i$, and therefore we state it without proof.
     \begin{claim}\
         For $i, j \in [m]$ with $i \leq j$, and $e, e' \in E_i$, if $e \sim_i e'$ then $e \sim_j e'$. 
     \end{claim} 
     
     As $\sim_i$ is an equivalence relation, $\sim_i$ partitions $E_i$ into equivalence classes. We will prove that the number of equivalence classes under each $\sim_i$ is at most $\simcwone$, which we will then use to bound the cliquewidth of $\lplus(H, \ord)$. 
     \begin{claim}\label{claim:eq-classes}
         For each $i \in [m]$, the number of equivalence classes under the  relation $\sim_i$ is  at most $\simcwone$. 
     \end{claim}
    
    We postpone the proof of Claim~\ref{claim:eq-classes}. For now, we simply note that the main argument behind Claim~\ref{claim:eq-classes} is this: If the claim is not true and there are $\simcw$ equivalence classes $C_1, C_2,\ldots, C_{\simcw}$ under the relation $\sim_i$, then  there exist $b + 1$ indices $j_1, j_2,\ldots, j_{\simsize + 1} \in [\simcw]$ and $2(b + 1)$ (primal graph) edges $x_1 v_1, x_2 v_2,\ldots, x_{\simsize + 1} v_{\simsize + 1}, ~v_1 y_1, v_2 y_2,\ldots, v_{\simsize + 1} y_{\simsize + 1} \in E(H)$ such that (i) $x_p v_p \in C_{j_p} \subseteq E_i$ and $v_p y_p \in \set{e_{i + 1}, e_{i + 2},\ldots, e_m}$ for every $p \in [\simsize + 1]$ and (ii) the set $\set{\set{x_1 v_1, v_1 y_1}, \set{x_2 v_2, v_2 y_2}, \ldots, \set{x_{\simsize + 1} v_{\simsize + 1}, v_{\simsize + 1} y_{\simsize + 1}}}$ of line graph edges is a semi-induced matching in $L(H)$.  As $x_p v_p \ord e_i \ord v_q y_q$ for every $p, q \in [\simsize + 1]$, this semi-induced matching will be an orsim, and in particular an orsim of size $\simsize + 1$, which will contradict our assumption that every orsim has size at most $\simsize$.  
    
    Assuming Claim~\ref{claim:eq-classes}, let us  complete the proof of the lemma. We now construct a cliquewidth expression of $\lplus(H, \ord)$ using at most $\simcw$ labels; we  use labels $1, 2, 3$, etc. in our construction, which is an iterative procedure that takes place in $m$ stages. 
    In stage $1$, we introduce $e_1$ with label $1$. 
    For each $i = 2, 3, \ldots, m$, in this order, we execute stage $i$, which consists of the following steps. 
    \begin{description}
        \item[Step $i$.1.] We introduce $e_i$ with label $\hat \ell$, where $\hat \ell \in \mathbb{N}$ is a ``fresh'' label---a label that is not shared by any of $e_1, e_2,\ldots, e_{i  - 1}$ at the end of stage $(i - 1)$, and  we choose the least $\hat \ell \in \mathbb{N}$ such that $\hat \ell$ is fresh. 
        \item[Step $i$.2.] We perform a union operation to combine $e_i$ and the previously introduced $e_1, e_2,\ldots, e_{i - 1}$.
        \item[Step $i$.3.] We  insert the arcs directed from $\set{e_1, e_2,\ldots, e_{i - 1}}$ to $e_i$ using arc-join operations.  
        \item[Step $i$.4.] We  insert the (line-graph) edges between $e_i$ and the neighbors of $e_i$ in $\set{e_1, e_2,\ldots, e_{i - 1}}$ using edge-join operations. 
        \item[Step $i$.5.] We perform relabeling operations in the following manner: We relabel in such a way that by the end of this step, for $e_p, e_q \in E_i$, both $e_p$ and $e_q$ would have the same label if and only if $e_p \sim_i e_q$; in particular, if there exist labels $\ell, \ell' \in \mathbb{N}$, where $\ell < \ell'$, such that for every $e_p$ with label $\ell$ and every $e_q$ with label $\ell'$, we have $e_p \sim_i e_q$, then we relabel all $\ell'$-labelled elements of $E_i$ with the label $\ell$. This completes stage $i$.
    \end{description}
     For each $i \in [m]$, let $\eqc(\sim_i)$ denote the number of equivalence classes under the relation $\sim_i$. We now state the following two claims, which show that the above construction is feasible and indeed constructs $\lplus(H, \ord)$, and that the number of labels used in this construction is $1 + \max_{i} \eqc(\sim_i)$; we postpone the proofs of these claims.  
    \begin{claim}\label{claim:feasible}
        For every $i \in [m]$, all the operations in stage $i$ are feasible. Moreover, at the end of stage $i$, the following statements hold. 
        \begin{description}
            \item[(S1).] Let $e_p, e_q \in E_i$.   Then  $e_p$ and $e_q$ have the same label if and only if $e_p \sim_i e_q$. 

            \item[(S2).] Let $\ell \in \mathbb{N}$ be such that there exists at least one element of $E_i$ with the label $\ell$; we call such a label an active label. Then the set of all elements of $E_i$ with label $\ell$ is an equivalence class under the relation $\sim_i$. Consequently, the number of active labels  is precisely $\eqc(\sim_i)$. 
            
            \item[(S3).] The graph constructed at the end of stage $i$ is precisely $\lplus(H[E_i], \ord_{i})$, where $H[E_i]$ is the subgraph of $H$ induced by $E_i$, and $\ord_i$ is the restriction of $\ord$ to $E_i$. 
        \end{description}  
    \end{claim}
    \begin{claim}\label{claim:labels-eqc}
        The number of labels used in the construction above is 
        \[
        \begin{cases}
            1, ~~~~~~~~~~~~~~~~~~~~~~~~~~~~~~~~~~~~~~~~~~~~~~~~~~~~~~\text{ if } m = 1 \\
             1 + \max\set{\eqc(\sim_i) ~|~ i \in [m - 1]}, \text{ otherwise.}
        \end{cases}
        \]
       
    \end{claim}
    Let us quickly complete the proof of the lemma before proving the claims. By Claim~\ref{claim:feasible}, the construction above is feasible and we indeed construct $\lplus(H, \ord)$. By Claim~\ref{claim:labels-eqc}, we use at most $1 + \max_{i \in [m - 1]} \eqc(\sim_i)$ labels in this construction. Finally by Claim~\ref{claim:eq-classes}, for each $i \in [m]$, we have $\eqc(\sim_i) \leq \simcwone$. And hence we can conclude that the cliquewidth of $\lplus(H, \ord)$ is at most $\simcw$, which proves the lemma. We now prove Claims~\ref{claim:feasible},  \ref{claim:labels-eqc} and \ref{claim:eq-classes}.

    \begin{claimproof}[Proof of Claim~\ref{claim:feasible}]
        We prove the claim by induction on $i$. The base case, i.e., $i = 1$, is trivial. 
        Consider $i > 1$, and assume that the claim holds for all $j < i$. 
        Then by the induction hypothesis, the graph constructed at the end of stage ${(i - 1)}$ is precisely $\lplus(H[E_{i - 1}], \ord_{i - 1})$. 
        
        Let us first see that the operations in stage $i$ are feasible. 
        Step $i$.1 is clearly feasible, as it only involves the introduce operation, which can always be performed irrespective of the labels of any of $e_1, e_2,\ldots, e_{i - 1}$. Let $\hat \ell \in \mathbb{N}$ be the label with which we introduce $e_i$ in Step $i$.1. Recall that $\hat \ell$ is a fresh label, and thus none of $e_1, e_2,\ldots, e_{i - 1}$ have the label $\hat \ell$ at the end of stage ${(i - 1)}$.  We can therefore perform the union operation involving $\lplus(H[E_{i - 1}], \ord_{i - 1})$ and $e_i$, and thus Step $i.2$ is feasible.  As for Step $i.3$, notice that in this step, we need to insert the arc $(e_p, e_i)$ for every $e_p \in E_{i - 1}$. We can indeed do this because $e_p$ and $e_i$ have different labels; recall again that the label $\hat \ell$ with which we introduce $e_i$ is not shared by any $e_p \in E_{i - 1}$. 
        
        Now consider Step $i$.4. Here we need to insert the (line graph) edge $e_p e_i$ for every $e_p \in N(e_i) \cap E_{i - 1}$. Notice  that for every label $\ell \in \mathbb{N} \setminus \set{\hat \ell}$, and for every $e_p, e_q \in E_{i - 1}$ such that $e_p$ and $e_q$ have label $\ell$ at the end of stage ${(i - 1)}$, by the induction hypothesis (applied to statement (S1)), we have $e_p \sim_{i - 1} e_q$, which implies that  $N_{> i - 1}(e_p) = N_{> i - 1}(e_q)$. Thus either both $e_p$ and $e_q$ are adjacent to $e_i$, or neither of them is. We can therefore insert all the (line-graph) edges between $e_i$ and $E_{i - 1}$ by performing the $\ejoin_{\hat \ell, \ell}(\cdot)$ operation for every label $\ell \in \mathbb{N} \setminus \set{\hat \ell}$ such that there exists $e_p \in N(e_i) \cap E_{i - 1}$ with label $\ell'$. Thus Step $i$.4 is also feasible. 
        
        Now Step $i$.5. Recall that in this step, we need to relabel the elements of $E_i$ in such a way that at the end of this step $e_p$ and $e_q$ should have the same label if and only if $e_p \sim_i e_q$. To that end, consider $e_p, e_q \in E_i$, where $p < q$. 
        Suppose first that $e_p \nsim_i e_q$. We will argue that $e_p$ and $e_q$ have different labels immediately before we perform Step $i$.5; notice that this will ensure that $e_p$ and $e_q$ will continue to have different labels at the end of stage $i$ as well, because our relabeling in Step $i$.5 would give $e_p$ and $e_q$ the same label only if $e_p \sim_i e_q$. If $q = i$, then $e_p$ and $e_q$ have different labels, as we introduced $e_q = e_i$ with a fresh label. So assume that $q < i$; recall that $p < q$. Thus $e_p, e_q \in E_{i - 1}$. As $e_p \nsim_{i} e_q$, by Claim~\ref{claim:eq-classes}, we can conclude that $e_p \nsim_{i - 1} e_q$, which implies that $N_{> i - 1}(e_p) \neq N_{> i - 1} (e_q)$. Then, by the induction hypothesis (applied to statement (S1)), $e_p$ and $e_q$ have different labels at the end of stage ${(i - 1)}$, and consequently  they have different labels immediately before we perform the operations in Step $i$.5 as well, because Steps $i$.1-$i$.4 do not involve any relabeling operation. 
        
        Conversely, suppose that $e_p \sim_i e_q$. We will show that either $e_p$ and $e_q$ already have the same label or we can indeed perform the relabeling operation as in Step $i$.5 so that they both receive the same label. To do this, we will consider two cases depending on whether $e_p \sim_{i - 1} e_q$ or not. If $e_p \sim_{i - 1} e_q$ (which implicitly implies that $e_p, e_q \in E_{i - 1}$ and in particular that $q < i$), then  the induction hypothesis (applied to statement (S1))  implies that $e_p$ and $e_q$ have the same label at the end of stage ${(i - 1)}$; hence $e_p$ and $e_q$ will continue to have the same label thereafter. So suppose that $e_p \nsim_{i - 1} e_q$. Then, either $q = i$, or $q < i$ but $e_p \nsim_{i - 1} e_q$. Let $\ell, \ell' \in \mathbb{N}$ be such that immediately before Step $i$.5, $\ell$ is the label of $e_p$ and $\ell'$ is the label of $e_q$. And let $C \subseteq E_i$ be the set of elements in $E_i$ with label $\ell$, and $C' \subseteq E_i$ the set of elements in $E_i$ with label $\ell'$. We have $e_p \in C$ and $e_q \in C'$. Now, if $q = i$, then we must have introduced $e_q = e_i$ with label $\ell'$ in Step $i$.1, which means that $\ell'$ must have been a fresh label, and therefore $e_q = e_i$ is the only element of $E_i$ with label $\ell'$, and all elements of $E_i$ with label $\ell$ indeed belong to $E_{i - 1}$; that is, $C' = \set{e_i}$ and $C \subseteq E_{i - 1}$. On the other hand, if $q < i$, then notice that both $C$ and $C'$ are contained in $E_{i - 1}$; to see this, observe that $p < q < i$ and therefore $e_i \notin \set{e_p, e_q}$, and as $e_p$ has label $\ell$ and $e_q$ has label $\ell'$, and as we introduced $e_i$ with a fresh label, we can conclude that $e_i$ cannot have the labels $\ell$ or $\ell'$. Thus, either $C \subseteq E_{i - 1}$ and $C' = \set{e_i}$ (the $q = i$ case), or $C, C' \subseteq E_{i - 1}$ (the $q < i$ but $e_p \nsim_{i - 1} e_q$ case). Then, by the induction hypothesis (applied to statement (S2)), we can conclude that in the former case $C$ is an equivalence class under the relation $\sim_{i - 1}$, and in the latter case both $C$ and $C'$ are equivalence classes under the relation $\sim_{i - 1}$. In either case, for every $e \in C$ and every $e' \in C'$, we will argue that $e \sim_i e'$; to re-emphasize, in the $q = i$ case, $C' = \set{e_q}$ and hence $e' = e_q$. This follows from the following two series of equivalences; the first series holds irrespective of whether $q = i$ or $q < i$, whereas the second series holds only in the $q < i$ case; the justifications for these are stated below.  
        \[
        \begin{matrix}
         \text{When } q \leq i:~~ & e \sim_{i - 1} e_p & \implies & e \sim_i e_p & \implies & e \sim_i e_q. \\
       \text{When } q < i:~~ &e' \sim_{i - 1} e_q & \implies & e' \sim_i e_q. & {} & {} 
        \end{matrix}
        \]
        
        \begin{enumerate}
            \item First, $e \sim_{i - 1} e_p$; this holds because $C$ is an equivalence class under $\sim_{i - 1}$, and $e, e_p \in C$. In the $q < i$ case, we also have $e' \sim_{i - 1} e_q$, because $C'$ is an equivalence class under $\sim_{i -1}$ in this case  and $e', e_q \in C'$. 
            \item The fact that $e \sim_{i - 1} e_p$, along with Claim~\ref{claim:eq-classes}, implies that $e \sim_i e_p$. Similarly, in the $q < i$ case, we also have $e' \sim_i e_{q}$.  
            \item Now, $e \sim_i e_p$ implies that $e \sim_i e_q$; this holds because of the transitivity of the equivalence relation $\sim_i$, and $e_p \sim_i e_q$ by our assumption. Notice that in the $q = i$ case, we have already shown that $e \sim_i e'$, as $e' = e_q$ in this case.  
            \item Finally, in the the $q < i$ case, as $e \sim_i e_q$ and $e' \sim_i e_q$, by the symmetry and transitivity of $\sim_i$, we can conclude that $e \sim_i e'$.  
        \end{enumerate}
        We have thus shown that $e \sim_i e'$ for every $e \in C$ and $e' \in C'$, and thus we can indeed relabel in such a way that all elements of $C \cup C'$ receive the label $\min\set{\ell, \ell'}$. And thus Step $i$.5 is feasible. 

        Let us now quickly observe that the feasibility of the operations guarantee that at the end of stage $i$, statements (S1), (S2) and (S3) hold. Statement (S3) holds because Steps $i$.1-$i$.4 are feasible. To see this, recall that by induction hypothesis, the graph constructed at the end of stage $(i - 1)$ is $(\lplus(H[E_{i - 1}]), \ord_{i - 1})$. Then in Steps $i$.1-$i$.4, we introduce $e_i$ and add the edges and arcs between $e_i$ and $E_{i - 1}$ to $(\lplus(H[E_{i - 1}]), \ord_{i - 1})$, which produces precisely the graph $\lplus(H[E_i], \ord_i)$. Statement (S1) holds because Step $i$.5 is feasible: Step $i$.5 guarantees that  $e_p, e_q \in E_i$ have the same label if and only if $e_p \sim_i e_q$. Finally, notice that statement (S2) is an immediate consequence of statement (S1). 
    \end{claimproof}

    \begin{claimproof}[Proof of Claim~\ref{claim:labels-eqc}]
        Note first that in stage $1$, we use just one label; in particular, if $m = 1$, then our construction consists only of stage $1$, and hence we use just one label. Consider $m > 1$. 
        By Claim~\ref{claim:feasible}-statement (S2), for every $i > 1$, the number of active labels at the end of stage $(i - 1)$ is exactly $\sim_{i - 1}$. 
        
        Observe now that the only points in our construction when we would possibly use a previously unused label are Steps $i$.1 for $i \in [m]$.
        Observe now that at any point in our construction, the labels used until that point (not necessarily within a stage) is a set of consecutive numbers starting from $1$, say $1, 2, \ldots, \ell$, for some $\ell \in \mathbb{N}$. This holds because in stage $1$, we introduce $e_1$ with label $1$. Then in every subsequent stage, say in stage $i$ for some $i > 1$, when we introduce $e_i$, we introduce $e_i$ with the \emph{least} fresh label, say $\hat \ell \in \mathbb{N}$. In particular, we would use a previously unused label $\hat \ell$ in Step $i$.1 only if none of the labels $1, 2,\ldots, \hat \ell - 1$ is available at that point, which would happen only if the labels $1, 2,\ldots, \hat \ell - 1$ are all active at the end of stage $(i - 1)$, which would happen only if there are exactly $\hat \ell - 1$ equivalence classes under the relation $\sim_{i - 1}$. We can thus conclude that the total number of labels used in our construction is precisely $1 + \max_{i \in [m - 1]}\eqc(\sim_{i - 1})$. 
    \end{claimproof}
   
    \begin{claimproof}[Proof of Claim~\ref{claim:eq-classes}]
        Recall that $(e_1, e_2,\ldots, e_m)$ is the ordering on $E(H)$ imposed by $\ord$, and for every $i \in [m]$, $E_i = \set{e_1, e_2,\ldots, e_m}$. 
        Assume that the claim is false. Then there exists $i \in [m]$ such that there are at least $\simcw$ equivalence classes under the relation $\sim_i$; fix such an $i$. In the rest of the proof, the only equivalence relation we deal with will be $\sim_i$, and the only equivalence classes we deal with will be equivalence classes under $\sim_i$; so we will simply refer to equivalence classes without specifying the relation $\sim_i$. We will show that $\lplus(H, \ord)$ contains an orsim of size at least $\simsize + 1$, which will contradict our assumption that every orsim in $\lplus(H, \ord)$ has size at most $b$. In particular, for every (line graph edge) $ee'$ in this orsim, we will have $e \in E_i$ and $e' \in \set{e_{i + 1}, e_{i + 2},\ldots, e_m}$. 

        We first define three terms that we will use in this proof. For $e \in E_i$, we say that $e$ is \emph{live} if $e$ has a neighbor in $\set{e_{i + 1}, e_{i + 2},\ldots, e_m}$; that is, $e$ is live if $N_{> i}(e) \neq \emptyset$. For an equivalence class $C$, with a slight abuse of terminology, we say that $C$ is \emph{live} if every element of $C$ is live. Notice that by the definition of $\sim_i$, either every element of $C$ is live or none of them is.  
        Notice also that all equivalence classes, except possibly one, are live; if two equivalence classes $C$ and $C'$ are not live, then $N_{> i}(e) = N_{> i}(e') = \emptyset$ for every $e \in C$ and $e' \in C'$, and thus $e \sim_i e'$, which implies that both $e$ and $e'$ must belong to the same equivalence class and hence $C = C'$. 
        For a vertex $v$ of $H$ and an edge $ee'$ of the line graph $L(H)$, where $e \in E_i$ and $e' \in \set{e_{i + 1}, e_{i + 2},\ldots, e_m}$, we say that $v$ is \emph{responsible for $ee'$} if $v$ is an endpoint of both $e$ and $e'$, i.e., $e = uv$ and $e' = vw$ for some $u, w \in V(H)$. To emphasize, if $v$ is responsible for $ee'$, then by definition, the line graph edge $ee'$ is an edge between $E_i$ and $\set{e_{i + 1}, e_{i + 2},\ldots, e_{m}}$. Also, for a vertex $v$ of $H$, with a slight abuse of terminology, we say that $v$ is \emph{responsible} if there exists a line graph edge $ee'$ such that $v$ is responsible for $ee'$; notice in this case that (i) $e  \in E_i$ and $e' \in \set{e_{i + 1}, e_{i + 2},\ldots, e_{m}}$,   (ii) the vertex $v$ is  the common endpoint of $e$ and $e'$, and (iii) $e$ is live. We also emphasize the following relationship between a live edge and a responsible vertex: For  $e \in E_i$, where $e = uv$, notice that $e$ is live if and only if at least one of $u$ and $v$ is responsible. 
        Finally, for a vertex $v$ of $H$ and an equivalence class $C$, we say that $v$ \emph{touches} $C$ if there exists $e \in C$ such that $v$ is an endpoint of $e$. i.e., there exists $e \in C$ where $e = uv$ for some $u \in V(H)$. 

        We claim that every vertex $v$ of $H$ touches at most $\simsize( \simsize + 2) + 2$ equivalence classes. Suppose this is not true. Then there exist a vertex $v \in V(H)$ and $\simsize (\simsize + 2) + 3$ distinct equivalence classes $C'_1, C'_2,\ldots, C'_{\simsize (\simsize + 2) + 3}$ such that $v$ touches $C'_{j}$ for all $j \in [\simsize (\simsize + 2) + 3]$. As at most one these equivalence classes is not live, at least $\simsize (\simsize + 2) + 2$ of them are live; assume without loss of generality that $C'_1, C'_2,\ldots, C'_{\simsize (\simsize + 2) + 2}$ are live. As $v$ touches each of these equivalence classes, for each $j \in [\simsize (\simsize + 2) + 2]$, the equivalence class $C'_{j}$ contains a (primal graph) edge incident with $v$; for each $j \in [\simsize (\simsize + 2) + 2]$, fix such an edge $v u_j \in C'_j$. Let us now observe the following three facts.
        \begin{description}
            \item[(F1).] As $H$ has no parallel edges, the (primal graph) vertices  $u_1, u_2,\ldots, u_{\simsize (\simsize + 2) + 2}$ are all distinct. 
            \item[(F2).] For each $j \in [\simsize (\simsize + 2) + 2]$, as $C'_j$ is live, $v u_j$ is live in particular. 
            \item[(F3).] All $u_j$s, except possibly one, are responsible. Suppose not. Say there exist distinct $j_1, j_2 \in [\simsize (\simsize + 2) + 2]$ such that $u_{j_1}$ and $u_{j_2}$ are not responsible. We will argue that $v u_{j_1} \sim_i v u_{j_2}$, which will contradict the fact that $v u_{j_1}$ and $v u_{j_2}$ belong to different equivalence classes. To see this, notice that as $u_{j_1}$ is not responsible, $v$ must be responsible for all the  line graph edges between $v u_{j_1}$ and $\set{e_{i + 1}, e_{i + 2},\ldots, e_m}$. In other words, in the line graph $L(H)$, the neighbors of $v u_{j_1}$ in the set $\set{e_{i + 1}, e_{i + 2}, \ldots, e_m}$ are precisely those $e_p \in \set{e_{i + 1}, e_{i + 2}, \ldots, e_m}$ that have $v$ as an endpoint; that is, $N_{> i}(v u_{j_1}) = \set{e_p ~|~ i + 1 \leq p \leq m  \text { and } e_p = wv \text{ for some } w \in V(H)}$. 
            Identical reasoning applies to $v u_{j_2}$, and we consequently have $N_{> i}(v u_{j_1}) = N_{> i}(v u_{j_2})$, which implies that $v u_{j_1} \sim_i v u_{j_2}$, a contradiction.  
        \end{description}  
        In light of Fact (F3), assume without loss of generality that $u_1, u_2,\ldots, u_{\simsize (\simsize + 2) + 1}$ are responsible. Then for every $j \in [\simsize (\simsize + 2) + 1]$, there exists a (primal graph) edge $u_j w_j$ in $ \set{e_{i + 1}, e_{i + 2}, \ldots, e_{m}}$, so that $\set{v u_j, u_j w_j}$ is an edge in the line graph. As $u_1, u_2,\ldots, u_{\simsize (\simsize + 2) + 1}$ are all distinct (Fact  (F1) above), $u_1 w_1, u_2 w_2,\ldots, u_{\simsize (\simsize + 2) + 1} w_{\simsize (\simsize + 2) + 1}$ are all distinct. (As an aside, notice that $w_1, w_2,\ldots, w_{\simsize (\simsize + 2) + 1}$ need not all be distinct.) 
        Thus the set $M'_v = \set{\set{v u_1, u_1 w_1}, \set{v u_2, u_2 w_2}, \ldots, \set{v u_{\simsize (\simsize + 2) + 1}, u_{\simsize (\simsize + 2) + 1} w_{\simsize (\simsize + 2) + 1}}}$ is a set of $\simsize (\simsize + 2) + 1$ distinct line graph edges. 
        
        We will argue that $M'_v$ contains a semi-induced matching $M_v$ of size at least $b + 1$; for every $j, p \in [\simsize (\simsize + 2) + 1]$, the fact that $v u_j \in E_i$ and $u_{p} w_{p} \in \set{e_{i + 1}, e_{i + 2},\ldots, e_m}$ and hence $v u_j \ord u_p w_p$ will then immediately imply that $M_v$ is an orsim in $\lplus(H, \ord)$, a contradiction to our assumption that every orsim in $\lplus(H, \ord)$ has size at most $\simsize$.  To that end, notice that as $H$ is a simple graph, and hence has no self-loops or parallel edges, and as $v u_j$ and $u_j w_j$ are distinct edges of $H$, we can conclude that $u_j \neq v$, $w_j \neq v$ and $w_j \neq u_j$ for every $j \in [\simsize (\simsize + 2) + 1]$. 
        It is however possible that $w_j = u_{p}$ for distinct $j, p \in [\simsize (\simsize + 2) + 1]$. 
        Notice also that if $w_j = u_{p}$, then any semi-induced matching can contain at most one of the two line graph edges $\set{v u_j, u_j w_j}$ and $\set{v u_{p}, u_{p} w_{p}}$, as the line graph would contain an edge between $u_j w_j$ and $v u_{p}$ in this case. But for each $j \in [\simsize (\simsize + 2) + 1]$, we can easily bound the number of indices $p$ such that $w_j = u_p$ or $u_j = w_p$. 
        \begin{description}
            \item[(F4).] First of all, if there exist $\simsize + 1$ distinct indices $p_1, p_2,\ldots, p_{\simsize + 1} \in [\simsize (\simsize + 2) + 1]$ such that $w_{p_1} = w_{p_2} = \cdots = w_{p_{\simsize + 1}}$, then notice that the subset $\set{\set{v u_{p_1}, u_{p_1} w_{p_1}}, \set{v u_{p_2}, u_{p_2} w_{p_2}},\ldots, \set{v u_{p_{\simsize + 1}}, u_{p_{\simsize + 1}} w_{p_{\simsize + 1}}}}$ of $M'_v$ is a semi-induced matching of size $\simsize + 1$, and in particular, an orsim of size $\simsize + 1$, a contradiction. So we can assume from now on that there do not exist $\simsize + 1$ distinct indices $p_1, p_2,\ldots, p_{\simsize + 1} \in [\simsize (\simsize + 2) + 1]$ such that $w_{p_1} = w_{p_2} = \cdots = w_{p_{\simsize + 1}}$. 

            \item[(F5).] Fact (F4) immediately implies that for every $j \in [\simsize (\simsize + 2) + 1]$, there exist at most $\simsize$ many indices $p \in [\simsize (\simsize + 2) + 1] \setminus \set{j}$ such that $u_j = w_p$.  

            \item[(F6).] For every $j \in [\simsize (\simsize + 2) + 1]$, there exists at most one index $p \in [\simsize (\simsize + 2) + 1] \setminus \set{j}$ such that $w_j = u_p$; this holds because $u_1, u_2,\ldots, u_{\simsize (\simsize + 2) + 1}$ are all distinct (Fact (F1)). 
        \end{description}
        
        We now construct a semi-induced matching $M_v \subseteq M'_v$ by greedily adding elements of $M'_v$ to $M_v$: We start with $M_v = \emptyset$, and for $j = 1, 2,\ldots, \simsize (\simsize + 2) + 1$ in this order, we add $\set{v u_j, u_j w_j}$ to $M_v$ if $w_j \neq u_p$ and $u_j \neq w_p$ for every $\set{v u_p, u_p w_p} \in M_v$; this completes the construction of $M_v$. 
        \begin{description}
            \item[(F7).] Observe now that the set $M_v$ constructed this way will have size at least $\simsize + 1$. Because each time we add a $\set{v u_j, u_j w_j}$ to $M_v$, we render at most $\simsize + 1$ many elements $\set{v u_{p}, u_{p} w_{p}} \in M'_v \setminus M_v$ ineligible to be added in the future, as there exists at most $\simsize$ many indices $p$ such that $u_j = w_p$, and at most one index $p$ such that $w_j = u_{p}$ (Facts (F5) and (F6) above).  So more generally, for $r \in [\simsize (\simsize + 2) + 1]$, after adding $r$ many elements to $M_v$, we will have $\card{M_v} = r$, and we will have rendered at most $r(\simsize + 1)$  elements of $M'_v \setminus M_v$ ineligible-for-future, and thus $M'_v$ would still have at least $\card{M'_v} - \card{M_v} - r(\simsize + 1) = \card{M'_v} -  r(\simsize + 2)$ ``eligible'' elements left, each of which could be added to $M_v$. In particular, as $\card{M'_v} = \simsize (\simsize + 2) + 1$, even after adding $\simsize$ many elements to $M_v$, we would still have at least one eligible element left in $M'_v$, which we would add to $M_v$. 
            We can thus conclude that we would add at least $\simsize + 1$ elements to $M_v$, and  we will have $\card{M_v} \geq \simsize + 1$ in the end. 

            \item[(F8).] Also, by construction, $M_v$ is a semi-induced matching; to see this, consider two distinct elements $\set{v u_p, u_p w_p}, \set{v u_j, u_j w_j} \in M_v$, and let us observe that in the line graph $L(H)$, there does not exist an edge between $v u_p$ and $u_j w_j$ or between $u_p w_p$ and $v u_j$.  Assume without loss of generality that $p < j$, so that we would have added $\set{v u_p, u_p w_p}$ to $M_v$ before adding $\set{v u_j, u_j w_j}$. Then by the construction of $M_v$, we have $w_j \neq u_p$ and $u_j \neq w_p$. We have already noted that  $u_p \neq u_j$ (Fact (F1)).  As $v u_p$ and $v u_j$ are edges of $H$ and as $H$ has no self-loops, we have $v \neq u_p$ and $v \neq u_j$. And finally, as $v u_p$ and $u_p w_p$ are distinct edges of $H$ (recall that $v u_p \in E_i$ and $u_p w_p \in E(H) \setminus E_i$), and as $H$ has no parallel edges, and in particular, $v u_p$ and $u_p w_p$ cannot be parallel edges, we can conclude that $v \neq w_p$; by identical reasoning, we have $v \neq w_j$. Notice now that $v \neq u_j$, $v \neq w_j$, $u_p \neq u_j$ and $u_p \neq w_j$ together imply that there does not exist an edge in $L(H)$ between $v u_p$ and $u_j w_j$. Similarly, $u_p \neq v$, $u_p \neq u_j$, $w_p \neq v$ and $w_p \neq u_j$ together imply that there does not exist an edge in $L(H)$ between $u_p w_p$ and $v u_j$. 
        \end{description} 
        Thus $M_v$ is a semi-induced matching of size at least $b + 1$, and in particular an orsim of size at least $b + 1$, which is a contradiction. Recall that the assumption that led to this contradiction was that $H$ has a vertex $v$ that touches at least $\simsize(\simsize + 2) + 3$ equivalence classes. 

        We have thus shown that every vertex of $H$ touches at most $\simsize(\simsize + 2) + 2$ equivalence classes. Recall that we are under the assumption that there are at least $\simcw$  equivalence classes under $\sim_i$; recall also that at most one of those equivalence classes is not live, so at least $\simcwone$  equivalence classes are live. Let $C_1, C_2,\ldots, C_{\simcwone}$ be distinct, live equivalence classes. We will show that $\lplus(H, \ord)$ contains an orsim of size at least $\simsize + 1$, a contradiction to our assumption that every orsim in $\lplus(H, \ord)$ has size at most $\simsize$. To do this, for each $j \in [\simcwone]$, fix an element $x_j v_j \in C_j$. 
        
        For every $j \in [\simcwone]$, notice that as the equivalence class $C_j$ is live, every element of $C_j$ is live. In particular, $x_j v_j$ is live, and hence at least one of $x_j$ and $v_j$ is responsible; assume without loss of generality that $v_j$ is responsible. So there exists a (primal graph) edge $v_j y_j \in \set{e_{i + 1}, e_{i + 2},\ldots, e_m}$, and hence $\set{x_j v_j, v_j y_j}$ is an edge in the line graph $L(H)$. 
        Again, for distinct indices $j, p \in [\simcwone]$, as $x_j v_j$ and $x_p v_p$ are  are distinct edges of $H$, the line graph edges $\set{x_j v_j, v_j y_j}$ and $\set{x_p v_p, v_p y_p}$ are distinct. 
        Let $M'' = \set{\set{x_j v_j, v_j y_j} ~|~ j \in [\simcwone]}$ be the set consisting of these $\simcwone$ line graph edges. We will argue that $M''$ contains a semi-induced matching $M$ of size at least $\simsize + 1$. Again, for every $j, p \in [\simcwone]$, the fact that $x_j v_j \in E_i$ and $v_{p} y_{p} \in \set{e_{i + 1}, e_{i + 2}, \ldots, e_m}$ and hence $x_j v_j \ord  v_{p} y_{p}$ will then immediately imply that $M$ is an orsim of size at least $\simsize + 1$. 
        We construct $M$ in a two-stage process. In the first stage, we construct a subset $M'$ of $M''$ with the following two properties: (i) $M'$ will have size at least $2 \simsize(\simsize + 1) + 1$, and (ii) the $x_j$s and $v_j$s that appear in different elements of $M'$ will all be distinct, i.e., for distinct $\set{x_j v_j, v_j y_j}, \set{x_{p} v_{p}, v_{p} y_{p}} \in M'$, the vertices $x_j, v_j, x_{p}, v_{p}$ will all be distinct. But it may be the case that $y_j = v_{p}$ or $y_j = x_{p}$. Then in the second stage, we prune $M'$ and construct a subset $M$ of $M'$, in which we will also have $y_j \neq v_p$ and $y_j \neq x_p$, which will ensure that $M$ is a semi-induced matching; also, $M$ will have size at least $\simsize + 1$. 
        
        We now proceed to constructing $M'$. To that end, consider $p \in [\simcwone]$. Recall that $C_p$ is the equivalence class that contains $x_p v_p$; recall also that for an equivalence class $C$, we say that $x_p$ (resp. $v_p$) touches $C$ if there exists (a primal graph edge) $e \in C$ such that $x_p$ (resp. $v_p$) is an endpoint of $e$. Let $\touch(p)$ be the set of all indices $j \in [\simcwone]$ such that either $x_p$ or $v_p$  touches the equivalence class $C_j$. Notice that $p \in \touch(p)$, as $x_p$ touches $C_p$; in fact, both $x_p$ and $v_p$ touch $C_p$. Let us now observe that $\card{\touch(p)} \leq 2\simsize (\simsize + 2) + 3$. As every vertex of $H$---and in particular each of $x_p$ and $v_p$---touches at most $\simsize (\simsize + 2) + 2$ equivalence classes, each of $x_p$ and $v_p$  contributes at most $\simsize (\simsize + 2) + 2$ indices to $\touch(p)$. So we immediately have $\card{\touch(p)} \leq 2[\simsize (\simsize + 2) + 2] =  2\simsize (\simsize + 2) + 4$. But notice that we over-counted by $1$ in this bound for $\card{\touch(p)}$: As both $x_p$ and $y_p$ touch $C_p$, we accounted for $p \in \touch(p)$ in the contribution of $x_p$ and in the contribution for $y_p$. We can therefore conclude that $\card{\touch(j)} \leq 2\simsize (\simsize + 2) + 3$. 
        We now greedily construct $M'$ from $M''$ as follows: We start with $M' = \emptyset$, and for each $j = 1, 2,\ldots, \simcwone$ in this order, we add $\set{x_j v_j, v_j y_j}$ to $M'$ if $j \notin \touch(p)$ for every index $p$ such that $\set{x_{p} v_{p}, v_{p} y_{p}} \in M'$; this completes the construction of $M'$.   
        We now prove the  following facts about the set $M'$ constructed this way.
        \begin{description}
            \item[(F9).] First, we have $\card{M'} \geq 2 \simsize(\simsize + 1) + 1$. This holds because each time we add an $\set{x_j v_j, v_j y_j}$ to $M'$, for every index $q \in [\simcwone]$ with $q \in \touch(j)$, we render $\set{x_q v_q, v_q y_q}$ ineligible to be added in the future. And as $\card{\touch(j)} \leq 2\simsize (\simsize + 2) + 3$, we render at most $2\simsize (\simsize + 2) + 3$ elements of $M''$ ineligible this way; notice that we also count $\set{x_j v_j, v_j y_j}$ in these $2\simsize (\simsize + 2) + 3$ ineligible-for-future elements as $j \in \touch(j)$. So after adding $r$ many elements to $M'$, we will have $\card{M'} = r$, we will have rendered at most $r(2\simsize (\simsize + 2) + 3)$ elements of $M''$ ineligible, and we will still have $\card{M''} - r(2\simsize (\simsize + 2) + 3)$ eligible elements in $M''$. And as $\card{M''} = \simcwone = (2\simsize^2 + 2 \simsize)(2\simsize^2 +  4\simsize + 3) + 1 = 2\simsize(\simsize + 1)[2 \simsize(\simsize + 2) + 3] + 1$, we can conclude that we would add at least $2 \simsize(\simsize + 1) + 1$ elements to $M'$, and thus we would have $\card{M'} \geq 2 \simsize(\simsize + 1) + 1$ in the end. 

            \item[(F10).] For every $\set{x_j v_j, v_j y_j} \in M'$, the vertices $x_j, v_j$ and $y_j$ are distinct;  this is in fact true for every $\set{x_j v_j, v_j y_j} \in M''$. This statement holds because $H$ is a simple graph.  First of all, $x_j v_j$ and $v_j y_j$ are distinct edges of $H$; recall that $x_j v_j \in E_i$ and $v_j y_j \in E(H) \setminus E_i$. Then, as $H$ has no self loops, we have $x_j \neq v_j$ and $v_j \neq y_j$, and as $H$ has no parallel edges, we have $x_j \neq y_j$. 
            
            \item[(F11).]  For distinct $\set{x_p v_p, v_p y_p}, \set{x_j v_j, v_j y_j} \in M'$, the vertices $x_p, v_p, x_j, v_j$ are all distinct. We already have $x_p \neq v_p$ and $x_j \neq v_j$ (Fact  (F10)). To see that other pairs of vertices are distinct too, assume without loss of generality that $p < j$ so that we would have added $\set{x_p v_p, v_p y_p}$ to $M'$ before adding $\set{x_j v_j, v_j y_j}$. And we can therefore conclude that $j \notin \touch(p)$, which   
             implies that neither $x_p$ nor $v_p$ touches the equivalence class $C_j$; recall that $C_j$ is the equivalence class that contains $x_j v_j$. As both $x_j$ and $v_j$ touch $C_j$, and neither $x_p$ nor $v_p$  touches $C_j$, we can conclude that $x_j \neq x_p$, $x_j \neq v_p$, $v_j \neq x_p$ and $v_j \neq v_p$. 
        \end{description}

        We thus have a subset $M'$ of $M''$ with $\card{M'} \geq 2 \simsize(\simsize + 1) + 1$. Let $t = \card{M'}$, and assume without loss of generality that $M' = \set{\set{x_1 v_1, v_1 y_1}, \set{x_2 v_2, v_2 y_2},\ldots, \set{x_t v_t, v_t y_t}}$. Notice that while the $x_j$s and $y_j$ appearing in different $M'$ are all distinct, we may still have $y_j = x_p$ or $y_j = v_p$ for distinct $j, p \in [t]$. But we note the following fact about such $y_j$s.   
        \begin{description}
            \item[(F12).] If there exist $\simsize + 1$ distinct indices $p_1, p_2,\ldots, p_{b + 1} \in [t]$ such that $y_{p_1} = y_{p_2} = \cdots = y_{p_{\simsize + 1}}$, then the subset $\set{\set{x_{p_1} v_{p_1}, v_{p_1} y_{p_1}}, \set{x_{p_2} v_{p_2}, v_{p_2} y_{p_2}},\ldots, \set{x_{p_{\simsize + 1}} v_{p_{\simsize + 1}}, v_{p_{\simsize + 1}} y_{p_{\simsize + 1}}}}$ of $M'$ is a semi-induced matching of size $\simsize + 1$, and in particular, an orsim of size $\simsize + 1$, a contradiction. So we assume from now on that there do not exist $\simsize + 1$ distinct indices $p_1, p_2,\ldots, p_{\simsize + 1} \in [t]$ such that $y_{p_1} = y_{p_2} = \cdots = y_{p_{\simsize + 1}}$. 
            
            \item[(F13).] Fact (F12) immediately implies that for every $j \in [t]$, there exist at most $\simsize$ many indices $p \in [t] \setminus \set{j}$ such that $y_p = x_j$, and there exist at most $\simsize$ many indices $p' \in [t] \setminus \set{j}$ such that $y_{p'} = v_j$. 
            
            \item[(F14).] For every $j \in [t]$, there exists at most one index $q \in [t] \setminus \set{j}$ such that $x_q = y_j$ or $v_q = y_j$. This holds because the vertices $x_1, x_2,\ldots, x_t, v_1, v_2,\ldots, v_t$ are all distinct (Facts (F10) and (F11) above). 
        \end{description}
        
        We now follow our familiar strategy and  construct a semi-induced matching $M$ from $M'$. We start with $M = \emptyset$, and for $j = 1, 2,\ldots, t$ in this order, we add $\set{x_j v_j, v_j y_j}$ to $M$ if  $y_p \neq x_j$, $y_p \neq v_j$, $x_p \neq y_j$ and  $v_p \neq y_j$ for every $\set{x_p v_p, v_p y_p} \in M$; this completes the construction of $M$. And we observe the following facts about $M$. 

        \begin{description}
            \item[(F15).] We have $
            \card{M} \geq \simsize + 1$. This holds because each time we add an $\set{x_j v_j, v_j y_j}$ to $M$, we render at most $2 \simsize + 1$ elements of $M' \setminus M$ ineligible to be added in the future. These ineligible-for-future elements (if any of them exists) are precisely (i) $\set{x_p v_p, v_p y_p}$ with $y_p = x_j$ (at most $\simsize$ many such indices $p \in [t]$, by (F13)), (ii) $\set{x_{p'} v_{p'}, v_{p'} y_{p'}}$ with $y_{p'} = v_j$ (at most $\simsize$ many such indices  $p' \in [t]$, again, by (F13)), and (iii) $\set{x_q v_q, v_q y_q}$ with $x_q = y_j$ or $v_q = y_j$ (at most one such index $q \in [t]$, by (F14)). Hence, after adding $r$ many elements to $M$, we will have $\card{M} = r$, we will have rendered at most $r(2 \simsize + 1)$ elements of $M' \setminus M$ ineligible for future, and we will still have $\card{M'} - \card{M} - r(2 \simsize + 1) = \card{M'} - r (2 \simsize + 2)$ eligible elements left in $M'$. As $\card{M'} \geq 2 \simsize(\simsize + 1) + 1$, we can conclude that we would add at least $\simsize + 1$ elements to $M$, and thus we would have $\card{M} \geq \simsize + 1$ in the end. 

            \item[(F16).] The set $M$ is a semi-induced matching in $L(H)$. To see this, consider distinct $\set{x_p v_p, v_p y_p}$ and $\set{x_j v_j, v_j y_j} \in M$, and let us observe that in the line graph $L(H)$, there does not exist an edge between $x_p v_p$ and $v_j y_j$ or between $v_p y_p$ and $x_j v_j$.  Assume without loss of generality that $p < j$, so that we would have added $\set{v u_p, u_p w_p}$ to $M_v$ before adding $\set{v u_j, u_j w_j}$. We already have $x_p \neq v_j$ and $v_p \neq v_j$ (Fact (F11)). By the construction of $M$, we also have $x_p \neq y_j$ and $v_p \neq y_j$, for otherwise we would not have added $\set{x_j v_j, v_j y_j}$ to $M$. These facts imply that there is no line graph edge between $x_p v_p$ and $v_j y_j$. Using symmetric arguments, we can conclude that there is no line graph edge between $v_p y_p$ and $x_j v_j$ either. 
        \end{description}

          We have thus shown that $M$ is a semi-induced matching in $L(H)$ of size at least $\simsize + 1$.    And as $x_p v_p \ord v_j y_j$ for every $\set{x_p v_p, v_p y_p}, \set{x_j v_j, v_j y_j} \in M$, we can conclude that $M$ is indeed an orsim of size at least $\simsize + 1$ in $\lplus(H, \ord)$, a contradiction to the fact that every orsim in $\lplus(H, \ord)$ has size at most $b$. Recall  that the assumption that led to this contradiction was that the equivalence relation $\sim_i$ has at least $\simcw$ equivalence classes. We can thus conclude that $\sim_i$ has at most $\simcwone$ equivalence classes. 
    \end{claimproof}
    This completes the proof of the lemma. 
\end{proof}

\subsection{Parameterised Intractability along large Semi-Induced Matchings}\label{sec:lower-bound}
Recall that a semi-induced matching of a graph $L$ is a set $M$ of pairwise disjoint edges $\{u_1,v_1\},\dots,\{u_k,v_k\}$ such that there are no edges in $L$ between $u_i$ and $v_j$ for $i \neq j$.
A $P_2$-\emph{packing} of a graph $H$ is a set $S$ of pairs of $2$-paths $(u_1,v_1,w_1),\dots,(u_k,v_k,w_k)$ such that the $v_i$ are pairwise distinct. 

\begin{observation}\label{obs:semi-induced-line-match_P2packing}
    Let $H$ be a graph and let $L=L(H)$ be the line graph of $H$. If $L$ has a semi-induced matching of size $k$, then $H$ contains a $P_2$-packing of size $k$.\qed
\end{observation}

\subsubsection{Construction of the Global Gadget}
In this subsection, we will define and analyse the main construction employed in our hardness reduction.
To this end, we introduce an edge-coloured intermediate version of the homomorphism counting problem.

\begin{definition}[$H$-coloured temporal graphs]
    An $H$-colouring of a graph $G$ is a homomorphism $c=(\nu_c,\xi_c)$ from $G$ to $H$ such that $\nu$ is surjective.\footnote{We note that the surjectivity condition is not usually present for $H$-colourings in the literature; however, for the purpose of this work, we will rely on the surjectivity constraint.} An $H$-coloured temporal graph is a triple $(G,\tau,c)$ such that $(G,\tau)$ is a temporal graph, and $c$ is an $H$-colouring of $G$.
\end{definition}

\begin{definition}[Edge-colourful homomorphisms]
     Let $(H,\preccurlyeq)$ be a totally ordered temporal pattern, and let $(G,\tau,c)$ be an $H$-coloured temporal graph with $c=(\nu_c,\xi_c)$. A homomorphism $\varphi=(\nu_\varphi,\xi_\varphi)$ from $(H,\preccurlyeq)$ to $(G,\tau)$ is called \emph{edge-colourful} if the mapping $e \mapsto \xi_c(\xi_\varphi(e))$ is a bijection on $E(H)$. In other words, $\varphi$ is called edge-colourful if it hits all edge-colours of $G$ w.r.t.\ the edge-colouring $\xi_c$. We write
     \[ \colhom{(H,\preccurlyeq)}{(G,\tau,c)} \]
     for the set of all edge-colourful homomorphisms from $(H,\preccurlyeq)$ to $(G,\tau,c)$.
\end{definition}

The majority of the work in this section is dedicated towards proving the following construction, which will allow us to reduce the clique problem to counting colourful temporal homomorphisms from totally ordered temporal patterns that contain large $P_2$-packings. 
\begin{lemma}\label{lem:gadget_construction}
    There is a polynomial-time algorithm $\mathbb{A}$ that receives as input 
    \begin{itemize}
        \item  a totally ordered temporal pattern $(H,\preccurlyeq)$,
        \item a positive integer $k$, encoded in binary,
        \item a $P_2$-packing of size $(k+1)^2-1$ of $H$, and
        \item a graph $F$ with $|V(F)|>k$.
    \end{itemize}
     $\mathbb{A}$ computes an $H$-coloured temporal graph $(G,\tau,c)$ and a total order $\preccurlyeq$ on $E(H)$ such that $F$ contains a $k$-clique if and only if $\#\colhom{(H,\preccurlyeq)}{(G,\tau,c)}\geq 1$.\qed
\end{lemma}
For the purpose of avoiding notational clutter, we allow ourselves to fix the following objects:
\begin{itemize}
    \item $k\geq 2$ is a positive integer.
    \item $H$ is a graph containing a $P_2$-packing of size $(k+1)^2-1$.
    \item $F$ is a graph with $n>k$ vertices and $m$ edges. We assume that $F$ does not contain isolated vertices as they will not contribute to any $k$- cliques.
\end{itemize}
For what follows, recall that, for any natural number $N$, we set $\nzero{N}=\{0,\dots,N-1\}$. Moreover, further streamlining the construction, we will use the below conventions within this section, illustrated in Figure~\ref{fig:gadget_graphH}.
\begin{itemize}
    \item We use $U$ to denote the set of vertices of $F$, and we assume $U=\nzero{n}$. Furthermore, we will use the letter $u$ for denoting vertices of $F$, such as $u,u',u_1,u_2,\dots$.
    \item We will use the letters $\ell$ and $r$ solely in the context of \textbf{l}eft and \textbf{r}ight.
    \item We will use the letters $h$ and $v$ solely in the context of \textbf{h}orizontal and \textbf{v}ertical indexes in a grid.
    \item The 2-paths in the $P_2$-packing of $H$ are denoted by $p_{h,v}$ for $h,v \in \nzero{k+1}$, excluding $p_{0,0}$. We denote the vertices of $p_{h,v}$ by $p^1_{h,v}$, $p^2_{h,v}$, and $p^3_{h,v}$, and we denote its edges by $e^\ell_{h,v}=\{p^1_{h,v},p^2_{h,v}\}$ and $e^{r}_{h,v}=\{p^2_{h,v},p^3_{h,v}\}$. 
\end{itemize}

\begin{figure}[t]
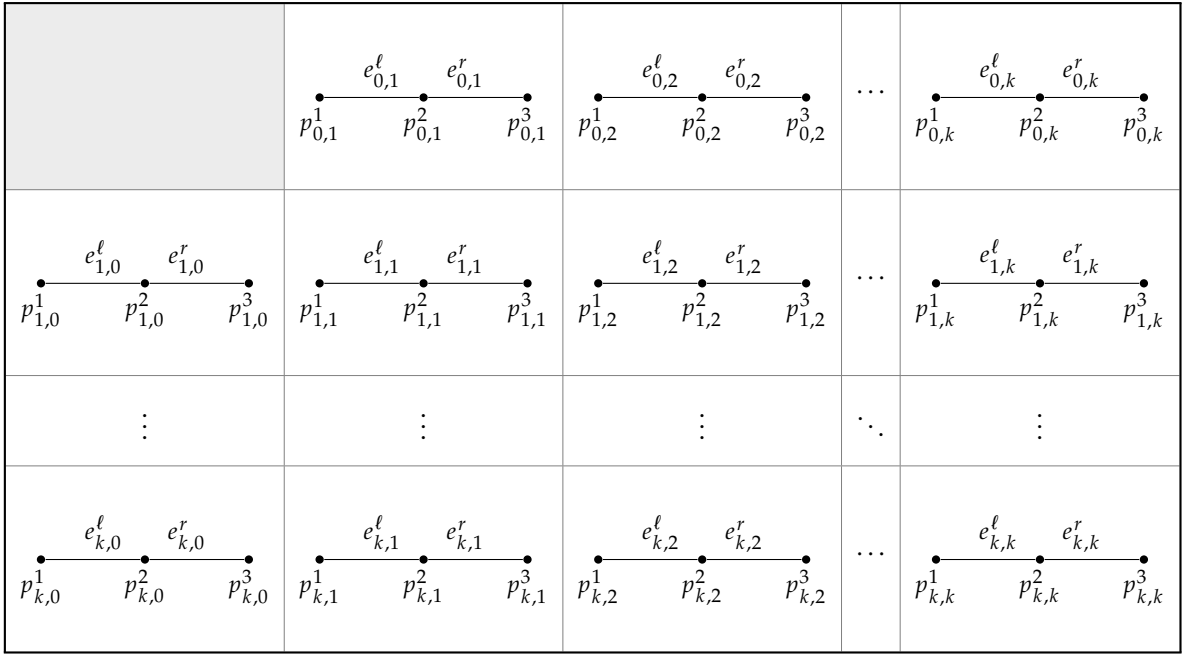

    \centering
    \begingroup
\renewcommand{\arraystretch}{2.5}
\setlength{\arrayrulewidth}{0.25pt}
\arrayrulecolor{black!50}

\[
\setlength{\fboxsep}{0pt}
\fboxrule=0.8pt
\fbox{%
$%
\begin{array}{c|c|c|c|c}
\cellcolor{gray!15}
&
\PTwoCell{0}{1}{~}{different}
&
\PTwoCell{0}{2}{~}{different}
&
\cdots
&
\PTwoCell{0}{k}{~}{different}
\\ \hline
\PTwoCell{1}{0}{~}{different}
&
\PTwoCell{1}{1}{~}{same}
&
\PTwoCell{1}{2}{~}{different}
&
\cdots
&
\PTwoCell{1}{k}{~}{different}
\\ \hline
\vdots & \vdots & \vdots & \ddots & \vdots
\\ \hline
\PTwoCell{k}{0}{~}{different}
&
\PTwoCell{k}{1}{~}{different}
&
\PTwoCell{k}{2}{~}{different}
&
\cdots
&
\PTwoCell{k}{k}{~}{same}
\end{array}%
$%
}
\]
\endgroup

    \caption{Visualization of the indexing for vertices and edges in the $P_2$-packing contained in $H$ as used for the construction in Lemma~\ref{lem:gadget_construction}.}
    \label{fig:gadget_graphH}
\end{figure}

\paragraph*{Constructing $\Gamma=(G,\tau,c)$ from $H$ and $F$}
We provide an illustration in Figure~\ref{fig:gadgetG}.
For the construction of $G$, we will start from $H$ by iteratively cloning the center vertices $p^2_{h,v}$ in the $P_2$-packing. Specifically, we proceed for all $h,v \in \nzero{k}$ as follows:
\begin{enumerate}
    \item If $h=v>0$, we replace $p^2_{v,v}$ by $n$ fresh vertices $\{p^2_{v,v}(u)\}_{u \in U}$. \\ Each $p^2_{v,v}(u)$ will have the same neighbourhood as $p^2_{v,v}$. We emphasize that this creates a biclique between $\{p^2_{i,i}(u)\}_{u \in U}$ and $\{p^2_{j,j}(u)\}_{u \in U}$ if $p^2_{i,i}$ and $p^2_{j,j}$ are adjacent in $H$. The edges connecting $p^2_{v,v}(u)$ to $p^1_{v,v}$ and $p^3_{v,v}$ are denoted by, respectively, $e^{\ell}_{v,v}(u)$ and $e^{r}_{v,v}(u)$. 
    \item If $h=0$ or $v=0$, we proceed similarly: we replace $p^2_{h,v}$ by $n$ fresh vertices $\{p^2_{h,v}(u)\}_{u \in U}$. Again, the edges connecting $p^2_{h,v}(u)$ to $p^1_{h,v}$ and $p^3_{h,v}$ are denoted by, respectively, $e^{\ell}_{h,v}(u)$ and $e^{r}_{h,v}(u)$.
    Similarly as in (1), this can create bicliques between sets of fresh vertices, but those will not cause a problem.
    \item Otherwise, that is, if $h\neq v$ and $h,v>0$, we replace $p^2_{h,v}$ by $2m$ fresh vertices as follows: For each $e=\{u,u'\}$ of $F$, we add $p^2_{h,v}(u,u')$ and $p^2_{h,v}(u',u)$.\\ 
    The edges connecting $p^2_{h,v}(u,u')$ to $p^1_{h,v}$ and $p^3_{h,v}$ are denoted by, respectively, $e^{\ell}_{h,v}(u)$ and $e^{r}_{h,v}(u')$.\\ 
    The edges connecting $p^2_{h,v}(u',u)$ to $p^1_{h,v}$ and $p^3_{h,v}$ are denoted by, respectively, $e^{\ell}_{h,v}(u')$ and $e^{r}_{h,v}(u)$.\\
    Similarly as in (1), this can create bicliques between sets of fresh vertices, but those will not cause a problem.
\end{enumerate}
This finishes the construction of $G$. Intuitively, the vertices $p_{h,v}(\cdot)$ in $G$ are organised in a grid-like subgraph; we refer again to Figure~\ref{fig:gadgetG} for an illustration. Nodes of type (1) are the diagonal, nodes of type (2) are what we call \emph{anchors} for rows ($p_{h,0}$) and columns ($p_{0,v}$), and nodes of type (3) correspond to the remaining cells.

Before continuing with $\tau$, we make the following easy but crucial observation:
\begin{observation}
The vertices of $G$ can be partitioned as follows:
    \[ V(G) = V(H)\setminus \{p^2_{h,v}\}_{h,v \in [k+1]} ~\dot\cup~ \{p^2_{v,v}(u)\}_{\substack{u \in U\\0<v}}~ \dot\cup~ \{p^2_{h,v}(u)\}_{\substack{u \in U\\v=0 \oplus h = 0}} ~\dot\cup~\{p^2_{h,v}(u,u')\}_{\substack{\{u,u'\}\in E(F)\\h,v>0 \wedge h\neq v}} \]
Moreover, the function $c: V(G) \to V(H)$ defined by
\[ c(x)=\begin{cases}
    p^2_{v,v} & x= p^2_{v,v}(u) \text{  for some } 0<v\in [k+1]\text{, and } u \in U\\
    p^2_{h,v} & x= p^2_{h,v}(u) \text{  for some } h,v\in [k+1] \text{ with $h=0 \oplus v=0$, and } u \in U\\
    p^2_{h,v} & x= p^2_{h,v}(u,u') \text{  for some } h,v\in [k] \text{ with $v \neq h$ and $h,v>0$, and } \{u,u'\} \in E(F)\\
    x & \text{otherwise}
\end{cases} \]
induces a vertex-surjective homomorphism from $G$ to $H$. \qed
\end{observation}
Note that the mapping $c$ merely maps back the cloned vertices to the orginal vertices in $H$, inducing a canonical $H$-colouring $(\nu_c,\xi_c)$ of $G$. We allow ourselves to use $c$ interchangeably for both $\nu_c$ and $\xi_c$.

We will now proceed by defining $\tau$. 
To this end, define the following two functions taking as input $u \in U (=\nzero{n})$ and $h,v \in \nzero{k+1}$.
\begin{align}
    t^{\ell}(u,h,v) &:= (u+1)n^{2(h+1)}+(v+1) \\
    t^r(u,h,v) &:= (u+1) n^{2(k+v+2)} + (h+1)
\end{align} 

\begin{figure}
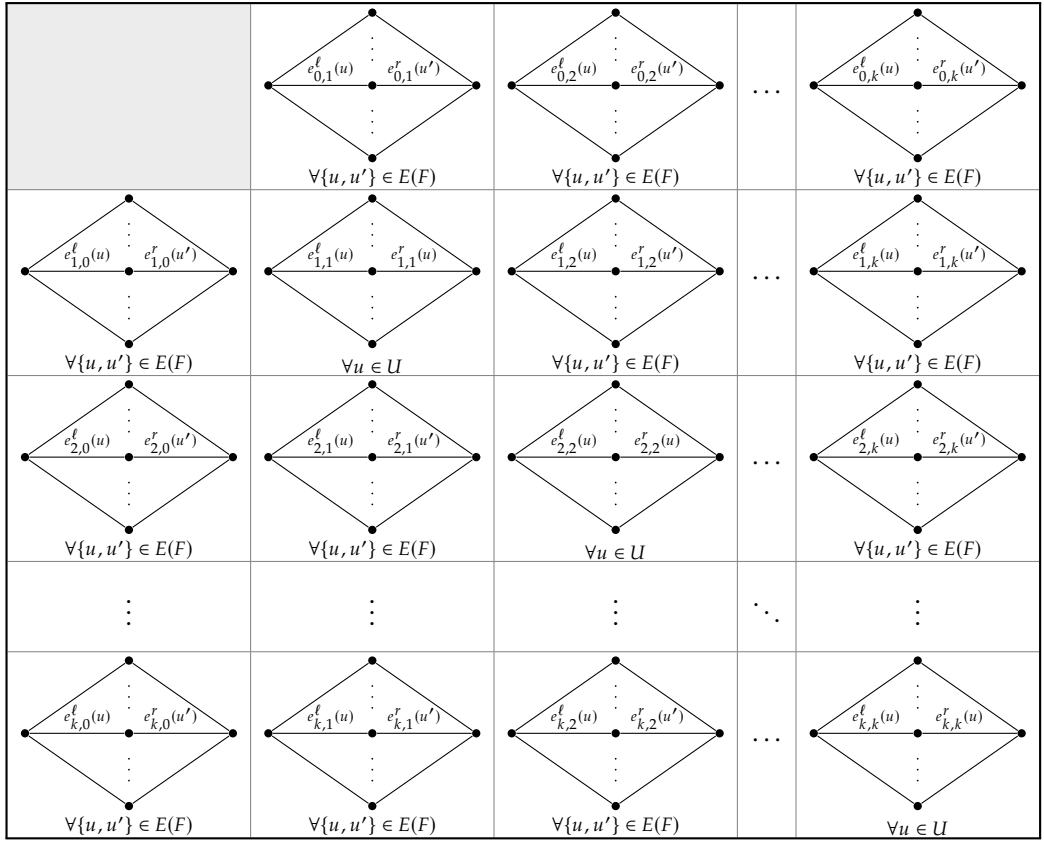

    \centering
    \begingroup
\renewcommand{\arraystretch}{2.5}
\setlength{\arrayrulewidth}{0.25pt}
\arrayrulecolor{black!50}

\[
\setlength{\fboxsep}{0pt}
\fboxrule=0.8pt
\fbox{%
$%
\begin{array}{c|c|c|c|c}
\cellcolor{gray!15}
&
\DiamondCell{0}{1}{\forall\{u,u'\}\in E(F)}{different}
&
\DiamondCell{0}{2}{\forall\{u,u'\}\in E(F)}{different}
&
\cdots
&
\DiamondCell{0}{k}{\forall\{u,u'\}\in E(F)}{different}
\\ \hline
\DiamondCell{1}{0}{\forall\{u,u'\}\in E(F)}{different}
&
\DiamondCell{1}{1}{\forall u\in U}{same}
&
\DiamondCell{1}{2}{\forall\{u,u'\}\in E(F)}{different}
&
\cdots
&
\DiamondCell{1}{k}{\forall\{u,u'\}\in E(F)}{different}
\\ \hline
\DiamondCell{2}{0}{\forall\{u,u'\}\in E(F)}{different}
&
\DiamondCell{2}{1}{\forall\{u,u'\}\in E(F)}{different}
&
\DiamondCell{2}{2}{\forall u\in U}{same}
&
\cdots
&
\DiamondCell{2}{k}{\forall\{u,u'\}\in E(F)}{different}
\\ \hline
\vdots & \vdots & \vdots & \ddots & \vdots
\\ \hline
\DiamondCell{k}{0}{\forall\{u,u'\}\in E(F)}{different}
&
\DiamondCell{k}{1}{\forall\{u,u'\}\in E(F)}{different}
&
\DiamondCell{k}{2}{\forall\{u,u'\}\in E(F)}{different}
&
\cdots
&
\DiamondCell{k}{k}{\forall u\in U}{same}
\end{array}%
$%
}
\]
\endgroup
    \caption{Illustration of the construction of $G$ from $H$ and $F$.}
    \label{fig:gadgetG}
\end{figure}

\begin{lemma}\label{lem:tdef}
The functions $t^\ell$ and $t^r$ have the following properties:
\begin{itemize}
    \item[(a)] For all $u,u' \in U$ and $0<h<k$ we have $ t^\ell(u,h,k)+1 < t^\ell(u',h+1,0)$.
    \item[(b)] For all $u,u' \in U$ and $0<v<k$ we have $t^r(u,k,v)+1 < t^r(u',0,v+1)$.
    \item[(c)] For all $u,u'\in U$ and $h,h',v,v'\in \nzero{k}$, we have $t^{\ell}(u,h,v)+1 <
    t^r(u',h',v')$.
    \end{itemize}
\end{lemma} 

\begin{proof}
We prove all three claims separately.
\begin{enumerate}
    \item[(a)] We have
    \[t^\ell(u,h,k)+1 \leq n\cdot n^{2(h+1)} + k+2 < n^{2h+3} + n^{2h+3} < n^{2h+4} = n^{2((h+1)+1)} \leq t^\ell(u',h+1,0) \,. \]
    \item[(b)] We have
    \[t^r(u,k,v)+1 \leq n\cdot n^{2(k+v+2)} + k+2 < n^{2k+2v+5} + n^{2k+2v+5} < n^{2k+2v+6} = n^{2(k+(v+1)+2)} \leq t^\ell(u',0,v+1) \,. \]
    \item[(c)] We have 
    \[t^\ell(u,h,v)+1 \leq n\cdot n^{2k+2} + (k+2) <  n^{2k+3} + n^{2k+3} <  n^{2k+4} = n^{2(k+2)} \leq t^r(u',h',v')  \]
    \end{enumerate}
\end{proof}

Next, for defining $\tau$, let $R$ denote the set of edges of $H$ not contained in the $P_2$-packing (note that $R$ might contain edges between two vertices of the $P_2$-packing). Set $\kappa = |R|$ and fix any total order $\prec_R$ on $R$; we write $\hat{e}_1,\dots,\hat{e}_\kappa$ for the elements of $R$, ordered by $\prec_R$. We then set $\tau$ as follows:
\begin{definition}[Specification of $\tau$]\label{def:tau_in_hardness_reduction}
We assign the edges of $G$ the following times:
\begin{enumerate}
    \item For all $0<h,v \in \nzero{k+1}$ and $u\in U$, we set $\tau(e^\ell_{h,v}(u))=t^\ell(u,h,v)$ and $\tau(e^r_{h,v}(u))=t^r(u,h,v)$.
    \item For all $0<h \in \nzero{k+1}$ and $u \in U$, we set $\tau(e^\ell_{h,0}(u))=t^\ell(u,h,0)$ and $\tau(e^r_{h,0}(u))= t^\ell(u,h,k)+1$.
    \item For all $0<v \in \nzero{k+1}$ and $u \in U$, we set $\tau(e^r_{0,v}(u))=t^r(u,0,v)$ and $\tau(e^\ell_{0,v}(u))= t^r(u,k,v)+1$.
    \item For any remaining edge $e$ of $G$ we proceed as follows: Observe that $c(e)=\hat{e}_j$ for some $\hat{e}_j\in R$. We set $\tau(e)= t^r(n,k+1,k+1)+2+j$.
\end{enumerate}
\end{definition}

\subsubsection{Defining the total order for $E(H)$}\label{sec:ordering_on_H}
We define the total order $\preccurlyeq$ as follows: For each $0<h \in \nzero{k+1}$ we define the segment
\[\mathsf{row}_h :~e^\ell_{h,0} \prec e^\ell_{h,1} \prec \dots \prec e^\ell_{h,k} \prec e^r_{h,0}\]
For each $0<v \in \nzero{k+1}$ we define the segment
\[ \mathsf{column}_v :~e^r_{0,v} \prec e^r_{1,v} \prec \dots \prec e^r_{k,v} \prec e^\ell_{0,v}\]
Finally, we concatenate the segments (and $R$) as follows:
\[ \prec ~:=~~~~\mathsf{row}_1 \prec \mathsf{row}_2 \prec \dots \prec \mathsf{row}_{k} \prec \mathsf{column}_1 \prec \mathsf{column}_2 \prec \dots \prec \mathsf{column}_k\prec \hat{e}_1 \prec \dots \prec \hat{e}_\kappa \]

\subsubsection{Correctness of the Global Gadget}

\begin{lemma}\label{lem:gadget_correct}
    $F$ has a $k$-clique if and only if there is an edge-colourful homomorphism from $(H,\preccurlyeq)$ to $(G,\tau,c)$.
\end{lemma}
\begin{proof}
    We show both directions separately, and we start with the easier direction:
    \paragraph*{$F$ has a $k$-clique $\Rightarrow$ there is an edge-colourful homomorphism from $(H,\preccurlyeq)$ to $(G,\tau,c)$}
    Let $u_1,\dots,u_k$ be a $k$-clique in $F$. We define $\varphi$ as follows --- note that we only need to specify the mapping from vertices of $H$ to vertices of $G$ as, by construction, the only possible multi-edges are between non-cloned vertices:
    \begin{itemize}
        \item For all $x \in V(H)$ with $x \neq p^2_{\cdot,\cdot}(\cdot)$, we just set $\varphi(x)=x$.
        \item For all $0<v \in \nzero{k+1}$, we set $\varphi(p^2_{v,v}) = p^2_{v,v}(u_v)$.
        \item For all $0<h,v \in \nzero{k+1}$ with $h \neq v$ we set $\varphi(p^2_{h,v})=p^2_{h,v}(u_h,u_v)$.
        \item For all $0<h \in \nzero{k+1}$, we set $\varphi(p^2_{h,0})=p^2_{h,0}(u_h)$.
        \item For all $0<v \in \nzero{k+1}$, we set $\varphi(p^2_{0,v})=p^2_{0,v}(u_v)$.
    \end{itemize}
    Since the $u_i$ form a clique, all vertices $p^2_{h,v}(u_h,u_v)$ for $0<h,v \in \nzero{k+1}$ exist. Moreover, it follows also immediately from the construction of $G$ that $\varphi$ preserves edges since $c(\varphi(x))=x$ for all vertices of $H$ and all of the fresh vertices $p^2_{h,v}(u_h,u_v)$ have been added by cloning the $p^2_{h,v}$. 

    It remains to be shown that $\varphi$ also satisfies the temporal constraints. We will verify first the constraints of the row and column segments separately, and then we will verify their combination.
    Let us start by considering, for some $0<h \in \nzero{k+1}$, the segment
\[\mathsf{row}_h :~e^\ell_{h,0} \prec e^\ell_{h,1} \prec \dots \prec e^\ell_{h,k} \prec e^r_{h,0}\]
We have to verify that
\[\tau(\varphi(e^\ell_{h,0})) < \tau(\varphi( e^\ell_{h,1})) < \dots < \tau(\varphi(e^\ell_{h,k})) < \tau(\varphi(e^r_{h,0}))\,.\]
Applying $\varphi$, this is equivalent to
\[\tau(e^\ell_{h,0}(u_h)) < \tau(e^\ell_{h,1}(u_h,u_1)) < \dots < \tau(e^\ell_{h,k}(u_h,u_k)) < \tau(e^r_{h,0}(u_h)) \,,\]
which, by inserting the definition of $\tau$, is equivalent to
\[t^\ell(u_h,h,0) < t^\ell(u_h,h,1) < \dots < t^\ell(u_h,h,k) < t^\ell(u_h,h,k)+1 \,,\]
which follows immediately from the definition of $t^\ell$.

With a symmetric argument, one can show that the images of the $\mathsf{column}_v$ segments are correctly ordered. Finally, Lemma~\ref{lem:tdef} (a) implies that the time of the image of the last edge of $\mathsf{row}_h$ is less than the time of the image of the first edge of $\mathsf{row}_{h+1}$. The symmetric property holds for $\mathsf{column}$ using Lemma~\ref{lem:tdef} (b). Finally, Lemma~\ref{lem:tdef} (c) implies that time of the image of the last edge in $\mathsf{row}_k$ is less than the time of any edge in any $\mathsf{column}_v$. We are hence able conclude this direction by observing that all edges in $R$ are mapped to edges the times of which are larger than all previous times, and that the ordering of times within those edges coincides with $\prec_R$ (see (4) in the definition of $\tau$).
\paragraph*{There is an edge-colourful homomorphism from $(H,\preccurlyeq)$ to $(G,\tau)$ $\Rightarrow$ $F$ has a $k$-clique}\label{sec:hard_direction}

Let $\varphi$ be an edge-colourful homomorphism from $(H,\preccurlyeq)$ to $(G,\tau,c)$ --- we identify again $\varphi$ as a mapping from $V(H)$ to $V(G)$ as there are no multi-edges that have a cloned vertex as an endpoint, but only edges incident to cloned vertices will be relevant for specifying the $k$-clique in $F$. 

Recall, specifically, that $G$ has been obtained from $H$ by cloning intermediate vertices $p^2_{h,v}$ of the designated $P_2$-packing in $H$.
Recall further that $R \subseteq E(H)$ is the set of edges of $H$ not included in the $P_2$ packing. Moreover, recall also that $R= \hat{e}_1, \dots, \hat{e}_\kappa$, ordered by $\hat{e}_1 \prec \dots \prec \hat{e}_\kappa$; additionally, those $\kappa$ edges are the last ones in the order $\prec$. Now, from (4) in Definition~\ref{def:tau_in_hardness_reduction} it follows that any edge $e \in E(G)$ with $c(e)=\hat{e}_j$ was assigned time $t(e)=t^r(n,k+1,k+1)+2+j$. Moreover, for each $j>0$, we have
\[ t^r(n,k+1,k+1)+2+j > t^r(u,h,v) > t^\ell(u,h,v) \]
for any $u\in \nzero{N}, h,v \in \nzero{k+1}$ by Lemma~\ref{lem:tdef}. In other words \[t^r(n,k+1,k+1)+2+1~,~\dots~,~t^r(n,k+1,k+1)+2+\kappa\] constitute the $\kappa$ highest times of edges in $G$, and for each $1\leq j \leq \kappa$ there is an edge $e\in E(G)$ with $c(e)=\hat{e}_j$ and $t(e)=t^r(n,k+1,k+1)+2+j$. Since $\varphi$ is edge-colourful and since $\varphi$ preserves the ordering $\prec$, it follows that $c(\varphi(e_j))=e_j$ for all $1\leq j \leq \kappa$. Consequently, the edges $\hat{e}$ of $H$ not included in the designated $P_2$ packing must be mapped by $\varphi$ to edges $e$ in $G$ coloured by $\hat{e}$. This enables us to focus on the images of the edges of the $P_2$-packing. 

To this end, it will be convenient to consult again the two illustrations of $H$ and $G$ in Figures~\ref{fig:gadget_graphH} and~\ref{fig:gadgetG}.

Note that the left and right endpoints of the $P_2$-packing in $H$ are not necessarily disjoint. We only know that the intermediate vertices are disjoint. The same is true for the corresponding fragments of $G$. Note also that the depiction of $G$ reflects its $H$-colouring: For any $h,v,u$, the edges $e^\ell_{h,v}(u)$ and $e^r_{h,v}(u)$ are coloured, respectively, by $e^{\ell}_{h,v}$ and $e^r_{h,v}$.

Note that $\varphi$ induces an edge-colourful homomorphism from the above fragment of $H$ to the above fragment of $G$ that satisfies the temporal constraints given by $\prec$. What remains to be shown is that this mapping induces a $k$-clique in $F$.

In the first step, we show that $\varphi$ must in fact map each edge $e^\ell_{h,v}$ / $e^r_{h,v}$ to an edge $e^\ell_{h,v}(u)$ / $e^r_{h,v}(u)$ for some $u\in U$, that is, we claim that $\varphi$ is not only colourful in the sense that it hits each edge-colour of $H$, but it satisfies in fact $c(\varphi(e))=e$. To this end, recall from Section~\ref{sec:ordering_on_H} that the edges of $H$ are ordered, w.r.t. $\prec$, by
\[ \mathsf{row}_1 \prec \dots \prec \mathsf{row}_{k} \prec \mathsf{column}_1 \prec \dots \prec \mathsf{column}_k \,, \text{ where}\]
\[\mathsf{row}_h :~e^\ell_{h,0} \prec e^\ell_{h,1} \prec \dots \prec e^\ell_{h,k} \prec e^r_{h,0}\]
and
\[ \mathsf{column}_v :~e^r_{0,v} \prec e^r_{1,v} \prec \dots \prec e^r_{k,v} \prec e^\ell_{0,v}\]

Let us start with the last edge of the last column: $e^\ell_{0,k}$. Note that all edges of $G$ coloured by $c$ with $e^{\ell}_{0,k}$ are assigned time $t^r(u,k,k)+1$ for some $u$ (see Definition~\ref{def:tau_in_hardness_reduction}; item (3)). Let us assume for contradiction that  $e^{\ell}_{0,k}$ is not mapped to an edge of $G$ with time $t^r(u,k,k)+1$ for some $u$. Then, by Definition~\ref{def:tau_in_hardness_reduction}, the edge $e^{\ell}_{0,k}$ must be mapped to an edge of time $t^r(u,h,v)$ or $t^\ell(u,h,v)$, corresponding to (1) and (2) in Definition~\ref{def:tau_in_hardness_reduction} --- note that edges with times assigned in item (4) in Definition~\ref{def:tau_in_hardness_reduction} are not available anymore since the corresponding colours have been covered by the edges of $H$ outside of the $P_2$-packing. However, as $\varphi$ is edge-colourful, we obtain as a consequence, that there must be an edge $e$ of $H$ with $e \prec e^{\ell}_{0,k}$ but the time of $\varphi(e)$ is $t^r(u,k,k)+1> \max\{t^r(u,h,v),t^\ell(u,h,v)\}$ (see Lemma~\ref{lem:tdef}). This yields a contradiction and thus $e^{\ell}_{0,k}$ is mapped to an edge of $G$ with time $t^r(u,k,k)+1$. The edges with those times are precisely the edges of the form $e^\ell_{0,k}(u)$.
Continuing inductively, using the choice of $\tau$ in Definition~\ref{def:tau_in_hardness_reduction} and the properties of $t^\ell$ and $t^r$ (Lemma~\ref{lem:tdef}), we obtain, as desired that 
$\varphi$ must map each edge $e^\ell_{h,v}$ / $e^r_{h,v}$ to an edge $e^\ell_{h,v}(u)$ / $e^r_{h,v}(u)$ for some $u\in U$. Referring to the illustrations in Figures~\ref{fig:gadget_graphH} and~\ref{fig:gadgetG} of $H$ and $G$, this means that each $2$-path $e^\ell_{h,v},e^r_{h,v}$ of $H$ is mapped to one of the corresponding $2$-paths in the same cell (with index $h$ and $v$) in $G$. 

For each $(h,v)\in \nzero{k+1}^2\setminus\{(0,0)\}$, let $u^\ell_{h,v},u^r_{h,v} \in V(F)$ such that $e^\ell_{h,v}$ and $e^r_{h,v}$ are mapped, respectively, to $e^\ell_{h,v}(u^\ell_{h,v})$ and $e^r_{h,v}(u^r_{h,v})$. 

\begin{claim}\label{clm:hardness_rows_equal}
    For all $0<h \in \nzero{k+1}$ and $v \in \nzero{k+1}$, we have $u^\ell_{h,0}=u^\ell_{h,v}$.
\end{claim}
\begin{claimproof}
From the definition of the total order and the fact that $\varphi$ is a temporal homomorphism, it follows that
\[\tau(e^\ell_{h,0}(u^\ell_{h,0})) <\dots < \tau(e^\ell_{h,v}(u^\ell_{h,v})) < \dots < \tau(e^r_{h,0}(u^r_{h,0})) \,.\]
Now note that, by construction of $G$, we have $u^\ell_{h,0}=u^r_{h,0}$; inserting the definition of $\tau$ we thus obtain:
\[t^\ell(u^\ell_{h,0},h,0) <t^\ell(u^\ell_{h,v},h,v) < t^\ell(u^\ell_{h,0},h,k) +1 \,.\]
Recall that $n=|V(F)|>k$ and set $A=n^{2(h+1)}$; hence note that $A>k$. Inserting the definition of $t^\ell$, we obtain
\[ (u^\ell_{h,0}+1)A + 1 \stackrel{(a)}{<} (u^\ell_{h,v} +1)A + v+1 \stackrel{(b)}{<} (u^\ell_{h,0}+1)A + k+2 \,.\]
Finally, using that $v \leq k$ and $A>k$, we have that (a) implies $u^\ell_{h,0} \leq u^\ell_{h,v}$, and (b) implies $u^\ell_{h,v} \leq u^\ell_{h,0}$.
\end{claimproof}
\begin{claim}\label{clm:hardness_columns_equal}
    For all $0<v \in \nzero{k+1}$ and $h \in \nzero{k+1}$, we have $u^r_{0,v}=u^r_{h,v}$.
\end{claim}
\begin{claimproof}
    The argument is symmetric to the proof of Claim~\ref{clm:hardness_rows_equal}.
\end{claimproof}
Using the previous claims, we are now able to construct a $k$-clique in $F$ from this mapping --- the argument is, at this point, similar to the standard version of reducing colourful homomorphisms from grids to cliques (cf.\ \cite[Claim 2.46]{Roth19}): we claim that for $i \in [k]$, the vertices $u^\ell_{i,i}(=u^r_{i,i})$ form a $k$-clique in $F$. To this end, we show that for all $i<j$, we have that $\{u^\ell_{i,i},u^\ell_{j,j}\}$ is an edge of $F$. By Claim~\ref{clm:hardness_rows_equal}, we have that $u^\ell_{i,i}=u^\ell_{i,j}$, and by Claim~\ref{clm:hardness_columns_equal}, we have that $(u^\ell_{j,j}=)u^r_{j,j}=u^r_{i,j}$. However, by definition of $G$, the pair of edges $e^\ell_{i,j}(u^\ell_{i,j})$ and $e^r_{i,j}(u^r_{i,j})$ can only constitute a $2$-path (which is required for the homomorphism) if $\{u^\ell_{i,j},u^r_{i,j}\}$ is an edge of $F$. This concludes the proof.
\end{proof}

With the correctness of the gadget established, we can effortlessly proceed with the proof of Lemma~\ref{lem:gadget_construction}.
\begin{proof}[Proof of Lemma~\ref{lem:gadget_construction}]
    Given the specified inputs, we construct the total ordering $\preccurlyeq$ and the $H$-coloured temporal graph $(G,\tau,c)$ as in the previous section --- clearly, the construction can be done in polynomial time. Correctness follows from Lemma~\ref{lem:gadget_correct}. 
\end{proof}

\begin{corollary}\label{cor:main_hardness}
    Let $\mathcal{H}$ be a recursively enumerable class of graphs such that the linegraphs of graphs in $\mathcal{H}$ have unbounded semi-induced matching number. Then $\textsc{Clique}\fptred\#\textsc{ColTemporalHom}_\mathrm{TO}(\mathcal{H})$.
\end{corollary}
\begin{proof}
    By Observation~\ref{obs:semi-induced-line-match_P2packing}, the graphs in $\mathcal{H}$ contain $P_2$-packings of unbounded size. Moreover, given $k \in \mathcal{N}$, using that $\mathcal{H}$ is recursively enumerable, we can find in time only depending on $k$ a graph $H\in \mathcal{H}$ together with a $P_2$-packing of size $(k+1)^2-1$ of $H$. The reduction then follows by Lemma~\ref{lem:gadget_construction}. 
\end{proof}

\subsubsection{Proof of the intractability result}
First, using a standard inclusion-exclusion argument, we show the following:
\begin{lemma}\label{lem:coltemhom_totemphom}
    For all $\mathcal{H}$, we have $\#\textsc{ColTemporalHom}_\mathrm{TO}(\mathcal{H})\fptred \#\textsc{TemporalHom}_\mathrm{TO}(\mathcal{H})$.
\end{lemma}
\begin{proof}
    Let $(H,\preccurlyeq)$ and $\Gamma=(G,\tau,c)$ be the input to $\#\textsc{ColTemporalHom}_\mathrm{TO}(\mathcal{H})$. For each subset $J \subseteq E(H)$, we write $\Gamma^c_J$ for the (non-coloured) temporal graph obtained from $G$ by deleting all edges coloured with an element in $J$.

    Using the inclusion-exclusion formula, we observe
    \[\#\colhom{(H,\preccurlyeq)}{(G,\tau,c)} = \sum_{J \subseteq E(H)} (-1)^J\cdot \#\homs{(H,\preccurlyeq)}{\Gamma^c_J} \,.\]
    Thus, we can compute $\#\colhom{(H,\preccurlyeq)}{(G,\tau,c)}$ via $2^{|E(H)|}$ oracle calls to $\#\textsc{TemporalHom}_\mathrm{TO}(\mathcal{H})$.
\end{proof}

We are now able to conclude our hardness reduction.
\begin{lemma}\label{lemma:main_hardness}
    Let $\mathcal{H}$ be a recursively enumerable class of graphs. If the line graphs of $\mathcal{H}$ have unbounded semi-induced matching number, then $\#\textsc{TemporalHom}_\mathrm{TO}(\mathcal{H})$ is $\mathrm{W}[1]$-hard under parameterised Turing-reductions.
\end{lemma}
\begin{proof}
    The reduction chain from $\textsc{Clique}$ is given by Corollary~\ref{cor:main_hardness} and Lemma~\ref{lem:coltemhom_totemphom}.
\end{proof}

\begin{proof}[Proof of Theorem~\ref{thm:intro_main_classification}]
    The upper bound follows from Lemma~\ref{lem:sim-cw-bound} and Theorem~\ref{thm:intro_main_algo}. The lower bound follows from Lemma~\ref{lemma:main_hardness}.
\end{proof}

\section{Conclusion and Future Work}
We have provided the first comprehensive treatment of homomorphism counting in temporal graphs, both with respect to descriptive complexity and expressive power in terms of a Lov{\'{a}}sz-style isomorphism theorem, as well as with respect to algorithmic complexity in terms of an FPT dynamic programming algorithm along the newly introduced toadwith of temporal patterns and the explicit complexity dichotomy for counting homomorphisms from totally ordered temporal patterns.

While our results fully resolve the natural special case of total orders, subsuming e.g.\ an FPT algorithm for counting temporal walks, the primary question left open for future work is whether the algorithm along toadwidth is optimal even for not necessarily totally ordered patterns. In other words, and more formally:
\begin{center}
    \textit{Is} $\#\textsc{TemporalHom}(\mathcal{C})$ \textit{intractable for each $\mathcal{C}$ of unbounded toadwidth?}
\end{center}
We suspect the answer to the above question to be affirmative.

Moreover, we propose further work on isomorphisms and temporal homomorphism indistinguishability for temporal graphs: as illustrated in the introduction, there are various well-motivated notions for isomorphisms of temporal graphs in the literature, and we believe it to be worthwhile to investigate whether our proof for order-isomorphisms can be adapted for point-wise and time-wise isomorphisms via modifying the definition of temporal homomorphisms.

Finally, with homomorphism counts being the basis for motif counting in static graphs~\cite{CurticapeanDM17}, we believe that temporal homomorphisms have the potential to provide a similar unifying framework for pattern counting in temporal graphs. In fact, our proof of Theorem~\ref{thm:lovasz_improved_intro} already makes use of a transformation of temporal subgraph isomorphisms to finite linear combinations of temporal homomorphisms, and we are optimistic that similar transformations will allow us to relate the complexity of counting more complex temporal motifs to the upper and lower bounds for counting temporal homomorphisms.

\bibliographystyle{plain}
\bibliography{references}

\begin{thebibliography}{10}

\bibitem{BeddarWiesingetal24}
Silvia Beddar{-}Wiesing, Giuseppe~Alessio D'Inverno, Caterina Graziani, Veronica Lachi, Alice Moallemy{-}Oureh, Franco Scarselli, and Josephine~Maria Thomas.
\newblock Weisfeiler-lehman goes dynamic: An analysis of the expressive power of graph neural networks for attributed and dynamic graphs.
\newblock {\em Neural Networks}, 173:106213, 2024.

\bibitem{BentertHNN20}
Matthias Bentert, Anne{-}Sophie Himmel, Andr{\'{e}} Nichterlein, and Rolf Niedermeier.
\newblock Efficient computation of optimal temporal walks under waiting-time constraints.
\newblock {\em Appl. Netw. Sci.}, 5(1):73, 2020.

\bibitem{DBLP:journals/tcs/BergougnouxK19}
Benjamin Bergougnoux and Mamadou~Moustapha Kant{\'{e}}.
\newblock Fast exact algorithms for some connectivity problems parameterized by clique-width.
\newblock {\em Theor. Comput. Sci.}, 782:30--53, 2019.

\bibitem{CasteigtsFQS12}
Arnaud Casteigts, Paola Flocchini, Walter Quattrociocchi, and Nicola Santoro.
\newblock Time-varying graphs and dynamic networks.
\newblock {\em Int. J. Parallel Emergent Distributed Syst.}, 27(5):387--408, 2012.

\bibitem{CasteigtsHMZ21}
Arnaud Casteigts, Anne{-}Sophie Himmel, Hendrik Molter, and Philipp Zschoche.
\newblock Finding temporal paths under waiting time constraints.
\newblock {\em Algorithmica}, 83(9):2754--2802, 2021.

\bibitem{Chenetal05}
Jianer Chen, Benny Chor, Mike Fellows, Xiuzhen Huang, David~W. Juedes, Iyad~A. Kanj, and Ge~Xia.
\newblock Tight lower bounds for certain parameterized {N}{P}-hard problems.
\newblock {\em Inf. Comput.}, 201(2):216--231, 2005.

\bibitem{Chenetal06}
Jianer Chen, Xiuzhen Huang, Iyad~A. Kanj, and Ge~Xia.
\newblock Strong computational lower bounds via parameterized complexity.
\newblock {\em J. Comput. Syst. Sci.}, 72(8):1346--1367, 2006.

\bibitem{CourcelleE12}
Bruno Courcelle and Joost Engelfriet.
\newblock {\em Graph Structure and Monadic Second-Order Logic - {A} Language-Theoretic Approach}, volume 138 of {\em Encyclopedia of mathematics and its applications}.
\newblock Cambridge University Press, 2012.

\bibitem{CourcelleO00}
Bruno Courcelle and Stephan Olariu.
\newblock Upper bounds to the clique width of graphs.
\newblock {\em Discret. Appl. Math.}, 101(1-3):77--114, 2000.

\bibitem{CurticapeanDM17}
Radu Curticapean, Holger Dell, and D{\'{a}}niel Marx.
\newblock Homomorphisms are a good basis for counting small subgraphs.
\newblock In Hamed Hatami, Pierre McKenzie, and Valerie King, editors, {\em Proceedings of the 49th Annual {ACM} {SIGACT} Symposium on Theory of Computing, {STOC} 2017, Montreal, QC, Canada, June 19-23, 2017}, pages 210--223. {ACM}, 2017.

\bibitem{CurticapeanM14}
Radu Curticapean and D{\'{a}}niel Marx.
\newblock Complexity of {C}ounting {S}ubgraphs: {O}nly the {B}oundedness of the {V}ertex-{C}over {N}umber {C}ounts.
\newblock In {\em Proc.\ of IEEE FOCS}, pages 130--139, 2014.

\bibitem{DabrowskiJP19}
Konrad~K. Dabrowski, Matthew Johnson, and Dani{\"{e}}l Paulusma.
\newblock Clique-width for hereditary graph classes.
\newblock In Allan Lo, Richard Mycroft, Guillem Perarnau, and Andrew Treglown, editors, {\em Surveys in Combinatorics, 2019: Invited lectures from the 27th British Combinatorial Conference, Birmingham, UK, July 29 - August 2, 2019}, pages 1--56. Cambridge University Press, 2019.

\bibitem{DalmauJ04}
V{\'{\i}}ctor Dalmau and Peter Jonsson.
\newblock The complexity of counting homomorphisms seen from the other side.
\newblock {\em Theor. Comput. Sci.}, 329(1-3):315--323, 2004.

\bibitem{DellGR18}
Holger Dell, Martin Grohe, and Gaurav Rattan.
\newblock Lov{\'{a}}sz meets weisfeiler and leman.
\newblock In Ioannis Chatzigiannakis, Christos Kaklamanis, D{\'{a}}niel Marx, and Donald Sannella, editors, {\em 45th International Colloquium on Automata, Languages, and Programming, {ICALP} 2018, Prague, Czech Republic, July 9-13, 2018}, volume 107 of {\em LIPIcs}, pages 40:1--40:14. Schloss Dagstuhl - Leibniz-Zentrum f{\"{u}}r Informatik, 2018.

\bibitem{DoringMW24}
Simon D{\"{o}}ring, D{\'{a}}niel Marx, and Philip Wellnitz.
\newblock Counting small induced subgraphs with edge-monotone properties.
\newblock In Bojan Mohar, Igor Shinkar, and Ryan O'Donnell, editors, {\em Proceedings of the 56th Annual {ACM} Symposium on Theory of Computing, {STOC} 2024, Vancouver, BC, Canada, June 24-28, 2024}, pages 1517--1525. {ACM}, 2024.

\bibitem{Dvorak10}
Zdenek Dvor{\'{a}}k.
\newblock On recognizing graphs by numbers of homomorphisms.
\newblock {\em J. Graph Theory}, 64(4):330--342, 2010.

\bibitem{EnrightMM25}
Jessica~A. Enright, Kitty Meeks, and Hendrik Molter.
\newblock Counting temporal paths.
\newblock {\em Algorithmica}, 87(5):736--782, 2025.

\bibitem{FlumG06}
J{\"{o}}rg Flum and Martin Grohe.
\newblock {\em Parameterized Complexity Theory}.
\newblock Texts in Theoretical Computer Science. An {EATCS} Series. Springer, 2006.

\bibitem{TempTreewidth}
Till Fluschnik, Hendrik Molter, Rolf Niedermeier, Malte Renken, and Philipp Zschoche.
\newblock As time goes by: Reflections on treewidth for temporal graphs.
\newblock In Fedor~V. Fomin, Stefan Kratsch, and Erik~Jan van Leeuwen, editors, {\em Treewidth, Kernels, and Algorithms - Essays Dedicated to Hans L. Bodlaender on the Occasion of His 60th Birthday}, volume 12160 of {\em Lecture Notes in Computer Science}, pages 49--77. Springer, 2020.

\bibitem{FockeR24}
Jacob Focke and Marc Roth.
\newblock Counting small induced subgraphs with hereditary properties.
\newblock {\em {SIAM} J. Comput.}, 53(2):189--220, 2024.

\bibitem{GaoR22}
Jianfei Gao and Bruno Ribeiro.
\newblock On the equivalence between temporal and static equivariant graph representations.
\newblock In Kamalika Chaudhuri, Stefanie Jegelka, Le~Song, Csaba Szepesv{\'{a}}ri, Gang Niu, and Sivan Sabato, editors, {\em International Conference on Machine Learning, {ICML} 2022, 17-23 July 2022, Baltimore, Maryland, {USA}}, volume 162 of {\em Proceedings of Machine Learning Research}, pages 7052--7076. {PMLR}, 2022.

\bibitem{GurskiW07}
Frank Gurski and Egon Wanke.
\newblock Line graphs of bounded clique-width.
\newblock {\em Discret. Math.}, 307(22):2734--2754, 2007.

\bibitem{Hanetal04}
Jing-Dong~J. Han, Nicolas Bertin, Tong Hao, Debra~S. Goldberg, Gabriel~F. Berriz, Lan~V. Zhang, Denis Dupuy, Albertha J.~M. Walhout, Michael~E. Cusick, Frederick~P. Roth, and Marc Vidal.
\newblock Evidence for dynamically organized modularity in the yeast protein--protein interaction network.
\newblock {\em Nature}, 430:88--93, 2004.

\bibitem{Heegetal25}
Franziska Heeg, Jonas Sauer, Petra Mutzel, and Ingo Scholtes.
\newblock Weisfeiler and leman follow the arrow of time: Expressive power of message passing in temporal event graphs.
\newblock {\em CoRR}, abs/2505.24438, 2025.

\bibitem{DBLP:conf/esa/HegerfeldK23}
Falko Hegerfeld and Stefan Kratsch.
\newblock Tight algorithms for connectivity problems parameterized by clique-width.
\newblock In Inge~Li G{\o}rtz, Martin Farach{-}Colton, Simon~J. Puglisi, and Grzegorz Herman, editors, {\em 31st Annual European Symposium on Algorithms, {ESA} 2023, September 4-6, 2023, Amsterdam, The Netherlands}, volume 274 of {\em LIPIcs}, pages 59:1--59:19. Schloss Dagstuhl - Leibniz-Zentrum f{\"{u}}r Informatik, 2023.

\bibitem{HolmeS12}
Petter Holme and Jari Saramäki.
\newblock Temporal networks.
\newblock {\em Physics Reports}, 519(3):97--125, 2012.
\newblock Temporal Networks.

\bibitem{ETH}
Russell Impagliazzo and Ramamohan Paturi.
\newblock On the complexity of k-{SAT}.
\newblock {\em J. Comput. Syst. Sci.}, 62(2):367--375, 2001.

\bibitem{KempeKK02}
David Kempe, Jon~M. Kleinberg, and Amit Kumar.
\newblock Connectivity and inference problems for temporal networks.
\newblock {\em J. Comput. Syst. Sci.}, 64(4):820--842, 2002.

\bibitem{Tempcount4}
Paul Liu, Austin~R. Benson, and Moses Charikar.
\newblock Sampling methods for counting temporal motifs.
\newblock In J.~Shane Culpepper, Alistair Moffat, Paul~N. Bennett, and Kristina Lerman, editors, {\em Proceedings of the Twelfth {ACM} International Conference on Web Search and Data Mining, {WSDM} 2019, Melbourne, VIC, Australia, February 11-15, 2019}, pages 294--302. {ACM}, 2019.

\bibitem{Tempcount2}
Penghang Liu, Valerio Guarrasi, and Ahmet~Erdem Sar{\i}y{\"u}ce.
\newblock Temporal network motifs: Models, limitations, evaluation.
\newblock {\em IEEE Transactions on Knowledge and Data Engineering}, 35(1):945--957, 2021.

\bibitem{Longaetal23}
Antonio Longa, Veronica Lachi, Gabriele Santin, Monica Bianchini, Bruno Lepri, Pietro Lio, Franco Scarselli, and Andrea Passerini.
\newblock Graph neural networks for temporal graphs: State of the art, open challenges, and opportunities.
\newblock {\em Trans. Mach. Learn. Res.}, 2023, 2023.

\bibitem{Lovasz12}
L{\'{a}}szl{\'{o}} Lov{\'{a}}sz.
\newblock {\em Large Networks and Graph Limits}, volume~60 of {\em Colloquium Publications}.
\newblock American Mathematical Society, 2012.

\bibitem{Marx10}
D{\'{a}}niel Marx.
\newblock Can you beat treewidth?
\newblock {\em Theory Comput.}, 6(1):85--112, 2010.

\bibitem{Miloetal02}
R.~Milo, S.~Shen-Orr, S.~Itzkovitz, N.~Kashtan, D.~Chklovskii, and U.~Alon.
\newblock Network {Motifs}: {Simple} {Building} {Blocks} of {Complex} {Networks}.
\newblock {\em Science}, 298(5594):824--827, 2002.
\newblock \_eprint: https://www.science.org/doi/pdf/10.1126/science.298.5594.824.

\bibitem{Morrisetal19}
Christopher Morris, Martin Ritzert, Matthias Fey, William~L. Hamilton, Jan~Eric Lenssen, Gaurav Rattan, and Martin Grohe.
\newblock Weisfeiler and leman go neural: Higher-order graph neural networks.
\newblock In {\em The Thirty-Third {AAAI} Conference on Artificial Intelligence, {AAAI} 2019, The Thirty-First Innovative Applications of Artificial Intelligence Conference, {IAAI} 2019, The Ninth {AAAI} Symposium on Educational Advances in Artificial Intelligence, {EAAI} 2019, Honolulu, Hawaii, USA, January 27 - February 1, 2019}, pages 4602--4609. {AAAI} Press, 2019.

\bibitem{Tempcount1}
Ashwin Paranjape, Austin~R. Benson, and Jure Leskovec.
\newblock Motifs in temporal networks.
\newblock In Maarten de~Rijke, Milad Shokouhi, Andrew Tomkins, and Min Zhang, editors, {\em Proceedings of the Tenth {ACM} International Conference on Web Search and Data Mining, {WSDM} 2017, Cambridge, United Kingdom, February 6-10, 2017}, pages 601--610. {ACM}, 2017.

\bibitem{Roth19}
Marc Roth.
\newblock {\em Counting problems on quantum graphs}.
\newblock PhD thesis, Saarland University, Germany, 2019.

\bibitem{Roth26}
Marc Roth.
\newblock Parameterised counting complexity theory.
\newblock {\em Comput. Sci. Rev.}, 59:100837, 2026.

\bibitem{Tempcount3}
Ahmet~Erdem Sar{\i}y{\"u}ce.
\newblock A powerful lens for temporal network analysis: temporal motifs.
\newblock {\em Discover Data}, 3(1):14, 2025.

\bibitem{WalegaR25}
Przemyslaw~Andrzej Walega and Michael Rawson.
\newblock Expressive power of temporal message passing.
\newblock In Toby Walsh, Julie Shah, and Zico Kolter, editors, {\em Thirty-Ninth {AAAI} Conference on Artificial Intelligence, Thirty-Seventh Conference on Innovative Applications of Artificial Intelligence, Fifteenth Symposium on Educational Advances in Artificial Intelligence, {AAAI} 2025, Philadelphia, PA, USA, February 25 - March 4, 2025}, pages 21000--21008. {AAAI} Press, 2025.

\bibitem{ZhaoTHOJL10}
Qiankun Zhao, Yuan Tian, Qi~He, Nuria Oliver, Ruoming Jin, and Wang{-}Chien Lee.
\newblock Communication motifs: a tool to characterize social communications.
\newblock In Jimmy~X. Huang, Nick Koudas, Gareth J.~F. Jones, Xindong Wu, Kevyn Collins{-}Thompson, and Aijun An, editors, {\em Proceedings of the 19th {ACM} Conference on Information and Knowledge Management, {CIKM} 2010, Toronto, Ontario, Canada, October 26-30, 2010}, pages 1645--1648. {ACM}, 2010.

\end{thebibliography}
\end{document}